\documentclass[11pt,letterpaper]{article}
\usepackage{tabu}
\usepackage{multirow}
\usepackage[utf8]{inputenc}
\usepackage{etex}
\usepackage{xspace,enumerate}
\usepackage{subcaption}
\usepackage[dvipsnames]{xcolor}
\usepackage[T1]{fontenc}
\usepackage[full]{textcomp}
\usepackage[american]{babel}
\usepackage{mathtools}
\usepackage{amsthm}
\newtheorem{theorem}{Theorem}[section]
\newtheorem*{theorem*}{Theorem}

\newtheorem*{proposition*}{Proposition}
\newtheorem{lemma}[theorem]{Lemma}
\newtheorem*{lemma*}{Lemma}
\newtheorem{corollary}[theorem]{Corollary}
\newtheorem*{conjecture*}{Conjecture}
\newtheorem{fact}[theorem]{Fact}
\newtheorem*{fact*}{Fact}

\newtheorem*{hypothesis*}{Hypothesis}
\newtheorem{conjecture}[theorem]{Conjecture}
\theoremstyle{definition}
\newtheorem{definition}[theorem]{Definition}
\newtheorem*{definition*}{Definition}

\newtheorem{algorithm}[theorem]{Algorithm}

\theoremstyle{remark}

\newtheorem*{claim*}{Claim}
\newtheorem{remark}[theorem]{Remark}
\newtheorem*{remark*}{Remark}

\newtheorem*{observation*}{Observation}

\usepackage[
letterpaper,
top=1.2in,
bottom=1.2in,
left=1in,
right=1in]{geometry}
\usepackage{newpxtext} % T1, lining figures in math, osf in text
\usepackage{textcomp} % required for special glyphs
\usepackage[varg,bigdelims]{newpxmath}
\usepackage[scr=rsfso]{mathalfa}% \mathscr is fancier than \mathcal
\usepackage{bm} % load after all math to give access to bold math
\let\mathbb\varmathbb
\usepackage{microtype}
\usepackage[
 colorlinks=true,
 urlcolor=blue,
 linkcolor=red,
 citecolor=OliveGreen,
 ]{hyperref}
\usepackage[capitalise,nameinlink]{cleveref}
\crefname{lemma}{Lemma}{Lemmas}
\crefname{fact}{Fact}{Facts}
\crefname{theorem}{Theorem}{Theorems}
\crefname{corollary}{Corollary}{Corollaries}
\crefname{claim}{Claim}{Claims}
\crefname{example}{Example}{Examples}
\crefname{algorithm}{Algorithm}{Algorithms}
\crefname{problem}{Problem}{Problems}
\crefname{definition}{Definition}{Definitions}
\usepackage{paralist}
\usepackage{turnstile}
\usepackage{tikz}
\usepackage{caption}
\DeclareCaptionType{Algorithm}
\usepackage{newfloat}
\usepackage{pst-node}

\usepackage{url}

\usepackage{comment}
\def\AA{\mathbf{A}}
\def\BB{\mathbf{B}}

\def\MM{\mathbf{M}}
\def\HH{\mathbf{H}}

\def\WW{\mathbf{W}}
\def\CC{\mathbf{C}}

\def\II{\mathbb{I}}
\def\YY{\mathbf{Y}}
\def\XX{\mathbf{X}}
\def\KK{\mathbf{K}}
\def\PP{\mathbf{P}}

\def\UU{\mathbf{U}}

\def\VV{\mathbf{V}}

\def\SS{\mathbf{S}}

\def\WW{\mathbf{W}}

\def\QQ{\mathbf{Q}}
\def\GG{\mathbf{G}}
\def\MM{\mathbf{M}}

\def\ZZ{\mathbf{Z}}
\def\RR{\mathbf{R}}
\def\II{\mathbf{I}}
\def\YY{\mathbf{Y}}

\def\calS{\mathcal{S}}

\def\Sig{\mathbf{\Sigma}}

\def\poly{\textrm{poly}}
\def\rank{\textrm{rank}}

\def\trace#1{\mathrm{Tr} \left(#1 \right)}

\newcommand{\Paren}[1]{\left(#1\right)}

\def\abs#1{\left|#1  \right|}

\def\norm#1{\left\| #1 \right\|}

\newcommand{\bigO}[1]{\mathcal{O}\hspace{-0.1cm}\left(#1\right)}

\newcommand{\iprod}[1]{\left\langle #1\right\rangle}

\newcommand{\abnote}[1]{}
\newcommand{\kcnote}[1]{}
\newcommand{\dwnote}[1]{}

\numberwithin{equation}{section}
\numberwithin{figure}{section}

\usepackage[framemethod=TikZ]{mdframed}
\newcounter{Frame}

\title{Low-Rank Approximation with $1/\epsilon^{1/3}$ Matrix-Vector Products}
\author{ Ainesh Bakshi\\
abakshi@cs.cmu.edu\\
CMU
\and
Kenneth L. Clarkson\\
klclarks@us.ibm.com\\
IBM
\and 
David P. Woodruff\\
dwoodruf@cs.cmu.edu\\
CMU
  }
\date{}

\DeclareMathOperator{\nnz}{\mathtt{nnz}}

\newcommand{\eps}{\varepsilon}

\begin{document}

\maketitle

 \thispagestyle{empty}

\begin{abstract}
We study iterative methods based on Krylov subspaces for low-rank approximation under any Schatten-$p$ norm. Here, given access to a matrix $\AA$ through matrix-vector products, an accuracy parameter $\epsilon$, and a target rank $k$, the goal is to find a rank-$k$ matrix $\ZZ$ with orthonormal columns such that $\norm{ \AA \Paren{\II - \ZZ\ZZ^\top} }_{\calS_p} \leq (1+\epsilon)\min_{\UU^\top \UU = \II_k }\norm{\AA\Paren{\II - \UU\UU^\top}}_{\calS_p}$, where $\| \MM \|_{\calS_p}$ denotes the $\ell_p$ norm of the the singular values of~$\MM$. For the special cases of $p=2$ (Frobenius norm) and $p = \infty$ (Spectral norm), Musco and Musco (NeurIPS 2015) obtained an algorithm based on Krylov methods that uses $\tilde{O}(k/\sqrt{\epsilon})$ matrix-vector products, improving on the na\"ive $\tilde{O}(k/\epsilon)$ dependence obtainable by the power method, where $\tilde{O}(\cdot)$ suppresses poly$(\log(dk/\epsilon))$ factors. 

Our main result is an algorithm that uses only $\tilde{O}(kp^{1/6}/\epsilon^{1/3})$ matrix-vector products, and
works for \emph{all}, not necessarily constant,  $p \geq 1$. For $p = 2$ our bound improves the previous $\tilde{O}(k/\epsilon^{1/2})$ bound to $\tilde{O}(k/\epsilon^{1/3})$. Since the Schatten-$p$ and Schatten-$\infty$ norms of any matrix are the same up to a $1+ \epsilon$ factor when $p \geq (\log d)/\epsilon$, our bound recovers the result of Musco and Musco for $p = \infty$. 
Further, we prove a matrix-vector query lower bound of $\Omega(1/\epsilon^{1/3})$ for \emph{any} fixed constant $p \geq 1$, showing that surprisingly $\tilde{\Theta}(1/\epsilon^{1/3})$ is the optimal complexity for constant~$k$. 

To obtain our results, we introduce several new techniques, including optimizing over \emph{multiple Krylov subspaces} simultaneously, and \emph{pinching inequalities} for partitioned operators. Our lower bound for $p \in [1,2]$  uses the \textit{Araki-Lieb-Thirring} trace inequality, whereas for $p>2$, we appeal to a \textit{norm-compression} inequality for \textit{aligned partitioned operators}. As our algorithms only require matrix-vector product access, they can be applied in settings where alternative techniques such as sketching cannot, e.g., to covariance matrices, Hessians defined implicitly by a neural network, and arbitrary polynomials of a matrix.

\end{abstract}

\clearpage
\setcounter{page}{1}

\section{Introduction}
Iterative methods, and in particular Krylov subspace methods, are ubiquitous in scientific computing. Algorithms such as power iteration, Golub-Kahan Bidiagonalization, Arnoldi iteration, and the Lanczos iteration, are used in basic subroutines for matrix inversion, solving linear systems, linear programming, low-rank approximation, and numerous other fundamental linear algebra primitives \cite{saad1981krylov,liesen2013krylov}. %They are also used for computing the matrix exponential \cite{}, for numerically solving large-scale homogeneous linear differential equations ~\cite{}, for computing hyperbolic functions for non-standard alternating random walks on graphs~\cite{}, approximate matrix inversion \cite{}, and regularizing ill-conditioned problems~\cite{}. 
A common technique in the analysis of Krylov methods is the use of Chebyshev polynomials, which can be applied to the singular values of a matrix to implement an approximate interval or step function \cite{mason2002chebyshev,rivlin2020chebyshev}. Further, Chebyshev polynomials reduce the degree required to accurately approximate such functions, leading to significantly fewer iterations and faster running time. In this paper we investigate the power of Krylov methods for low-rank approximation in the matrix-vector product model. 

\vspace{0.1in}
\textbf{The Matrix-Vector Product Model.}  In this model, there is an underlying matrix $\AA$, which is often implicit, and for
which the only access to $\AA$ is via matrix-vector products. Namely, the algorithm chooses
a query vector $v^1$, obtains the product $\AA \cdot v^1$, chooses the next query
vector $v^2$, which is any randomized function of $v^1$ and $\AA \cdot v^1$, then receives $\AA \cdot v^2$,
and so on. If $\AA$ is a non-symmetric matrix, we assume access to products of the form $\AA^\top v$ as well. We refer to the minimal number $q$ of queries needed by the algorithm to solve a problem with constant probability as the {\it query complexity}. We note that upper bounds on the query complexity immediately translate to running time bounds for the RAM model, when $\AA$ is explicit, since a matrix-vector product can be implemented in $\nnz(\AA)$ time, i.e., the number of non-zero entries in the matrix. Since this model captures a large family of iterative methods, it is natural to ask whether Krylov subspace based methods yield optimal algorithms, where the complexity measure of interest is the number of matrix-vector products.
% A natural question is whether the large family of iterative methods based on Krylov subspaces yield optimal algorithms, where the complexity measure of interest is the number of matrix-vector products.

This model 
and related vector-matrix-vector query models were formalized for a number of problems 
in \cite{SunWYZ19,RWZ20}, though the model is standard for measuring efficiency in scientific computing and numerical linear algebra, see, e.g., \cite{bai1996some};
in that literature, methods that use only matrix-vector products are called \emph{matrix-free}.
Subsequently, for the problem of estimating the top eigenvector, nearly tight bounds were obtained
in \cite{SAR18,braverman2020gradient}. Also, for the problem of estimating the trace of a positive semidefinite matrix, tight bounds were obtained in \cite{MMMW21} (see, also \cite{WWZ14}, where tight bounds
were shown in the related vector-matrix-vector query model). For recovering a planted clique
from a random graph, upper and lower bounds were obtained in \cite{woodruff21}. 
In the non-adaptive setting, where
$v^1, \ldots, v^q$, are chosen before making any queries to $\AA$, this is equivalent to the {\it sketching model}, which is thoroughly studied on its own (see, e.g., \cite{nelson2011sketching,woodruff2014sketching}), and in the context of data streams \cite{M05,LNW14}. 

\vspace{0.1in}
\textbf{Why is the matrix $\AA$ implicit?} A small query complexity $q$ leads to an algorithm running in time $\bigO{ T(\AA) \cdot q + P(n, d, q)}$,
where $T(\AA)$ is the time to multiply the $n \times d$ matrix $\AA$ by an arbitrary vector, and $P(n,d,q)$ is the time
needed to form the queries and process the query responses, which is typically small. When the matrix
$\AA$ is given as a list of $\nnz(\AA)$ non-zero entries,
% ($\nnz(\AA)$ denotes the number of non-zero entries of $\AA$),
then $T(\AA) \leq \nnz(\AA)$. However, 
in many problems $\AA$ is not given explicitly, and it is too expensive to write $\AA$ down. 
Indeed, one may be given $\AA$ but want to compute a low-rank approximation to the ``covariance'' (Gram) matrix $\AA^\top\AA$, and computing $\AA^\top\AA$ is too slow \cite{musco2017input}.
%\kcnote{If $\AA$ is the input (the matrix for which matrix-vector products will be computed, probably better to say that the input input is $\AA=\BB^\top\BB$. Also: do you prefer $\AA^\top$ to $\AA^\top$?}
More generally, one may be given $\AA = \UU \Sig \VV^\top$ and a function $f:\mathbb{R} \to \mathbb{R}$, and want to compute matrix-vector products with the generalized matrix function $f(\AA) = \UU f(\Sig) \VV^\top$, where $\UU$ has orthonormal columns, $\VV^\top$ has orthonormal rows, $\Sig$ is a diagonal matrix, and $f$ is applied entry-wise to each entry on the diagonal. 

The covariance matrix corresponds to $f(x) = x^2$, and other common functions $f$ include the matrix exponential $f(x) = e^x$ and low-degree polynomials. For instance, when $\AA$ is the adjacency matrix
of an undirected graph, $f(x) = x^3/6$ is used to count the number of triangles \cite{tsourakakis2008fast,avron2010counting}.
% \kcnote{I don't think the eigenvalues are used in that paper, although matrix-vector products are; the prior paper by Tsourakakis does.} 
Yet another example is when $\AA$ is the Hessian $\HH$ of a neural network with a huge number of parameters, for which it is often impossible to compute or store the entire Hessian~\cite{gkx2019investigation}. Typically $\HH \cdot v$, for any chosen vector $v$, is computed using Pearlmutter's trick~\cite{pearlmutter1994hv_trick}. However,  even with Pearlmutter's trick and distributed computation on modern GPUs, it takes 20 hours to compute the eigendensity of a single Hessian $\HH$ with respect to the cross-entropy loss on the CIFAR-10 dataset from a set of fixed weights for ResNet-18 ~\cite{krizhevsky2009learning}, which has approximately 11 million parameters~\cite{he2016deep,gkx2019investigation}. This time is directly
proportional to the number of matrix-vector products, and therefore minimizing this quantity is crucial. 

\vspace{0.1in}
\textbf{Algorithms and Lower Bounds for Low-Rank Approximation.} 
The low-rank approximation problem is well studied in numerical linear algebra,
with countless applications to clustering, data mining, principal component analysis, 
recommendation systems, and many more. (For
surveys on low-rank approximation, see the monographs \cite{KV09,M11,woodruff2014sketching} and references therein.) In this problem, given an implicit 
$n \times d$ matrix $\AA$, the goal is to output a matrix
$\ZZ \in \mathbb{R}^{d \times k}$  with orthonormal columns such that
\begin{eqnarray}\label{eqn:lowrank}
    \norm{\AA \Paren{\II -\ZZ\ZZ^\top }}_{X} \leq \Paren{1+\epsilon } \min_{\UU : \UU^\top \UU =\II_k} \norm{\AA \Paren{\II - \UU \UU^\top}}_{X},
\end{eqnarray}
%Note that this is equivalent to 
%\abnote{maybe here we should define the problem as $\min_{\UU^\top = \II }\norm{\AA\Paren{\II - \UU\UU^\top}}_X $ and say this is equivalent to $\min_{\rank(\BB)=k }\norm{\AA - \BB}_X$ as long as $X$ is unitarily invariant. }
where $\|\cdot\|_X$ denotes some norm. Note that given $\ZZ$, one can compute
$\AA \ZZ$ with an additional $k$ queries, which will be negligible, and then $(\AA\ZZ) \cdot \ZZ^\top$ is a
rank-$k$ matrix written in factored form, i.e., as the product of an $n \times k$ matrix and a $k \times d$ matrix. Among other things, low-rank approximation provides (1) a compression of $\AA$ from
$nd$ parameters to $(n+d)k$ parameters, (2) faster matrix-vector products, since $\AA \ZZ \cdot \ZZ^\top \cdot y$ can be computed in $O((n+d)k)$ time for an arbitrary vector $y$, as opposed to the $O(nd)$ time needed to compute $\AA \cdot y$, and (3) de-noising, as often matrices $\AA$ are close to low-rank (e.g., they are the product of latent factors) but only high rank due to noise. 

Despite its tremendous importance, 
%to the best of our knowledge, 
the optimal matrix-vector product complexity of low-rank approximation is unknown for any commonly used norm. The best known upper
bound is due to Musco and Musco \cite{musco2015randomized}, who achieve $\tilde{\mathcal{O}}(k/\epsilon^{1/2})$ queries\footnote{We let $\tilde{\mathcal{O}}(f) = f \cdot \poly(\log(dk/\epsilon))$.} for
both the case when $\|\cdot\|_X$ is the commonly studied Frobenius norm $\|\BB\|_F = \left (\sum_{i, j} \BB_{i,j}^2 \right )^{1/2}$ as well as  when $\|\cdot\|_X$ is the Spectral (operator) norm $\|\BB\|_2 = \sup_{\|y\|_2 = 1} \| \BB y\|_2$. %We note that using the techniques in \cite{musco2015randomized},
%together with the so-called Ky Fan Dominance Theorem, 
%it may be possible\footnote{Based on a personal discussion between the third author and Cameron Musco.} to extend the bounds in \cite{musco2015randomized} to give 
%$\tilde{\mathcal{O}}(k/\epsilon^{1/2})$ queries 

On the lower bound front, there is a trivial lower bound of $k$, since $\AA$ may be full rank 
and achieving (\ref{eqn:lowrank}) requires $k$ matrix-vector products since one must reconstruct
the column span of $\AA$ exactly. However, {\it no lower bounds in terms of the approximation factor
$\epsilon$ were known}. We note that Simchowitz, Alaoui and Recht~\cite{SAR18} prove lower bounds for approximating the top 
$r$ eigenvalues of a symmetric matrix; however these guarantees are incomparable to those that follow from a low-rank approximation, even when the norm $\|\cdot\|_X$ is the operator norm (see Appendix \ref{sec:appendix-sar} for a brief discussion). 

\vspace{0.1in}
\textbf{Relationship to the Sketching Literature.} Low-rank approximation has been extensively studied in the sketching literature which, when $\AA$ is given
explicitly, can achieve $\bigO{\nnz(\AA)}$ time both for the Frobenius norm \cite{clarkson2013low,meng2013low,NN13}, as well as 
for Schatten-$p$ norms \cite{lw20}. 
% Here, for a matrix $\BB$, the Schatten-$p$ norm $\|\BB\|_{\calS_p} = \left (\sum_i \sigma_i(\BB)^p \right )^{1/p}$, where the $\sigma_i(\BB)$ are the singular values of the matrix $\BB$ and $p \geq 1$.
However, these works require 
reading all of the entries in $\AA$, and thus do not apply to any of the settings mentioned above. 
Further, the matrix-vector query model is especially important for problems such as trace estimation, where a low-rank approximation
is used to first reduce the variance \cite{MMMW21}. As trace estimation is often applied to implicit 
matrices, e.g., in computing Stochastic Lanczos Quadrature (SLQ) for Hessian eigendensity estimation~\cite{gkx2019investigation}, in studying the effects of batch normalization and residual connections in neural networks~\cite{ygkm2020pyhessian}, and in computing a disentanglement regularizer for deep generative models~\cite{peebles2020hessian}, 
sketching algorithms for low-rank approximation often do not apply. 

Another
important application is low-rank approximation of covariance matrices \cite{musco2017input}, for which the covariance matrix is not
given explicitly. Here, we have a data matrix $\AA$ and we want a low-rank approximation for $\AA\AA^\top$. Even when $\SS$ is a sparse sketching matrix, the matrix $\SS \AA$ is no longer sparse, and one needs to multiply $\SS\AA$  by $\AA^\top$ to obtain a sketch of $\SS \AA \AA^\top$, which is a dense matrix-matrix multiplication. Moreover, when viewed in the matrix-vector product model, sketching algorithms obtain provably worse query complexity than existing iterative algorithms (see Table \ref{fig:results} for a comparison).  
% Moreover, when viewed in the matrix-vector product model, sketching algorithms have a query complexity of at least
% $\Omega(k/\epsilon)$ for the Frobenius norm \cite{clarkson2009numerical}, at least $\Omega(k^{2/p}/\epsilon^{4/p + 1})$ for any Schatten-$p$ norm for $p \in [1,2]$ \cite{lw20}, and at least $\Omega(\min(n,d)^{1-2/p})$ for any Schatten-$p$ norm with $p > 2$ \cite{lw20}. 
%The above bounds 
%are considerably worse than the $\tilde{\mathcal{O}}(k/\epsilon^{1/2})$ bound of \cite{musco2015randomized}.
Further, as modern
GPUs often do not exploit sparsity, {\it even when the matrix $\AA$ is given, a GPU may not be able to take advantage of sparse queries}, which means the total time taken is proportional to the number of matrix-vector products. 
% Finally, for sufficiently sparse matrices, iterative algorithms in the matrix-vector product model lead to algorithms in the RAM model that obtain a polynomial improvement over sketching  algorithms (see discussion after Theorem \ref{thm:inf1}).

\vspace{0.1in}
\textbf{Motivating Schatten-$p$ Norms.}
The Schatten norms for $1 \leq p < 2$ are more robust than the Frobenius norm, as they dampen the effect of large singular values. In particular, the Schatten-$1$ norm, also known as the nuclear norm, has been widely used for robust PCA~\cite{xu2010robust, candes2011robust, yi2016fast} as well as a convex relaxation of matrix rank in matrix completion~\cite{candes2009exact,candes2010matrix}, low-dimensional Euclidean embeddings~\cite{recht2010guaranteed, tenenbaum2000global, roweis2000nonlinear}, image denoising~\cite{gu2014weighted, gu2017weighted} and tensor completion~\cite{yuan2016tensor}. In contrast, 
for $p > 2$, Schatten norms are more sensitive to large singular values and provide an approximation
to the operator norm. In particular, for a rank $r$ matrix, it is easy to see that setting $p =\log(r)/\eta$ yields a $(1+ \eta)$-approximation to the operator norm (i.e., $p=\infty$). While the Block Krylov algorithm of Musco and Musco~\cite{musco2015randomized} implies a matrix-vector query upper bound of $\tilde{\mathcal{O}}\Paren{k/\epsilon^{1/2}}$ for Schatten-$\infty$ low-rank approximation, the exact complexity of this problem remains an outstanding open problem. When $p >2$, we can interpolate between Frobenius and operator norm, and setting $p$ to be a large fixed constant can be a proxy for Schatten-$\infty$ low-rank approximation, with significantly fewer matrix-vector products (see Theorem \ref{thm:optimal_schatten_p_lra}).

%we refer the reader to \cite{lw20} for further discussion. \abnote{Add more meat here.} 

\vspace{0.1in}
\textbf{ Our Central Question.} The main question of our work is:
\begin{center}
{\it What is the matrix-vector product complexity of low-rank approximation for the Frobenius norm, and
more generally, for other matrix norms?}
\end{center}

\subsection{Our Results} 
\begin{figure*}[!htb]
\begin{center}
\resizebox{\columnwidth}{!}{
{\tabulinesep=1.2mm
\begin{tabu}{|c|c|c|c|}\hline
Problem & Frobenius & Schatten-$p$, $p \in [1,2)$ & Schatten-$p$, $p > 2$ \\\hline\hline
Sketching \cite{clarkson2009numerical,lw20} & $\Theta(k/\epsilon)$ & $\Omega(k^{2/p}/\epsilon^{4/p + 1})$ &  $\Omega(\min(n,d)^{1-2/p})$ \\\hline
Block Krylov \cite{musco2015randomized} & $\tilde{\mathcal{O}}(k/\epsilon^{1/2})$ & N/A & N/A \\\hline
Our Upper Bound & $\tilde{\mathcal{O}}(k/\epsilon^{1/3})$ & $\tilde{\mathcal{O}}(k/\epsilon^{1/3})$ & $\tilde{\mathcal{O}}(k p^{1/6}/\epsilon^{1/3})$\\\hline
Our Lower Bound & $\Omega(1/\epsilon^{1/3})$ & $\Omega(1/\epsilon^{1/3})$ &  $\Omega(1/\epsilon^{1/3})$\\\hline 
\end{tabu}
}}
\end{center}
\caption{Prior Upper and Lower Bounds on the Matrix Vector Product Complexity for Frobenius and Schatten-$p$ low-rank Approximation. The poly$(k/\epsilon)$ factors in prior sketching work for Schatten-$p$ are not explicit, but we have computed lower bounds on them to illustrate our improvements. Our bounds are optimal, up to logarithmic factors, for constant $k$. 
%Our lower bound for $p > 2$ is based on a standard conjecture in operator theory \cite{audenaert2008norm} and holds for any constant $p$. 
For $p> \log(d)/\epsilon$, spectral low-rank approximation \cite{musco2015randomized} implies an $\tilde{\mathcal{O}}\Paren{k/\sqrt{\epsilon}}$ upper bound.  %As discussed in an earlier footnote, it may be possible to adapt the arguments of \cite{musco2015randomized} to obtain a $\tilde{\mathcal{O}}(k/\epsilon^{1/2})$ upper bound also for every $p \geq 1$; however, since this is not formally written anywhere, we do not include this in the table. 
%However, it remains a polynomial factor in $1/\epsilon$ worse than the bounds we obtain.
}\label{fig:results}
\end{figure*}

%\kcnote{What is the pre-Musco-Musco bound for Krylov methods?} \abnote{gap dependent}
We begin by stating our results for Frobenius and more generally, Schatten-$p$ norm low-rank approximation
for any $p \geq 1$; see Table \ref{fig:results} for a summary. 

\begin{theorem}[Query Upper Bound, informal Theorem \ref{thm:optimal_schatten_p_lra}]
\label{thm:inf1}
Given a matrix $\AA \in \mathbb{R}^{n \times d}$, a target rank $k \in [d]$, an accuracy parameter $\epsilon \in (0,1)$ and any (not necessarily constant) $p \in [1, \bigO{\log(d)/\epsilon}]$, there exists an algorithm that uses $\tilde{\mathcal{O} }\Paren{ kp^{1/6}/\epsilon^{1/3}}$ matrix-vector products  and outputs a $d \times k$ matrix $\ZZ$ with orthonormal columns such that with probability at least $99/100$,
\begin{equation*}
    \norm{ \AA \Paren{\II - \ZZ \ZZ^\top}  }_{\calS_p} \leq (1+\epsilon) \min_{\UU : \hspace{0.05in} \UU^\top \UU = \II_k} \norm{ \AA \Paren{\II - \UU \UU^\top} \ }_{\calS_p}.
\end{equation*}
When $p \geq \log(d)/\epsilon$, we get $\tilde{\mathcal{O}}\Paren{k/\sqrt{\epsilon}}$ matrix-vector products.
\end{theorem}

%\begin{remark}[Comparison to Frobenius LRA]
We note that for Frobenius norm low-rank approximation (Schatten $p$ for $p=2$), we improve the prior matrix-vector product bound of $\tilde{\mathcal{O}}(k/\epsilon^{1/2})$ by Musco and Musco~\cite{musco2015randomized}
to $\tilde{\mathcal{O}}(k/\epsilon^{1/3})$. 
For Schatten-$p$ low-rank approximation for $p \in [1, 2)$, we improve work of
Li and Woodruff~\cite{lw20} who require query complexity at least $\Omega(k^{2/p}/ \epsilon^{4/p +1})$, which is a polynomial factor worse in both $k$ and $1/\epsilon$ than our $\tilde{\mathcal{O}}(k/\epsilon^{1/3})$ 
bound. 

For $p>2$, \cite{lw20} obtain a query complexity of 
$\Omega(\min(n, d)^{1-2/p})$. We drastically improve this to 
$\tilde{\mathcal{O}}(k/\epsilon^{1/3})$, which does not depend on $d$ or $n$ at all. 
%Although stated for $p = 2$ and $p =\infty$, 
Setting $p =\log(d)/\epsilon$ suffices to obtain a $(1+\epsilon)$-approximation to the spectral norm ($p=\infty)$, and we obtain an $\tilde{\mathcal{O}}\Paren{k/\sqrt{\epsilon}}$ query algorithm, matching the best known bounds for spectral low-rank approximation~\cite{musco2015randomized}. When $p>\log(d)/\epsilon$, we can simply run Block Krylov for $p=\infty$. 
%We emphasize that prior to our work, no dimension-independent algorithm was known for any $p$ between $2$ and $\log(d)/\epsilon$.
%
%a bound of $\tilde{\mathcal{O}}(k/\epsilon^{1/2})$. 
%
%\end{remark}
%
%\begin{remark}[Comparison to Schatten LRA]
%Next, in the setting where $p\in [1,2)$, recent work of Li and Woodruff~\cite{lw20} uses sketching techniques to get a running time of $\tilde{\mathcal{O}} \Paren{\nnz(\AA) + n k^{2(\omega-1)/p}/\epsilon^{(4/p-1)\cdot(\omega-1)} }$, and for sparse matrices, we obtain a factor $(12/p -3)$ improvement in the exponent of $1/\epsilon$. Translating their result to the matrix-vector model, their query complexity is $\Omega(k^{2/p}/ \epsilon^{4/p +1})$, which is polynomial factor worse in both $k$ and $1/\epsilon$.
%
%Finally, in the setting where $p>2$, Li and Woodruff~\cite{lw20} obtain a running time of $\tilde{\mathcal{O}}\Paren{\nnz(\AA) + n d^{(\omega-1)(1-2/p)} (k/\epsilon)^{2(\omega-1)/p}  } $, which is a polynomial factor worse in the input dimension $d$ than our running time, and $d$ may be significantly larger than $k$. Similarly, the query complexity of their algorithm is $\Omega(d^{1-2/p})$. 
%\end{remark}

\begin{remark}[Comments on the RAM Model]
Although our focus is on minimizing the number of matrix-vector products, which is the key resource
in the applications described above, our bounds also improve the running time of low-rank approximation algorithms when the matrix $\AA$ has a small number of non-zero entries and is
explicitly given. For simplicity, we state our bounds and those of previous work without using
algorithms for fast matrix multiplication; similar improvements hold when using such
algorithms. When $\nnz(\AA) = O(n)$, 
for Frobenius norm low-rank approximation, 
work in the sketching literature, and in particular \cite{avron2017sharper} 
(building off of \cite{clarkson2013low, NN13, cohen2016nearly}), achieves
$O(nk^2/\epsilon)$ time. 
In contrast, in this setting our runtime is 
$\tilde{\mathcal{O}}(nk^2/\epsilon^{2/3})$. Similarly, 
for Schatten-$p$ low-rank approximation for $p \in [1,2)$, the previous best \cite{lw20} requires
$\tilde{\Omega}(n k^{4/p}/\epsilon^{(8/p-2)})$ time, while for $p > 2$ \cite{lw20} requires
$\tilde{\Omega}(n d^{2(1-2/p)} (k/\epsilon)^{4/p})$ time. In both cases our runtime is only  
$\tilde{\mathcal{O}}(nk^2p^{1/3}/\epsilon^{2/3})$. 
We obtain analogous improvements when the sparsity $\nnz(\AA)$ is allowed to be $n (k/\epsilon)^C$ for a
small constant $C > 0$. 
\end{remark} 

Next, we state our lower bounds on the matrix-vector query complexity of Schatten-$p$ low-rank approximation.

\begin{theorem}[Query Lower Bound for constant $p$, informal Theorem \ref{thm:mv_lowerboundn_schatten_12} and Theorem \ref{thm:query_lower_bound_p>2} ]
\label{thm:inf2}
Given $\eps >0$, and a fixed constant $p \geq 1$, there exists a distribution $\mathcal{D}$ over $n \times n$ matrices such that for $\AA \sim \mathcal{D}$, any algorithm that with at least constant probability outputs a unit vector $v$ such that $\norm{ \AA  \Paren{\II -  v v^\top} }^p_{\calS_p} \leq (1+\eps)\min_{\norm{u}_2=1} \norm{\AA \Paren{\II - u u^\top }  }^p_{\calS_p}$ must perform $\Omega(1/\eps^{1/3})$ matrix-vector queries to $\AA$.
\end{theorem}

\begin{remark}
We note that this is the first lower bound as a function of $\epsilon$ for this problem, even for the well-studied case of $p=2$, achieving an $\Omega(1/\epsilon^{1/3})$ bound, which is tight for any constant $k$, simultaneously for all constant $p \geq 1$.  
\end{remark}

\begin{remark}
Braverman,  Hazan, Simchowitz and Woodworth \cite{braverman2020gradient} and Simchowitz, Alaoui and Recht \cite{SAR18} establish 
eigenvalue estimation lower bounds that we use in our arguments, but
their results do not directly imply low-rank approximation lower bounds 
for any matrix norm that we are aware of, including spectral low-rank approximation, i.e., $p = \infty$ (see Appendix \ref{sec:appendix-sar}). 
\end{remark}

% Finally, for the case of $p > 2$, we obtain the same lower bound of
% $\Omega(1/\epsilon^{1/3})$,
% assuming a standard conjecture in operator theory \cite{audenaert2008norm} (see Conjecture \ref{conj:norm_compression}), and thus our algorithms are optimal for all constant $k$ and $p$. 

% \begin{theorem}[Query Lower Bound for $p >2$, informal Theorem \ref{thm:query_lower_bound_p>2} ]
% \label{thm:inf3}
% Given $\eps >0$, and $p > 2$, there exists a distribution $\mathcal{D}$ over $n \times n$ matrices such that for $\AA \sim \mathcal{D}$, any algorithm that with constant probability outputs a unit vector $v$ such that $\norm{ \AA\Paren{\II- vv^\top} }^p_{\calS_p} \leq (1+\eps) \min_{\norm{u}_2=1}\norm{\AA \Paren{\II - u u^\top }}^p_{\calS_p}$ must make $\Omega(1/\eps^{1/3})$ matrix-vector queries to $\AA$, given that Conjecture \ref{conj:norm_compression} holds.
% \end{theorem}

\paragraph{Matrix Polynomials and Streaming Algorithms.}

% \textbf{Low-Rank Approximation of Matrix Polynomials.}
Since our algorithms are based on iterative methods, they generalize naturally to low-rank approximations of matrices of the form $\AA\Paren{ \AA^\top \AA}^\ell$ and $\Paren{ \AA^\top \AA}^\ell$ for any integer $\ell$, given $\AA$ as input. We defer the details to Appendix \ref{sec:polynomial}.

% \kcnote{I think David suggested making the runtime in terms of M-v queries? Or maybe that was elsewhere. There was also David's observation that for $\ell=1$, sketching can do nnz(A) + lesser-terms, because $(SA)^T(SA)$ can be computed within budget} \abnote{fixed}
% %
%\begin{remark}
%The proof just follows from running our algorithm on $\MM =\Paren{\AA^\top\AA}^\ell$ instead. It is straightforward to simulate a matrix-vector product of the form $\MM v$ using access to matrix-vector products for $\AA$ with a $\bigO{\ell}$ overhead.
%\end{remark}

% \paragraph{Instance-Optimal Top-$k$ Subspace Computation.}

Since we work in the matrix-vector model, our algorithms naturally extend to the multi-pass turnstile streaming setting. Notably, for $p > 2$, with $\bigO{\log(d/\epsilon)p^{1/6}/\epsilon^{1/3}}$ passes we are able to improve the $\tilde{\mathcal{O}}\Paren{ n\Paren{\frac{k n^{1-2/p} }{\epsilon^2} + \frac{k^{2/p} + n^{1-2/p}}{\epsilon^{2+2/p}}}}$ memory bound of \cite{lw20} to $\tilde{\mathcal{O}}\Paren{nk/\epsilon^{1/3}}$. We defer the details to Appendix \ref{sec:stream}.

\subsection{Open Questions}

We note that our lower bounds are tight only when the target rank $k$ and Schatten norm $p$ are fixed constants. In particular, it is open to obtain matrix-vector lower bounds that grow as a function of $k$, $p$ and $1/\epsilon$. For the important special case of Spectral low-rank approximation ($p=\infty$), it is open to obtain any lower bound that grows as a function of $1/\epsilon$, even when the target rank $k=1$ (see Appendix \ref{sec:appendix-sar} for more details). We also note that improving our upper bound to even $p^{1/6-o(1)}$ would imply a faster algorithm for Spectral low-rank approximation, addressing the main open question in~\cite{woodruff2014sketching}.

% \begin{corollary}[Corollary of Theorems 3.6 and 4.4 in~\cite{lw20}]
% Given a matrix $\AA \in \mathbb{R}^{n \times d}$, a target rank $k \in [d]$, an accuracy parameter $\epsilon \in (0,1)$, there exists a streaming algorithm that makes a single pass over the input, and outputs a $d\times k$ matrix $\ZZ$ with orthonormal columns such that with probability at least $9/10$,
% \begin{equation*}
%     \norm{ \AA \Paren{\II - \ZZ \ZZ^\top}  }^p_{\calS_p} \leq (1+\epsilon) \min_{\UU : \hspace{0.05in} \UU^\top \UU = \II_k} \norm{ \AA \Paren{\II - \UU \UU^\top} \ }^p_{\calS_p}.
% \end{equation*}
% For $1\leq p<2$, the space required is $\tilde{\mathcal{O}}\Paren{n \Paren{ \frac{k+ k^{2/p} }{\epsilon^2} +  \frac{ k^{2/p} }{\epsilon^{1+2/p} }  }  }$ and for $p > 2$, the space required is $\tilde{\mathcal{O}}\Paren{ n\Paren{\frac{k n^{1-2/p} }{\epsilon^2} + \frac{k^{2/p} + n^{1-2/p}}{\epsilon^{2+2/p}  } } }$. 
% \end{corollary}

\section{Technical Overview} 
For our technical overview, we drop polylogarithmic factors appearing in the analysis and assume the input $\AA$ is a symmetric $n \times n$ matrix (we handle arbitrary $n \times d$ matrices in Section \ref{sec:upper_bound}).

\subsection{Algorithms for Low-Rank Approximation}
We first describe our algorithm for the special case of rank-$1$ approximation in the Frobenius norm, i.e., $p =2$.  Our algorithm is inspired by the Block Krylov algorithm of Musco and Musco~\cite{musco2015randomized}. Briefly, their algorithm begins with a random starting vector $g$ (block size is $1$) and computes the Krylov subspace $\KK = [ \AA g; \AA^2 g ; \ldots; \AA^q g ]$, for $q = \bigO{1/\epsilon^{1/2}}$. Next, their algorithm computes an orthonormal basis for the column span of $\KK$, denoted by a matrix $\QQ$, and outputs the top singular vector of $\QQ^\top \AA^2 \QQ$, denoted by $z$ (see Algorithm \ref{algo:simul_power_iter} for a formal description). It follows from Theorem~1, guarantee (1) in \cite{musco2015randomized} that 
\begin{equation}
\label{eqn:guarantee_k=1_p=2}
    \norm{ \AA \Paren{\II - z z^\top} }_F^2 \leq \Paren{1+\epsilon}\min_{\norm{u}_2=1 } \norm{\AA \Paren{\II - uu^\top }}_F^2,
\end{equation}
and it is easy to see that this algorithm requires $\Theta\Paren{1/\epsilon^{1/2}}$ matrix-vector products. A na\"ive analysis requires an $\bigO{1/\epsilon}$-degree polynomial in the matrix $\AA$ to obtain \eqref{eqn:guarantee_k=1_p=2}, while \cite{musco2015randomized} use Chebyshev polynomials to approximate the threshold function between first and second singular value, and save a quadratic factor in the degree. The guarantee in \eqref{eqn:guarantee_k=1_p=2} then follows from observing that the best vector in the Krylov subspace is at least as good as the one that exists using Chebyshev polynomial approximation.

\begin{mdframed}
  \begin{algorithm}[Algorithm Sketch for Frobenius rank-$1$ LRA ]
  \label{algo:sketch}
    \mbox{}
    \begin{description}
    \item[Input:] 
    An $n \times n$ symmetric matrix $\AA$,  accuracy parameter $0<\eps<1$. 
    %\item\mbox{}
    \begin{enumerate}
    \item Run Block Krylov for $\bigO{1/\epsilon^{1/3}}$ iterations with a random starting vector $g$. Let $z_1$ be the resulting output. 
    \item Run Block Krylov for $\bigO{\log(n/\epsilon)}$ iterations, but initialize with an $n \times b$ random matrix $\GG$, where $b = \bigO{1/\epsilon^{1/3}}$. Let $z_2$ be the resulting output. 
    \end{enumerate}
    \item[Output:]  $z = \arg\max_{z_1, z_2 } \Paren{ \norm{\AA z_1}^2_2, \norm{\AA z_2}^2_2 } $.
    \end{description}
  \end{algorithm}
\end{mdframed}

Our starting point is the observation that while we require degree $\Theta\Paren{1/\epsilon^{1/2}}$ to separate the first and second singular values, if any subsequent singular value is sufficiently separated from $\sigma_1$, a significantly smaller degree polynomial suffices. In the context of Krylov methods, this translates to the intuition that starting with a matrix $\GG$ with $b$ columns (block size is $b$) should result in fewer iterations to find some vector in the top $b$ subspace of $\AA$. On the other hand, if no such singular value exists, the norm of the tail must be large and we can get away with a less accurate solution.  We show that we can indeed exploit this trade-off by running Block Krylov on two different scales in parallel and then combine the solution. In particular, we use Algorithm \ref{algo:sketch}.

Algorithm \ref{algo:sketch} captures the extreme points of the trade-off between the size of the starting matrix and the number of iterations, such that the total number of matrix-vector products is at most $\tilde{\mathcal{O}}({1/\epsilon^{1/3}})$. Further, we can compute the squared Euclidean norms of $\AA z_1$ and $\AA z_2$ with an additional matrix-vector product, and it remains to analyze the Frobenius cost of projecting $\AA$ on the subspace $\II - zz^\top$, where $z$ is the unit vector output by Algorithm \ref{algo:sketch}.

Using gap-independent guarantees for Block Krylov (see Lemma \ref{lem:gap_independent} for a formal statement), it follows that with $\bigO{1/\epsilon^{1/3}}$ iterations, we have 
\begin{equation}
\label{eqn:intro_gap_indep}
    \norm{\AA z_1}_2^2 \geq \sigma_1^2(\AA) - \epsilon^{2/3} \sigma_{2}^2(\AA). 
\end{equation}
In contrast, using gap-dependent guarantees (see Lemma \ref{lem:gap_dependent}) for Block Krylov initialized with block size $b$, it follows that for any $\gamma>0$, running $q = \log(1/\gamma)\cdot\sqrt{ \sigma_1(\AA) / \Paren{\sigma_1(\AA) - \sigma_b(\AA)}}$ iterations results in $z_2$ such that
\begin{equation}
\label{eqn:intro_gap_dep}
    \norm{\AA z_2}_2^2 \geq \sigma_1^2(\AA) - \gamma \sigma_{2}^2(\AA). 
\end{equation}
If $\sigma_b(\AA) \leq \sigma_1(\AA)/2$, we can set $\gamma = \epsilon / n$ in Equation \eqref{eqn:intro_gap_dep} to obtain a highly accurate solution. Further, regardless of the input instance, Step 3 in Algorithm \ref{algo:sketch} ensures that we get the best of both guarantees, \eqref{eqn:intro_gap_indep} and \eqref{eqn:intro_gap_dep}. Then, observing that $\II - z z^\top$ is an orthogonal projection matrix (see Definition \ref{def:orthogonal_projection}) and using the Pythagorean Theorem for Euclidean space we have:
\begin{equation}
\label{eqn:inf_cost_lra}
    \norm{\AA \Paren{\II - z z^\top} }_F^2 = \norm{\AA }_F^2 - \norm{\AA zz^\top}_F^2 = \norm{\AA }_F^2 - \norm{\AA z}_2^2, 
\end{equation}
where the second inequality follows from unitary invariance (see Fact \ref{fact:unitary_inv}) of the Frobenius norm and that the squared Frobenius norm of a rank-$1$ matrix $\AA z$ (vector) is equal to its squared Euclidean norm. If it happens that $\sigma_2(\AA) \leq \sigma_1(\AA) /2$, i.e., a constant gap exists between the first two singular values, then since guarantee \eqref{eqn:intro_gap_dep} implies that $\norm{\AA z}_2^2 \geq \sigma_1^2(\AA) - (\epsilon/n) \sigma_2^2(\AA) $, we can plug this into \eqref{eqn:inf_cost_lra} to yield a $(1+\epsilon/n)$-approximate solution. Hence, we focus on instances where $\sigma_2(\AA) > \sigma_1(\AA) /2$. 

%Intuitively, we show that  for such instances, either
%the Frobenius norm of the tail is large and we can afford less accurate solutions or there is a constant gap between the first and $b$-th singular values, where $b = \bigO{1/\epsilon^{1/3}}$. A bit more formally, 
Consider the case where the Frobenius norm of the tail is large, i.e., $ \|\AA -\AA_1 \|_F^2 \geq \sigma_2^2(\AA) / \epsilon^{1/3}$, where $\AA_1$ is the best rank-$1$ approximation to $\AA$.  Then we only require an $\epsilon^{2/3}$-approximate solution (plugging guarantee \eqref{eqn:intro_gap_indep} into \eqref{eqn:inf_cost_lra} ) since
\begin{equation}
\label{eqn:weaker-error-guarantee}
    \norm{\AA \Paren{\II - z_1 z^\top_1} }_F^2 \leq \norm{\AA}_F^2 - \sigma_1^2(\AA) + \epsilon^{2/3} \sigma_2^2(\AA) \leq \norm{\AA - \AA_1}_F^2 + \epsilon \norm{\AA - \AA_1}_F^2.
\end{equation}
Otherwise, $\sum_{i=2}^n \sigma_i^2(\AA) <  \sigma_2^2(\AA) / \epsilon^{1/3} $, which implies that there is a constant gap between the second and $b$-th singular values, where $b= \bigO{1/\epsilon^{1/3}}$. To see this, observe if $\sigma_b(\AA) > \sigma_2(\AA)/4$, then 
$\sum_{i=2}^n \sigma_i^2(\AA) \geq  \sum_{i=2}^b \sigma_i^2(\AA) \geq b \sigma_2^2(\AA)/4$, 
which is a contradiction when $b > 10/\epsilon^{1/3}$, and thus $\sigma_b(\AA) \leq \sigma_2(\AA)/4 < \sigma_1/2$. Now we can apply guarantee \eqref{eqn:intro_gap_dep} with $q = \bigO{\log(n/\epsilon)}$ and conclude $\norm{\AA z}_2^2 \geq \sigma_1^2(\AA) - (\epsilon/n) \sigma_2^2(\AA) $, yielding a highly accurate solution yet again. Overall, this suffices to obtain a $(1+\epsilon)$-approximate solution with $\tilde{\mathcal{O}}({1/\epsilon^{1/3}})$ matrix-vector queries.

% \kcnote{Sometimes it's $\sigma_i(\AA)$, and sometimes it's just $\sigma_i$. It doesn't matter, except it make me wonder if the difference means anything.}\abnote{no diff, just typos by me }

\paragraph{Challenges in generalizing to Schatten $p\neq 2$ and rank $k>1$.} The outline above crucially relies on the norm of interest being Frobenius. In particular, we use the Pythagorean Theorem to analyze the cost of the candidate solution in Equation \eqref{eqn:inf_cost_lra}; however, the Pythagorean Theorem does not hold for non-Euclidean spaces. Therefore, a priori, it is unclear how to analyze the Schatten-$p$ norm of a candidate rank-$1$ approximation. 
%A natural approach to analyzing the spectrum of $\AA(\II - zz^\top)$ would be to use Cauchy's Interlacing theorem (Fact \ref{fact:cauchy}) on the matrices $\AA$ and $\AA z z^\top$. However, this is too weak, since in the worst-case this implies that the singular values of $\AA z z^\top$ can be $\sigma_2(\AA), \sigma_3(\AA), \ldots \sigma_{n}(\AA)$. Therefore, we would incur additive error $\sigma_1^2(\AA)$ which can be much larger than $\epsilon \norm{\AA -\AA_1}_F^2$. Yet another approach is to use 
A proxy for the Pythagorean Theorem that holds for Schatten-$p$ norms is Mahler's operator inequality (see Fact \ref{fact:mahler_ortho_ineq}), which is in the right direction but holds only for $p\geq 2$, whereas
we would like to handle all $p \geq 1$. Separately, for $p>2$, the case where the tail is small corresponds to $\norm{\AA - \AA_1}_{\calS_p}^p \leq \sigma_2^p\Paren{\AA}/\epsilon^{1/3}$. Therefore, na\"ively extending the above argument requires picking a block size that scales proportional to $\bigO{2^p/\epsilon^{1/3}}$ to induce a constant gap between $\sigma_1$ and $\sigma_b$, and the number of matrix-vector products scales exponentially in $p$.   
%We note that in fact for all $p>2$, we can use this inequality to obtain optimal bounds but our proposed method handles all $p\geq 1$ simultaneously.

Finally, in the above outline, we also crucially use that $\norm{\AA z z^\top}_F^2 = \norm{\AA z}_2^2$. Observe that this no longer holds
if we replace $z$ with a matrix $\ZZ$ that has $k$ orthonormal columns. Therefore, it remains unclear how to relate $\norm{\AA \ZZ }^p_{\calS_p}$ to $\norm{\AA \ZZ_{*,i}}_2^2$, yet the vector-by-vector error guarantee obtained by Block Krylov (see Lemmas \ref{lem:gap_independent} and \ref{lem:gap_dependent}) only bounds the latter.

\paragraph{Handling all Schatten-$p$ Norms and $k>1$.}
%Next, we outline our approach to analyze the Schatten-$p$ norm of a candidate rank-$k$ solution. 
We modify our algorithm to run Block Krylov on $\AA^\top$ and obtain an orthonormal matrix $\WW$ such that for all $i \in[k]$,
\begin{equation}
\label{eqn:per-vector-w}
    \norm{\AA^\top \WW_{*,i}}^2 \geq \sigma_i^2(\AA)- \gamma \sigma^2_{k+1}(\AA),
\end{equation}
for some $\gamma>0$. We then analyze the cost $\norm{ \AA\Paren{\II - \ZZ \ZZ^\top} }_{\calS_p}^p$, where $\ZZ$ is a basis for $\AA^\top \WW$. 
Our key insight is to interpret the input matrix $\AA$ as a partitioned operator (block matrix) and invoke \emph{pinching inequalities} for such operators. Pinching inequalities were originally introduced to understand unitarily invariant norms over direct sums of Hilbert spaces~\cite{von1937some,schatten1960norm}. 
In our setting, given a block matrix $\MM = \begin{pmatrix}
        \MM^{(1)}  &  \MM^{(2)} \\\
         \MM^{(3)} & \MM^{(4)}
\end{pmatrix} $, the \emph{pinching inequality} (see Fact \ref{fact:pinching}) implies that for all $p \geq 1$, 
\begin{equation}
\label{eqn:pinching-intro}
    \norm{\MM}_{\calS_p}^p \geq \norm{\MM^{(1)} }_{\calS_p}^p + \norm{\MM^{(4)}}_{\calS_p}^p.
\end{equation}
A priori, it is unclear how to use Equation \eqref{eqn:pinching-intro} to bound $\norm{\AA \Paren{\II - \ZZ \ZZ^\top}}_{\calS_p}^p$. First, we establish a general inequality for the Schatten norm of a matrix times an orthogonal projection. Let $\PP$ and $\QQ$ be any $n \times n$ orthogonal projection matrices with rank $k$ (see Definition \ref{def:orthogonal_projection}). Then, we prove (see Lemma \ref{lem:schatten_p_orthogonal_projections} for details) that for any matrix $\AA$, 
\begin{equation}
\label{eqn:schatten-inequality-projectors-intro}
    \norm{\AA}_{\calS_p}^p \geq \norm{\PP \AA \QQ}_{\calS_p}^p + \norm{\Paren{\II - \PP} \AA \Paren{ \II - \QQ} }_{\calS_p}^p.
\end{equation}
To obtain this inequality, we use a rotation argument along with the fact that the Schatten-$p$ norms are unitarily invariant to show that
%\begin{equation*}
$    \norm{\AA }_{\calS_p}^p = \norm{\begin{pmatrix}
        \AA^{(1)}  &  \AA^{(2)} \\\
         \AA^{(3)} & \AA^{(4)}
\end{pmatrix}  }_{\calS_p}^p$,
%\end{equation*}
where $\norm{\AA^{(1)}}_{\calS_p} = \norm{\PP \AA \QQ}_{\calS_p}$ and  $\norm{\AA^{(4)}}_{\calS_p} = \norm{\Paren{\II - \PP} \AA \Paren{\II - \QQ} }_{\calS_p}$,
and then we can apply Equation \eqref{eqn:pinching-intro} to the block matrix above.

Once we have established Equation \eqref{eqn:schatten-inequality-projectors-intro}, we can set $\PP=\WW \WW^\top $ and set $\QQ = \ZZ \ZZ^\top$ to be the projection matrix corresponding to the column span of $\AA^\top \WW \WW^\top$. Then, we have that $\PP \AA \QQ = \WW \WW^\top \AA $ and $ \Paren{\II - \PP}\AA\Paren{\II - \QQ} = \AA \Paren{\II - \ZZ \ZZ^\top}$, and combined with \eqref{eqn:schatten-inequality-projectors-intro} this yields 
\begin{equation}
\label{eqn:upper-bound-cost-sp-intro}
    \norm{ \AA \Paren{\II - \ZZ \ZZ^\top}}_{\calS_p}^p \leq \norm{ \AA }_{\calS_p}^p - \norm{ \WW \WW^\top \AA   }_{\calS_p}^p.
\end{equation}
To obtain a bound on $\norm{ \WW \WW^\top \AA   }_{\calS_p}^p$, we appeal to the per-vector guarantees in Equation \eqref{eqn:per-vector-w}. However, translating from $\ell_2^2$ error to $\sigma_p^p\Paren{\WW^\top \AA}$ incurs a mixed guarantee (see Lemma \ref{lem:correlated_projections} for details):
\begin{equation*}
     \norm{ \WW \WW^\top \AA   }_{\calS_p}^p \geq \norm{  \AA_k }_{\calS_p}^p - \bigO{ \gamma p } \sum_{i \in [k]}  \sigma_{k+1}^2\Paren{\AA} \sigma_{i}^{p-2}\Paren{\AA}.
\end{equation*}
To use this bound, we require $\sigma_1(\AA)$ to be comparable to $\sigma_{k+1}(\AA)$ and thus we require an involved case analysis, which appears in the proof of Theorem \ref{thm:optimal_schatten_p_lra}.

% Now that we have a bound on $\norm{ \AA \Paren{\II - \ZZ \ZZ^\top}}_{\calS_p}^p$, the proof for $p \in [1,2)$ and $k=1$ is relatively straightforward. However, generalizing to $p>2$ without incurring an exponential dependence on $p$ and handling arbitrary target rank $k$ requires additional ideas. 

\paragraph{Avoiding an exponential dependence on $p$.} Our main insight here is that we do not require a block size that induces a constant gap between singular values. Instead, we first observe that if the block size $b$ is large enough such that $\sigma_b \leq \sigma_{2}/(1 + 1/p)$, then $\bigO{\log(n/\epsilon)\sqrt{p}}$ iterations suffice to obtain a vector $z$ such that $\norm{\AA z}_2^2 \geq \sigma_1^2\Paren{\AA} - \Paren{\epsilon/n} \sigma_2^2\Paren{\AA}$. Therefore, we can trade-off the threshold for the Schatten norm of the tail with the number of iterations as follows: if $\norm{\AA - \AA_1}_{\calS_p}^p \leq \frac{1}{p^{1/3} \epsilon^{1/3}} \sigma_2^p\Paren{\AA}$, then setting $b = (1+1/p)^p/(\epsilon p)^{1/3} = \Theta(1/(\epsilon p)^{1/3})$  suffices to induce a gap of $1+1/p$ with block size $b$. The total number of matrix-vector products is $\bigO{b \cdot \log(n/\epsilon) \sqrt{p}} = \tilde{\mathcal{O}}(p^{1/6}/\epsilon^{1/3})$, since $p$ can be assumed to be at most $(\log n)/\epsilon$. 
Otherwise, $\norm{\AA - \AA_1}_{\calS_p}^p > \frac{1}{p^{1/3} \epsilon^{1/3}} \sigma_2^p\Paren{\AA}$, and we only require a  $(1+ \epsilon^{2/3}/p^{1/3})$-approximate solution instead (compare with Equation \eqref{eqn:weaker-error-guarantee}). Using gap-independent bounds (see Lemma \ref{lem:gap_independent}), it suffices to start with block size $1$ and run $\bigO{\log(n/\epsilon) p^{1/6}/\epsilon^{1/3} }$ iterations to obtain a $(1 + \epsilon^{2/3}/p^{1/3})$-approximate solution.

\paragraph{Avoiding a Gap-Dependent Bound.}
We note that even when there is a constant gap between the first and second singular values, and the per vector guarantee is highly accurate, i.e., for all $i \in [k]$, 
    $\norm{\AA \ZZ_{*, i}}^2 \geq \sigma^2_i(\AA) - \poly\Paren{ \frac{\epsilon}{d}} \sigma_{k+1}^2(\AA)$, 
it is not clear how to lower bound $\norm{ \AA\ZZ }_{\calS_p}^p$ in Equation \ref{eqn:upper-bound-cost-sp-intro}. In general, the best bound we can obtain using the above equation is 
\begin{equation}
\label{eqn:bound-az-intro}
    \norm{\AA \ZZ}_{\calS_p}^p \geq \norm{\AA_k }_{\calS_p}^p - \bigO{\frac{\epsilon}{\poly(d)}} \sigma_{k+1}^2 \cdot \sum_{i \in [k]} \sigma_i^{p-2},
\end{equation}
which may be vacuous when the top $k$ singular values are significantly larger than $\sigma_{k+1}$
and $p > 2$. One could revert to a gap-dependent bound, where the error is in terms of the gap between $\sigma_1$ and $\sigma_{k+1}$, which one could account for by running an extra factor of $\bigO{\log(\sigma_1/\sigma_{k+1})}$ iterations. 

To avoid this gap-dependent bound, we split $\AA$ into a head part $\AA_H$ and a tail part $\AA_T$, such that $\AA_H$ has all singular values that are at least $\Paren{1+1/d}\sigma_{k+1}$ and $\AA_T$ has the remaining singular values. We then bound $\norm{\AA_H\Paren{\II - \ZZ \ZZ^\top} }_{\calS_p}$ and $\norm{\AA_T\Paren{\II - \ZZ \ZZ^\top}}_{\calS_p}$ separately. Repeating the above analysis, we can obtain Equation \eqref{eqn:bound-az-intro} for $\AA_T$ instead, and since all singular values larger than $\sigma_{k+1}$ in $\AA_T$ are bounded, we can obtain $\norm{\AA_T \Paren{\II - \ZZ\ZZ^\top} }_{\calS_p}^p \leq \bigO{\epsilon k/\poly(d)} \sigma_{k+1}^p$. To adapt the analysis for $\AA_T$ and obtain this
bound, we use Cauchy's interlacing theorem to relate the $j$-th singular value of 
$\AA_T \Paren{\II - \ZZ\ZZ^\top}$ to the $(i^*+j)$-th singular value of 
$\AA \Paren{\II - \ZZ\ZZ^\top}$, where $i^*$ is the rank of $\AA_H$. We lower bound the
$(i^*+j)$-th singular value of $\AA \Paren{\II - \ZZ\ZZ^\top}$ using the per vector guarantee of
\cite{musco2015randomized}. 

To bound $\norm{\AA_H\Paren{\II - \ZZ \ZZ^\top} }_{\calS_p}$, we observe it has rank at most $k$ and thus
\begin{equation*}
\norm{\AA_H\Paren{\II - \ZZ \ZZ^\top} }_{\calS_p} \leq \sqrt{k}\cdot \norm{\AA_H\Paren{\II - \ZZ \ZZ^\top} }_{F} = \sqrt{k} \cdot \sqrt{ \norm{\AA_H}_F^2  - \norm{\AA_H \ZZ}_F^2 },
\end{equation*}
and we show how to bound this term in Section~\ref{sec:upper_bound}. Intuitively, while the
$k$-dimensional subspace that we find can ``swap out" singular vectors corresponding to
singular values $\sigma_i$ for which $\sigma_i$ is very close to $\sigma_{k+1}$, since they serve
equally well for a Schatten-$p$ low-rank approximation, for singular values $\sigma_i$ that are
a bit larger than $\sigma_{k+1}$, the $k$-dimensional subspace we find cannot do this. More precisely,
if $y$ is a singular vector of $\AA_H$ with singular value $\sigma_i$, then the projection of $y$ onto 
the $k$-dimensional subspace that our algorithm finds (namely, $\ZZ$) must be at least  
$1 - \sigma_{k+1}^2/((\sigma_i^2 - \sigma_{k+1}^2)\poly(d))$, which suffices to bound the above 
since the additive error is inversely proportional to $\sigma_i^2$ when $\sigma_i^2 \gg \sigma_{k+1}^2$,
and so the very tiny additive error negates the effect of very large singular values. 

\subsection{ Query Lower Bounds.}
%Next, we provide an overview of the ideas involved in proving Theorem  \ref{thm:inf2}. 
Our lower bounds rely on the hardness of estimating the smallest eigenvalue of a Wishart ensemble (see Definition \ref{def:wishart_ensemble}), as established in recent work of Braverman,  Hazan, Simchowitz and Woodworth~\cite{braverman2020gradient}. In particular, \cite{braverman2020gradient} show that for a $d \times d$ instance $\WW$ of a Wishart ensemble, estimating $\lambda_d(\WW)$ (minimum eigenvalue) to additive error $1/d^2$ requires $\Omega(d)$ adaptive matrix-vector product queries (see Theorem 3.1 in \cite{braverman2020gradient}). To obtain hardness for Schatten-$p$ low-rank approximation, we show that when $d=\Theta\Paren{1/\epsilon^{1/3}}$, any candidate unit vector $z$ that satisfies
%\begin{equation*}
    $\norm{ \Paren{\II - \WW/5}\Paren{\II - zz^\top} }_{\calS_p}^p \leq \Paren{1+\epsilon} \min_{\norm{u}_2=1 }\norm{ \Paren{\II - \WW/5}\Paren{\II - uu^\top} }_{\calS_p}^p,$
%\end{equation*}
can be used to obtain an estimate $\hat{\lambda}_d = \frac{5}{p} \Paren{1- \norm{\Paren{\II - \WW/5}z}_2^p }$ such that $\hat{\lambda}_d = (1\pm 1/d^2) \lambda_d\Paren{\II - \WW/5}$. Let $\AA = \Paren{\II - \WW/5}$.  To show our query lower bound, in contrast to the analysis of our algorithm, the challenge is now to lower bound $\norm{ \AA\Paren{\II - zz^\top} }_{\calS_p}^p$ in terms of $\norm{ \AA }_{\calS_p}^p$ and $\norm{ \AA z }_{2}^p$ (contrast with Equation \eqref{eqn:upper-bound-cost-sp-intro}).
%\kcnote{Sorry, why is that the challenge? To translate relative to additive error?}
% \abnote{well not really. we know exactly what the RHS is, since the matrix $\AA$ is fixed. We want to use $1- \| Az \|_2$ as an estimate for the min eigenvalue of $\AA$. So, we have to understand what $\| \AA (\II - zz^\top)\|_p $ implies about $\| Az \|_2$, and  so we want to obtain a lower bound on $\| \AA (\II - zz^\top)\|_p $. an upper bound would tell us nothing... }
% \kcnote{You mean the challenge is to prove the "To obtain hardness..." statement? That's not the challenge "now", after knowing the "To obtain hardness...." statement, it's what the challenge of proving the "To obtain hardness..." is. That is, the last statement is not an *additional* challenge, it's a re-statement. Or what?}

% \kcnote{Why not define some matrix, say $\hat{\WW}$,
% as $\II - \WW/5$, and use that throughout the above?}
% \abnote{good point!}

\paragraph{Projection Cost via Araki-Lieb-Thirring.} First, we note that the case of $p=2$ is easy given
the Pythagorean theorem. For $p \in [1,2)$, we can establish an inequality fairly straightforwardly: using the trace inner product definition of Schatten-$p$ (see Definition \ref{def:schatten} ) norms, we have, 
\begin{equation}
\label{eqn:lb_schatten_cost}
    \norm{ \AA \Paren{\II - zz^\top} }_{\calS_p}^p = \trace{  \Paren{\Paren{\II - zz^\top}^2  \AA^2 \Paren{\II - zz^\top}^2 }^{p/2} },
\end{equation}
Since $p/2 \in [1/2,1)$, we can use the reverse \emph{Araki-Lieb-Thirring} inequality (see Fact \ref{fact:alt_ineq}) to show that 
\begin{equation}
\label{eqn:alt-intro}
\begin{split}
    \trace{  \Paren{\Paren{\II - zz^\top}^2  \AA^2 \Paren{\II - zz^\top}^2 }^{p/2} } & \geq  \trace{  \Paren{\II - zz^\top} \AA^p \Paren{\II - zz^\top} }  \\
    & = \trace{    \AA^p  } - \trace{   \Paren{ zz^\top}^{p/2}  \Paren{ \AA^2}^{p/2} \Paren{ zz^\top}^{p/2} }\\
    & \geq  \norm{\AA}_{\calS_p}^p - \norm{  \AA z z^\top  }_{\calS_p}^p
\end{split}
\end{equation}
where we use the cyclicity of the trace and again use reverse \emph{Araki-Lieb-Thirring} (Fact \ref{fact:alt_ineq})  to show 
that 
\begin{equation*}
    \trace{   \Paren{ zz^\top}^{\frac{p}{2}}  \Paren{ \AA^2}^{\frac{p}{2}} \Paren{ zz^\top}^{\frac{p}{2}} } \leq \trace{ \Paren{    zz^\top \AA^2  zz^\top }^{p/2}  } = \norm{\AA zz^\top }_{\calS_p}^p.
\end{equation*}
Since we have $\norm{ \AA z z^\top  }_{\calS_p}^p= \norm{  \AA z  }_{2}^p$, we  conclude 
%\begin{equation*}
$    \norm{\AA \Paren{\II - zz^\top} }_{\calS_p}^p \geq \norm{\AA }_{\calS_p}^p - \norm{ \AA z z^\top  }_{2}^p.$
%\end{equation*}
This approach only works for $p \in [1,2)$; for $p > 2$ the application of \emph{Araki-Lieb-Thirring} is  reversed in Equation \ref{eqn:alt-intro} (since $p/2 > 1$, see Fact \ref{fact:alt_ineq}) and we no longer get a lower bound on the cost in Equation \ref{eqn:lb_schatten_cost}. We therefore require a new approach. 

\paragraph{Projection Cost via Norm Compression.} Recall, $z$ is the unit vector output by our candidate low-rank approximation and let $y = \AA z/\norm{\AA z}_2$.
We yet again interpret the input matrix $\AA$ as a partitioned operator by considering the projection of $\AA$ onto $zz^\top$, $y y^\top$ and the projection away from these rank-$1$ subspaces. In particular, let $\II - y y^\top = \YY \YY^\top$, and $\II - z z^\top = \ZZ \ZZ^\top$, where $\YY$ and $\ZZ$ have orthonormal columns. Then, using a rotation argument, we show that
\begin{equation*}
    \norm{ \AA }_{\calS_p} = \norm{\begin{pmatrix}
        y^\top \AA z   &  y^\top \AA \ZZ \\\
         \YY^\top \AA z  & \YY^\top \AA \ZZ
\end{pmatrix}  }_{\calS_p}.
\end{equation*}
We define the $p$-compression of $\AA$, $\CC_{\AA, p}$: 
\begin{equation*}
   \CC_{\AA , p} = \begin{pmatrix}
        \norm{y^\top \AA z}_{\calS_p}  &  \norm{y^\top \AA \ZZ}_{\calS_p} \\\
         \norm{\YY^\top \AA z}_{\calS_p} & \norm{\YY^\top \AA \ZZ}_{\calS_p}
\end{pmatrix}.
\end{equation*}
To relate the norms of $\AA$ and $\CC_{\AA, p}$, we consider 
Audenaert's Norm Compression Conjecture~\cite{audenaert2008norm}, a question in functional analysis concerning operator inequalities (see also  ~\cite{audenaert2012problems}):

\begin{conjecture}[Schatten-$p$ Norm Compression]
\label{conj:norm_compression}
Let $\MM$ be a partitioned operator (block matrix) such that $\MM =\begin{pmatrix}
        \MM_1 & \MM_2 \\
        \MM_3 & \MM_4 \end{pmatrix}$. Let $\CC_{\MM,p} =\begin{pmatrix}
        \norm{ \MM_1}_{\calS_p} & \norm{ \MM_2}_{\calS_p} \\
        \norm{ \MM_3}_{\calS_p} & \norm{ \MM_4}_{\calS_p}\end{pmatrix}$ be a $2 \times 2$ matrix that denotes the Schatten-$p$ compression of $\MM$ for any $p \geq 1$. Then, 
    $\norm{ \MM}_{\calS_p} \geq \norm{ \CC_{\MM,p} }_{\calS_p}$ if $1\leq p \leq 2$,
and
    $\norm{ \MM}_{\calS_p} \leq \norm{ \CC_{\MM,p} }_{\calS_p}$ if $2\leq p < \infty.$
\end{conjecture}

We could simply appeal to this conjecture to obtain that for all $p >2$,
\begin{equation}
\label{eqn:norm-compress-intro}
    \norm{\AA }_{\calS_p}\leq \norm{ \CC_{\AA, p} }_{\calS_p} = \norm{ \begin{pmatrix}
        \norm{y y^\top \AA zz^\top}_{\calS_p}  &  \norm{yy^\top \AA \Paren{I-zz^\top}}_{\calS_p} \\\
         \norm{\Paren{\II - yy^\top} \AA z z^\top}_{\calS_p} & \norm{\Paren{\II- yy^\top} \AA \Paren{\II- zz^\top} }_{\calS_p}\end{pmatrix}}_{\calS_p}.
\end{equation}
However, for our choice of $y$, $\norm{yy^\top \AA \Paren{I-zz^\top}}_{\calS_p}=0$. With padding and rotation arguments, we can then reduce our problem to a block matrix where the blocks in each row are aligned, i.e., each row is a scalar multiple of a fixed matrix (see Lemma \ref{lem:orthogonal_proj_block_matrices}). Then, we can use one of the few special cases of Conjecture \ref{conj:norm_compression} for aligned operators which has actually been proved, and appears in Fact \ref{fact:aligned-norm-compression}. We can thus unconditionally obtain the inequality in Equation \eqref{eqn:norm-compress-intro}.

% Using yet another rotation argument, and setting  $x = \AA z / \norm{\AA z}_2$, where $\AA = \Paren{\II - \WW/5 }$, we can replace the input matrix by the following block matrix:
% \begin{equation*}
% \norm{ \AA }_{\calS_p} = \norm{\begin{pmatrix}
%         xx^\top \AA zz^\top  &  xx^\top \AA \Paren{\II - zz^\top } \\\
%          \Paren{\II - xx^\top }\AA zz^\top  & \Paren{\II - xx^\top }\AA \Paren{\II - zz^\top }
% \end{pmatrix}  }_{\calS_p}    .
% \end{equation*}
% Further, using Equation \eqref{eqn:norm-compress-intro}, we have 
% \begin{equation*}
% \norm{ \AA }_{\calS_p} = \norm{\begin{pmatrix}
%         \norm{ \AA zz^\top }_{\calS_p}  &  0 \\\
%          \norm{\Paren{\II - xx^\top }\AA zz^\top}_{\calS_p}  & \norm{\AA \Paren{\II - zz^\top }}_{\calS_p}
% \end{pmatrix}  }_{\calS_p}    .
% \end{equation*}
% where we use the choice of $x$ to simplify each term. 

Now that we have reduced to the case where we have a $2 \times 2$ matrix with $3$ non-zero entries, we would like to bound its Schatten-$p$ norm. We explicitly compute the singular values of $\CC_{\AA,p}$ (see Fact \ref{fact:singular_vals_2x2} ), and then use the structure of the instance to directly lower bound $\norm{\AA z}_2^p$ as follows:
\begin{equation}
\label{eqn:combined-lowerbound-intro}
    \norm{\AA z }_2^p  +  \Paren{ 1+ \bigO{ \epsilon^{2p/3}}  } \norm{ \AA - \AA_1 }_{\calS_p}^p  \geq \norm{\CC_{\AA,p}}_{\calS_p}^p \geq  \norm{\AA }_{\calS_p}^p ,
\end{equation}
where the last inequality follows from Equation \eqref{eqn:norm-compress-intro}. Since we understand the spectrum of the matrix $\AA$, we can explicitly compute all the terms in \eqref{eqn:combined-lowerbound-intro}  above and show that we can obtain an accurate estimate of the minimum singular value of $\AA$ from $\norm{\AA z}_2^p$. 
See details in Section \ref{sec:lower_bound_p_large}.

\section{Additional Related Work}
% \kcnote{There is discussion of related work above, maybe this should be moved there?}\abnote{the work in the intro is something we are directly comparing to, if we move this up, it takes away space from the techniques}

Existing approaches to solve low-rank approximation problems under several norms fall into two broad categories: iterative methods and linear sketching. Iterative methods, such as Krylov subspace based methods, are captured by the matrix-vector product framework, whereas linear sketching allows for the choice of a matrix $\SS \in \mathbb{R}^{t \times n}$, where $t$ is the number of ``queries'', and then observes the product $\SS \cdot \AA$ and so on (see~\cite{woodruff2014sketching} and references therein).  The model has important applications to streaming and distributed algorithms and several recent works have focused on estimating spectral norms and the top singular values~\cite{andoni2013eigenvalues,li2014sketching,li2016tight,bakshi2021learning}, estimating Schatten and Ky-Fan norms~\cite{li2016tight, li2017embeddings,li2016approximating, braverman2019schatten} and low-rank approximation~\cite{clarkson2013low,mm13,NN13,BDN15,cohen2016nearly}. 

In addition to studying unitarily invariant norms, such as the Schatten norm, there also has been significant amount of work on studying low-rank approximation under matrix $\ell_p$ norms~\cite{song2017low,ban2019ptas,song2020average,mahankali2021optimal} and weighted low-rank approximation~\cite{srebro2003weighted,razenshteyn2016weighted,ban2019regularized}, settings in which the problem is known to be NP-Hard. Finally, there has been a recent flurry of work on sublinear time algorithms for low-rank approximation under various structural assumptions on the input~\cite{mw17,  bakshi2018sublinear,indyk2019sample, sw19,bakshi2020robust}  and in quantum-inspired models~\cite{kerenidis2016quantum,  chia2018quantum, tang2019quantum, rebentrost2018quantum, gilyen2018quantum,   gilyen2019quantum,chepurko2020quantum}. 

\section{Preliminaries}
Given an $n \times d$ matrix $\AA$ with rank $r$, and $n\geq d$, we can compute its 
singular value decomposition, denoted by ${SVD}(\AA) = \UU 
\mathbf{\Sigma} \VV^{\top}$, such that $\UU$ is an $n \times r$ matrix with 
orthonormal columns, $\VV^{\top}$ is an $r \times d$ matrix with orthonormal 
rows and $\mathbf{\Sigma}$ is an $r \times r$ diagonal matrix. The entries 
along the diagonal are the singular values of $\AA$, denoted by 
$\sigma_1, \sigma_2 \ldots \sigma_r$. Given an integer $k \leq r$, we 
define the truncated singular value decomposition of $\AA$ that zeros out 
all but the top $k$ singular values of $\AA$, i.e.,  $\AA_k = \UU 
\mathbf{\Sigma}_k \VV^{\top}$, where $\mathbf{\Sigma}_k$ has only $k$ non-zero 
entries along the diagonal. It is well-known that the truncated SVD 
computes the best rank-$k$ approximation to $\AA$ under any unitarily invariant norm, but in particular for any Schatten-$p$ norm (defined below), we have $\AA_k = \min_{ \rank(\XX)=k  } \| \AA -\XX \|_{\calS_p}$.
More generally, for any matrix $\MM$, we use the notation $\MM_k$ and 
$\MM_{\setminus k}$ to denote the first $k$ components and all but the 
first $k$ components respectively.
% \kcnote{The notation $\MM_k$ is not ``more generally,'' is it? The new thing is $\MM_{\setminus k} = \MM - \MM_k$, right?}\abnote{$\Sigma_k$ was defined only for diagonal matrices before.} 
We use $\MM_{i,*}$ and $\MM_{*,j}$ to 
refer to the $i^{th}$ row and $j^{th}$ column of $\MM$ respectively. 

We use the notation $\II_k$ to denote a \emph{truncated identity matrix}, that is, a square matrix with its top $k$ diagonal entries equal to one, and all other entries zero. The dimension of $\II_k$ will be determined by context.
% \kcnote{This is what you mean, right? I didn't see a definition anywhere.}\abnote{perfect!}

\begin{definition}[Orthogonal Projection Matrices]
\label{def:orthogonal_projection}
Given a $d \times d$ symmetric matrix $\PP$ and $k \in [d]$, $\PP$ is a rank-$k$ orthogonal projection matrix if $\rank(\PP) =k$ and $\PP^2 = \PP$.
\end{definition}
It follows from the above definition that $\PP$ has eigenvalues that are either $0$ or $1$ and admits a singular value decomposition of the form $\UU \UU^\top$ where $\UU$ has $k$ orthonormal columns. 

\begin{definition}[Unitary Matrices]
\label{def:unitary-matrices}
Given a symmetric matrix $\UU \in \mathbb{R}^{d\times d}$ we say  $\UU$ is a unitary matrix  if $\UU^\top \UU = \UU \UU^\top = \II$.
\end{definition}

\begin{definition}[Rotation Matrices]
\label{def:rotation-matrices}
Given a symmetric matrix $\RR \in \mathbb{R}^{d\times d}$ we say  $\RR$ is a rotation matrix  if $\RR$ is unitary and $\textrm{det}\Paren{\RR} = 1$. 
\end{definition}

% \begin{fact}[Product of Unitary Matrices is Unitary]
% \label{fact:product-of-unitaries}
% Let $\MM_1$ and $\MM_2$ be two $d \times d$ unitary matrices. Then, $\MM_1 \cdot \MM_2$ is unitary. 
% \end{fact}

\begin{fact}[Courant-Fischer for Singular Values]
\label{fact:courant-fischer}
Given an $n \times d$ matrix $\AA$ with singular values $\sigma_1 \geq \sigma_2 \geq \ldots \geq  \sigma_d$, the following holds:  for all $i \in [d]$,

\begin{equation*}
    \sigma_i  = \max_{S: \hspace{0.05in} \textrm{dim}(S) = i } \hspace{0.1in} \min_{x \in S:  \hspace{0.05in} \norm{x}_2 = 1} \hspace{0.1in} \norm{ x^\top \AA  }_2.
\end{equation*}
\end{fact}

\begin{fact}[Weyl's Inequality for Singular Values (see Exercise 22~\cite{tao2020notes})]
\label{fact:cauchy}
Given $n \times d$ matrices $\XX, \YY$, for any $i, (j-1) \in [d]$ such that $i +j \leq d$ , 
\[
\sigma_{i+j}\Paren{\XX + \YY } \leq  \sigma_{i}(\XX )  +   \sigma_{j+1}(\YY).
\]
\end{fact}

\begin{fact}[Bernoulli's Inequality]
\label{fact:bernoulli}
For any $x ,  p \in \mathbb{R}$ such that $x\geq -1$ and $p \geq 1$, $\Paren{1+x}^p \geq 1 + px$. 
\end{fact}

\paragraph{Schatten Norms and Trace Inequalities.} We recall some basic facts for Schatten-$p$ norms. We also require the following trace and operator inequalities.

\begin{definition}[Schatten-$p$ Norm]
\label{def:schatten}
Given a matrix $\AA \in \mathbb{R}^{n \times d}$, let $\sigma_1 \geq \sigma_2 \geq \ldots \geq \sigma_d $ be the singular values of $\AA$. Then, for any $p \in [0, \infty)$, the Schatten-$p$ norm of $\AA$ is defined as 
\[
\norm{\AA }_{\calS_p} =\trace{ \Paren{\AA^\top \AA}^{p/2} }^{1/p} = \Paren{ \sum_{i \in [d]} \sigma_i^p(\AA) }^{1/p}.
\]
\end{definition}

\begin{fact}[Schatten-$p$ norms are Unitarily Invariant]
\label{fact:unitary_inv}
Given an $n\times d$ matrix $\MM$, for any $m\times n$ matrix $\UU$ with orthonormal columns, a norm $\|\cdot \|_X$ is defined to be unitarily invariant if $\|\UU\MM \|_X = \|\MM \|_X$. The Schatten-$p$ norm is unitarily invariant for all $p\geq1$. 
\end{fact}

There exists a closed-form expression for the low-rank approximation problem under Schatten-$p$ norms:

\begin{fact}[Schatten-$p$ Low-Rank Approximation]
\label{fact:schtatten_lra}
Given a matrix $\AA \in \mathbb{R}^{n \times d}$ and an integer $k \in \mathbb{N}$, 
\begin{equation*}
    \AA_k = \arg\min_{\rank(\XX) \leq k} \norm{\AA -\XX}_{\calS_p},
\end{equation*}
where $\AA_k$ is the truncated SVD of $\AA$.
\end{fact}

\begin{fact}[Araki–Lieb–Thirring Inequality~\cite{araki1990inequality}]
\label{fact:alt_ineq}
Given PSD matrices $\AA , \BB \in \mathbb{R}^{d\times d}$, for any $r\geq1$, the following inequality holds: 
\[
\trace{ \Paren{ \BB \AA \BB}^{r} } \leq \trace{\BB^r \AA^r \BB^r   }. 
\]
Further, for $0< r < 1$, the reverse holds
\[
\trace{ \Paren{ \BB \AA \BB}^{r} } \geq \trace{\BB^r \AA^r \BB^r   }. 
\]
\end{fact}

\begin{fact}[Mahler's Orthogonal Operator Inequality, Theorem 1.7 in \cite{maher1990some}]
\label{fact:mahler_ortho_ineq}
Given $p \geq 2$, and matrices $\PP$ and $\QQ$ such that the row (column) span of $\PP$ is orthogonal to the row (column) span of $\QQ$, the following inequality holds: 
\[ \norm{\PP }^p_{\calS_p}+\norm{\QQ }^p_{\calS_p} \leq \norm{\PP +\QQ}^p_{\calS_p}. \]
\end{fact}

% \begin{fact}[Mahler's Orthogonal Operator Equality,  \cite{maher1990some}]
% \label{mahler_ortho_eq}
% Given $p \in (0, \infty)$, and matrices $\PP$ and $\QQ$ such that the row span and column span of $\PP$ is orthogonal to the row span and column span of $\QQ$ respectively, the following equality holds: 
% \[ \norm{\PP }^p_{\calS_p}+\norm{\QQ }^p_{\calS_p} =  \norm{\PP +\QQ}^p_{\calS_p}. \]
% \end{fact}

\begin{fact}[H{\"o}lder's Inequality for Schatten-$p$ Norms, Corollary 4.2.6 ~\cite{bhatia2013matrix}]
\label{fact:holder_schatten_p}
Given matrices $\AA , \BB^\top \in \mathbb{R}^{n \times d}$ and $p \in [1, \infty)$, the following holds 
\[
\norm{\AA \BB }_{\calS_p} \leq \norm{\AA }_{\calS_q} \cdot \norm{\BB }_{\calS_r},
\]
for any $q,r$ such that $\frac{1}{p} = \frac{1}{q} + \frac{1}{r}$.
\end{fact}

We also require \emph{pinching inequalities} that were originally introduced to relate norms for partitioned operators over direct sums of Hilbert spaces. In our context, these inequalities simplify to norm inequalities for block matrices: 

% Let $\CC(\MM) \in \mathbb{R}^{dt \times dt}$ be the following matrix:

% \begin{equation*}
%     \CC(\MM) = \left[\begin{array}{cccc}
% \MM_{(1,1)}&& &\\
% &\MM_{(2,2)}& &\\
%  & &\ddots & \\
% && &\MM_{(t,t)}\\
% \end{array}\right].
% \end{equation*}

\begin{fact}[Pinching Inequalities for Schatten-$p$ Norms, \cite{bhatia2002pinchings}]
\label{fact:pinching}
Let $\MM  \in \mathbb{R}^{td \times td}$ be the following block matrix
\begin{equation*}
    \MM = \left[\begin{array}{cccc}
\MM_{(1,1)}&\MM_{(1,2)}&\cdots &\MM_{(1,t)}\\
\MM_{(2,1)}&\MM_{(2,2)}&\cdots &\MM_{(1,t)}\\
\vdots & &\ddots &\vdots \\
\MM_{(t,1)}&\MM_{(t,2)}&\cdots &\MM_{(t,t)}\\
\end{array}\right] ,
\end{equation*}
where for all $i, j \in [t]$, $\MM_{(i,j)}  \in \mathbb{R}^{d \times d}$. 
For all $p \geq 1$, the following inequality holds:
\begin{equation*}
    \Paren{ \sum_{i \in [t]} \norm{ \MM_{(i,i)}}_{\calS_p}^p   }^{1/p} \leq \norm{\MM }_{\calS_p}.
\end{equation*}
\end{fact}

% We consider the following special case of Audenaert's Norm Compression Conjecture~\cite{audenaert2008norm}, a question in functional analysis concerning operator inequalities (see ~\cite{audenaert2012problems} for known special cases and implications):

% \begin{conjecture}[Schatten-$p$ Norm Compression]
% \label{conj:norm_compression}
% Let $\MM$ be a partitioned operator (block matrix) such that $\MM =\begin{pmatrix}
%         \MM_1 & \MM_2 \\
%         \MM_3 & \MM_4 \end{pmatrix}$. Let $\CC_{\MM,p} =\begin{pmatrix}
%         \norm{ \MM_1}_{\calS_p} & \norm{ \MM_2}_{\calS_p} \\
%         \norm{ \MM_3}_{\calS_p} & \norm{ \MM_4}_{\calS_p}\end{pmatrix}$ be a $2 \times 2$ matrix that denotes the Schatten-$p$ compression of $\MM$ for any $p \geq 1$. Then, 
% \begin{equation*}
%     \norm{ \MM}_{\calS_p} \geq \norm{ \CC_{\MM,p} }_{\calS_p}  \hspace{0.5in} \textrm{if } 1\leq p \leq 2,
% \end{equation*}
% and
% \begin{equation*}
%     \norm{ \MM}_{\calS_p} \leq \norm{ \CC_{\MM,p} }_{\calS_p}  \hspace{0.5in} \textrm{if } 2\leq p < \infty.
% \end{equation*}
% \end{conjecture}

We also require a norm compression inequality that is a special case of Conjecture \ref{conj:norm_compression} (and known to be true), when each block is aligned in the following sense: 

\begin{fact}[Aligned Norm Compression Inequality, Section 4.3 in \cite{audenaert2008norm}]
\label{fact:aligned-norm-compression}
Let $\MM = \begin{pmatrix}\MM_1 & \MM_2 \\
\MM_3 & \MM_4\end{pmatrix}$ such that there exist scalars $\alpha_1, \alpha_2,\beta_1 , \beta_2$ such that $\MM_1 = \alpha_1 \XX$, $\MM_2 = \alpha_2 \XX$, $\MM_3 = \beta_1 \YY $and $\MM_4 = \beta_2 \YY$. Then, for any $p \geq 2$,
\begin{equation*}
    \norm{\MM}_{\calS_p} \leq \norm{\begin{pmatrix} \norm{ \MM_1}_{\calS_p} & \norm{ \MM_2 }_{\calS_p} \\
\norm{\MM_3}_{\calS_p} & \norm{\MM_4}_{\calS_p}\end{pmatrix}}_{\calS_p}.
\end{equation*}
\end{fact}

% \abnote{check if the following thm is used anywhere. }

% \begin{theorem}(Generalized Low-Rank Approximation \cite{friedland2007generalized}.)
% \label{thm:generalized_low_rank_approximation}
% Given $k' \geq k\in\mathbb{N}$, let $\AA$ be an $n \times n$ matrix, and $\BB, \CC^\top$ be  $n \times k'$ matrices. Then, the Generalized Low-Rank Approximation problem is defined as
% \begin{equation*}
%     \min_{\textrm{rank}(\XX)= k} \|\AA - \BB \XX \CC \|^2_F
% \end{equation*}
% and minimized by $\XX = \BB^{\dagger}[\PP_B \AA \PP_C]_k \CC^{\dagger}$, where $\PP_{B},\PP_C$ are the projection matrices onto $\BB$ and $\CC$ respectively.   
% \end{theorem}

\paragraph{Random Matrix Theory.} Next, we recall some basic facts for Wishart ensembles from random matrix theory (we refer the reader to \cite{tao2012topics} for a comprehensive overview). 

\begin{definition}[Wishart Ensemble]
\label{def:wishart_ensemble}
An $n \times n$ matrix $\WW$ is sampled from a Wishart Ensemble, $\textsf{Wishart}(n)$, if $\WW = \XX \XX^{\top}$ such that for all $i,j \in[n]$ $\XX_{i,j} \sim \mathcal{N}\Paren{0, \frac{1}{n} \II}$. 
\end{definition}

\begin{fact}[Norms of a Wishart Ensemble]
\label{fact:norms_of_wishart}
Let $\WW\sim \textsf{Wishart}(n)$ such that $n = \Omega(1/\eps^3)$. Then, with probability $99/100$, $\norm{\WW}_{\textrm{op}} \leq 5$ and for any fixed constant $p$,  $\norm{\II - \frac{1}{5}\WW}^p_{\calS_p} = \Theta\Paren{ \frac{1}{\eps^{1/3}}}$.
\end{fact}

% \kcnote{I thought that 5, like 99/100, was somewhat arbitrary, but that F-norm bound is somehow mysterious if that's so. The diagonal entries of $\WW$ are (somewhat) concentrated around *5*?} \abnote{the constant is 4, 5 for good measure i guess.}
% \kcnote{I still don't understand the magic number 5 here in $\II - \frac{1}{5}\WW$. It's not enough for the magic number just to be big enough, is it? Could $\frac{1}{100}$ be used here? I could understand $\frac{1}{\norm{\WW}_{\textrm{op}}}$, I guess. }

% \abnote{this comes from part one above, $\|\WW\|_op \leq 5$, we need $\II - 1/c \WW$ to be PSD.}

\section{Algorithms for Schatten-$p$ LRA}
\label{sec:upper_bound}

In this section, we focus on obtaining algorithms for low-rank approximation in Schatten-$p$ norm, simultaneously for all real, not necessarily constant,  $p \in [1, \bigO{\log(d)/\epsilon}]$. For the special case of $p \in \{ 2, \infty\}$, Musco and Musco~\cite{musco2015randomized} showed an algorithm with matrix-vector query complexity 
$\tilde{O}(k/\epsilon^{1/2})$, given below as Algorithm~\ref{algo:simul_power_iter}. 
We show that the number of  matrix-vector products we require scales proportional to $\tilde{O}\Paren{kp^{1/6}/\epsilon^{1/3}}$ instead. Finally, recall when $p > \log(d)/\epsilon$, it suffices to run Block Krylov for $p =\infty$, which requires $\bigO{\log(d/\epsilon)k/\sqrt{\epsilon}}$ matrix-vector products.

% \kcnote{Given the emphasis on m-v products, shouldn't the runtime below be put in that form?}\abnote{done.}

\begin{theorem}[Optimal Schatten-$p$ Low-Rank Approximation]
\label{thm:optimal_schatten_p_lra}
Given a matrix $\AA \in \mathbb{R}^{n \times d}$, a target rank $k \in [d]$, an accuracy parameter $\epsilon \in (0,1)$ and any $p \in [1, \bigO{\log(d)/\epsilon}]$, Algorithm \ref{algo:optimal_schatten_p_lra}
performs $\bigO{\frac{k p^{1/6} \log(d/\epsilon)}{\epsilon^{1/3}} + \log(d/\epsilon)k \sqrt{p}}$ matrix-vector products and outputs a $d \times k$ matrix $\ZZ$ with orthonormal columns such that with probability at least $9/10$,
\begin{equation*}
    \norm{ \AA \Paren{\II - \ZZ \ZZ^\top}  }_{\calS_p} \leq \Paren{ 1+ \epsilon } \min_{\UU : \hspace{0.05in} \UU^\top \UU = \II_k} \norm{ \AA \Paren{\II - \UU \UU^\top} \ }_{\calS_p}.
\end{equation*}
Further, in the RAM model, the algorithm runs in time $\mathcal{O}\Paren{\frac{\nnz(\AA) p^{1/6} k\log^2(d/\epsilon)}{\epsilon^{1/3}} + \frac{np^{(\omega-1)/6} k^{\omega-1}}{\epsilon^{(\omega-1)/3}} }$.
\end{theorem}

We first introduce the following lemmas from Musco and Musco \cite{musco2015randomized} that provide convergence bounds for the performance of Block Krylov Iteration (Algorithm \ref{algo:simul_power_iter}) :

\begin{lemma}[Gap Independent Block Krylov with Arbitrary Accuracy]
\label{lem:gap_independent}
Let $\AA$ be an $n \times d$ matrix, $k$ be the target rank and $\gamma>0$ be an accuracy parameter.  Then, initializing Algorithm \ref{algo:simul_power_iter} with block size $k$ 
% \kcnote{The algorithm does its own initialization of $\UU$; the upshot is that the block size $s$ of the algorithm is set to $k$.} 
%\abnote{fixed}
and running for  $q = \Omega\Paren{\log(d/\gamma)/\sqrt{\gamma}}$ iterations outputs a $d \times k$ matrix $\ZZ$ such that with probability $99/100$, for all $i \in [k]$,
\begin{equation*}
    \norm{\AA \ZZ_{*,i} }^2_2 = \sigma_i^2 \pm \gamma \sigma^2_{k+1}.
\end{equation*}
Further, the total number of matrix-vector products is $\bigO{kq}$ and the running time in the RAM model is $\bigO{\nnz(\AA)kq + n \Paren{kq}^2 + \Paren{kq}^{\omega} }$. 
\end{lemma}
The aforementioned lemma follows directly from Theorem 1 in \cite{musco2015randomized},  using the per-vector error guarantee (3).

\begin{lemma}[Gap Dependent Block Krylov, Theorem 13 \cite{musco2015randomized}]
\label{lem:gap_dependent}
Let $\AA$ be an $n \times d$ matrix and $\gamma>0$, be an accuracy parameter and $p, k\in \mathcal{N}$ be such that $b\geq k$.
% \kcnote{Not the Schatten $p$, but a different $p$.}\abnote{fixed.} 
Let $\sigma_1,\sigma_2 \ldots \sigma_d$ be the singular values of $\AA$. Then, initializing Algorithm \ref{algo:simul_power_iter} with block size $b$ and running for $q = \Omega\Paren{\log(n/\gamma)\sqrt{\sigma_k} /\sqrt{\sigma_k - \sigma_b}} $ iterations outputs a $d \times k$ matrix $\ZZ$ such that with probability $99/100$, for all $i \in [k]$
\begin{equation*}
    \norm{\AA \ZZ_{*,i} }^2_2 = \sigma^2_i \pm \gamma \sigma^2_{k+1}.
\end{equation*}
Further, the total number of  matrix-vector products is $\bigO{p q}$ and the running time in the RAM model is $\bigO{\nnz(\AA) b q + n\Paren{bq}^2 + \Paren{bq}^{\omega} }$. 
\end{lemma}

\begin{mdframed}
  \begin{algorithm}[Optimal Schatten-$p$ Low-rank Approximation]
    \label{algo:optimal_schatten_p_lra}\mbox{}
    \begin{description}
    \item[Input:] 
    An $n \times d$ matrix $\AA$, target rank $k \leq d$, accuracy parameter $0<\eps<1$, and $p\geq1$.
    
    %\item\mbox{}
    % \begin{enumerate} 
    % \item If $p \leq \log(1/\epsilon)/(6\log(c)) $, for a sufficiently large fixed constant $c$, 
    \begin{enumerate}
    
    \item   Let   $\gamma_1 = \eps^{2/3} / p^{1/3}$. 
    Run Block Krylov Iteration (Algorithm \ref{algo:simul_power_iter}) on $\AA^\top$ with block size $k$,
    and number of iterations $q =\bigO{ \log(d/\gamma_1)/\sqrt{\gamma_1} + \log(d/\epsilon) \sqrt{p}}$. Let $\WW_1 \in \mathbb{R}^{n \times k} $ be the corresponding output with orthonormal columns.
    \item Let   $\gamma_2 = \eps$ and let $s = \bigO{ p^{-1/3} k/\eps^{1/3} }$. Run Block Krylov Iteration (Algorithm \ref{algo:simul_power_iter}) on $\AA^\top$ with block size $s$, and number of iterations $q = \bigO{\log(d/\gamma_2) \sqrt{p}}$. Let $\WW_2 \in \mathbb{R}^{n \times k}$ be the corresponding output with orthonormal columns.
    \item  Run Block Krylov on $\AA$ with target rank $k+1$ and number of iterations $q = \bigO{ (\log(dp) + \log(d/\epsilon))\sqrt{p}  }$, and let $\hat{\ZZ}_1$ be the resulting $d\times (k+1)$ output matrix. Compute $\hat{\sigma}^2_{1} = \norm{\AA ( \hat{\ZZ}_1)_{*,1} }_2^2$ and  $\hat{\sigma}^2_{k+1} = \norm{\AA ( \hat{\ZZ}_1)_{*,k+1} }_2^2$, rough estimates of the $1$-st and $(k+1)$-st singular values of $\AA$. Run Block Krylov on $\AA$ with target rank $s$, where $s=\bigO{ p^{-1/3} k/\eps^{1/3}}$ and iterations $q=\bigO{\log(d/\epsilon)\sqrt{p}}$, and let  $\hat{\ZZ}_2$ be the resulting $d\times s$ output matrix. Compute $\hat{\sigma}^2_{s} = \norm{\AA ( \hat{\ZZ}_2)_{*,s} }_2^2$, an estimate to the $s$-th singular value of $\AA$.   
    \item 
    If $\hat{\sigma}_{1}^2 \geq (1+0.5/p) \hat{\sigma}_{k+1}^2$, set $\ZZ = \ZZ_1$. Else, 
    if $\hat \sigma^2_s\leq  \hat \sigma^2_{k+1} / \Paren{1+0.5/p}$, set $\ZZ$ to be an orthonormal basis for $\AA^\top \WW_2 \WW_2^\top$ and otherwise set $\ZZ$ to be an orthonormal basis for $\AA^\top \WW_1 \WW_1^\top$.  
    \end{enumerate}
    
    % \item Else, let $\UU$ be a $d \times k$ matrix such that each entry is drawn i.i.d. from $\mathcal{N}(0,1)$. Let $\KK = \left[ \AA \UU ; (\AA \AA^\top) \AA \UU; (\AA \AA^\top)^2 \AA \UU; \ldots; (\AA \AA^\top)^q \AA \UU\right]$ be the $n \times k(q+1)$  Krylov matrix obtained by concatenating the matrices $\AA\UU, \ldots, \Paren{\AA\AA^\top}^q \UU$, for $q =\bigO{ \log(d/\epsilon)/\sqrt{\epsilon}}$. Let $\MM = p\Paren{\AA^\top \AA } \UU$ where $p = \frac{\Paren{1+\epsilon} \Paren{1-\epsilon \log(d)} T_q\Paren{\AA^\top \AA /\Paren{1-\epsilon \log(d)}} }{T_q\Paren{1+\epsilon}}$. Let $\QQ$ be an orthonormal basis for $p\Paren{\AA^\top \AA}\UU$.     
    
    % \end{enumerate} 
    \item[Output:] A matrix $\ZZ \in \mathbb{R}^{d \times k}$  with orthonormal columns such that \[
    \norm{\AA \Paren{\II -\ZZ\ZZ^\top }}_{\calS_p}^p \leq \Paren{1+\epsilon } \min_{\UU : \hspace{0.05in} \UU^\top \UU =\II_k} \norm{\AA \Paren{\II - \UU \UU^\top}}_{\calS_p}^p .\] 
    \end{description}
  \end{algorithm}
\end{mdframed}

Next, we prove the following key lemma relating the Schatten-$p$ norm of row and column projections applied to a matrix $\AA$ to the Schatten-$p$ norm of the matrix itself. We can interpret this lemma as an extension of the Pythagorean Theorem to Schatten-$p$ spaces and believe this lemma is of independent interest. We note that we appeal to \emph{pinching inequality} for partitioned operators to obtain this lemma.

\begin{lemma}[Schatten-$p$ Norms for Orthogonal Projections]
\label{lem:schatten_p_orthogonal_projections}
Let $\AA$ be an $n \times d$ matrix, let $\PP$ be an $n \times n$ matrix, and let $\QQ$ be a $d\times d$ matrix such that both $\PP$ and $\QQ$ are orthogonal projection matrices of rank $k$ (see Definition \ref{def:orthogonal_projection}). Then, the following inequality holds for all $p \geq 1$:
\begin{equation*}
    \norm{\AA }_{\calS_p}^p \geq \norm{ \PP \AA \QQ }_{\calS_p}^p +  \norm{ \Paren{ \II -\PP }\AA \Paren{\II - \QQ }  }_{\calS_p}^p.
\end{equation*}
\end{lemma}
\begin{proof}
Let $\AA = \UU \Sig \VV^\top$ be the SVD of $\AA$, where $\UU \in \mathbb{R}^{n \times d}$ and $\VV^\top \in \mathbb{R}^{d \times d}$ have orthonormal columns and rows respectively.
We construct unitary matrices $\RR$ and $\SS$, such that $\RR \in \mathbb{R}^{n \times n} $ and $\SS \in \mathbb{R}^{d\times d}$ 
that satisfy the following constraints: 
\begin{enumerate}
    \item $\RR^\top \II_k \RR \AA \SS^\top \II_k \SS = \PP \AA \QQ$, and
    \item $\RR^\top \Paren{\II - \II_k} \RR \AA \SS^\top \Paren{\II - \II_k} \SS = \Paren{\II - \PP} \AA \Paren{\II - \QQ}$,
\end{enumerate}
where the trunctated Identity matrix, $\II_k$, left multiplying $\AA$ is $n \times n$ and right multiplying $\AA$ is $d \times d$.

Recall, since $\PP$ is a rank-$k$ projection matrix, it admits a decomposition $\PP = \XX \XX^\top $ such that $\XX$ has $k$ orthonormal columns and similarly $\II - \PP = \YY \YY^\top$, where $\YY$ has $n -k$ orthonormal columns. Further, since $\XX$ and $\YY$ span disjoint subspaces, and the union of their span is $\mathbb{R}^n$, the matrix $\Paren{  \XX \mid \YY }$, obtained by concatenating their columns, is unitary. Then, it suffices to set $\RR = \Paren{  \XX \mid \YY }^\top$. To see this, observe,
\begin{equation*}
    \RR^\top \II_k \RR = \Paren{ \XX \mid 0 } \cdot \begin{pmatrix}
       \XX^\top \\
        0 \end{pmatrix}  = \XX \XX^\top=  \PP, 
\end{equation*}
and similarly, 
\begin{equation*}
    \RR^\top \Paren{\II - \II_k} \RR = \YY \YY^\top = \II - \PP.
\end{equation*}

We repeat the above argument for the projection matrix $\QQ$. Let $\QQ = \WW \WW^\top$, where $\WW$ is $d \times k$ and has orthonormal columns, and $\II - \QQ = \ZZ \ZZ^\top$, where $\ZZ$ is $d \times (d-k)$ and has orthonormal columns. Observe, it suffices to set $\SS = \Paren{\WW \mid \ZZ }^\top$, since $\SS$ is unitary and $\SS^\top \II_k \SS = \QQ $ and $\SS^\top \Paren{\II - \II_k} \SS =\II - \QQ$. Note, by construction, we satisfy the two aforementioned constraints.
% \kcnote{The constraint above holds if $\SS \II_k \SS^\top=\QQ$, but here
% it's $\SS^\top \II_k \SS = \QQ $. This needs to be made consistent throughout, I think.}\abnote{yeah, i think i flipped it in the requirement of the constraints, fixed. }

Let $\hat{\AA} = \RR \AA \SS^\top$. Since $\RR$ and $\SS$ are unitary, it follows from unitary invariance of the Schatten-$p$ norm that 
% \kcnote{Multiplying $\AA$ or $\hat{\AA}$? Here $\II_k$ is $n\times n$ but with $k$ 1's on the diagonal? If $\II_k$ can be $n\times n$ or $d\times d$,
% what is the $k$ mean?}
% \abnote{$\hat A$ is a $n \times d$ matrix, I think I mentioend we left multiply be $n \times n$ I and right multiply by $d \times d$ I, and use the notation $\II_k$ to denote a square matrix where only the first $k$ diagonal entries are $1$. }
% \kcnote{I don't think you did mention what $\II_k$ means, but if that's what it means, then it isn't (always) an identity matrix. And: the $\II_k$ are multiplying $\hat{\AA}$, not $\AA$, right?}
\begin{equation}
\label{eqn:schatten_p-rotation}
    \norm{\hat{\AA}  }_{\calS_p} = \norm{ \RR \UU \Sig\VV^\top \SS^\top }_{\calS_p} = \norm{ \AA }_{\calS_p}
\end{equation}
Further, observe for any $n\times d$ matrix $\MM$, we have have the following block decomposition
\begin{equation*}
\begin{split}
    \MM & = \II_k\MM \II_k  + \II_k \MM \Paren{ \II - \II_k } + \Paren{ \II - \II_k }\MM \II_k  + \Paren{ \II - \II_k }\MM \Paren{ \II - \II_k }\\
    & = \begin{pmatrix}
    \MM_{1:k, 1:k} & \MM_{1:k, k+1 : d} \\
    \MM_{k+1:n, 1:k} & \MM_{k+1:n, k+1:d} 
    \end{pmatrix},
\end{split}
\end{equation*}
where the notation $\MM_{i:i',  j: j'}$ picks the $(i'-i+1) \times (j'-j+1)$ sized sub-matrix corresponding to the rows indices $[i,i']$ and column indices $[j,j']$.  
Since appending rows and columns of $0$'s does not change the singular values, we have $\norm{\II_k \MM \II_k }_{\calS_p} = \norm{ \MM_{1:k, 1:k} }_{\calS_p}$ and $\norm{\Paren{\II - \II_k} \MM \Paren{\II - \II_k} }_{\calS_p} = \norm{ \MM_{k+1:n, k+1:d} }_{\calS_p}$. Setting $\MM= \hat{\AA}$, we have

\begin{equation}
\label{eqn:block_matrix_pinching}
\begin{split}
    \norm{ \hat{\AA} }_{\calS_p}^p
    & = \norm{  \begin{pmatrix}
        \hat{\AA}_{1:k, 1:k} & \hat{\AA}_{1:k, k+1 : d} \\
    \hat{\AA}_{k+1:n, 1:k} & \hat{\AA}_{k+1:n, k+1:d}  
    \end{pmatrix}}_{\calS_p}^p\\
    & \geq \norm{\hat{\AA}_{1:k, 1:k}}_{\calS_p}^p + \norm{\hat{\AA}_{k+1:n, k+1:d}    }_{\calS_p}^p\\
    & = \norm{\II_k\hat{\AA}\II_k  }_{\calS_p}^p + \norm{\Paren{ \II - \II_k } \hat{\AA} \Paren{ \II - \II_k }   }_{\calS_p}^p,
\end{split}
\end{equation}
where the inequality follows from using the \emph{ pinching inequality} on the block matrix (see Fact \ref{fact:pinching}).  By the unitary invariance of the Schatten-$p$ norm, we have 
\begin{equation*}
    \norm{\II_k\hat{\AA}\II_k  }_{\calS_p}^p = \norm{\RR^\top  \II_k\hat{\AA}\II_k \SS  }_{\calS_p}^p = \norm{\PP \AA \QQ }_{\calS_p}^p, 
\end{equation*}
and similarly,
\begin{equation*}
   \norm{\Paren{ \II - \II_k } \hat{\AA} \Paren{ \II - \II_k }   }_{\calS_p}^p = \norm{\RR^\top \Paren{ \II - \II_k } \hat{\AA} \Paren{ \II - \II_k } \SS  }_{\calS_p}^p = \norm{\Paren{\II - \PP }\AA \Paren{\II - \QQ}  }_{\calS_p}^p .
\end{equation*}
Plugging these two bounds back into Equation \eqref{eqn:block_matrix_pinching}, along with Equation \eqref{eqn:schatten_p-rotation}, we can conclude,
\begin{equation*}
    \norm{\AA }_{\calS_p}^p \geq \norm{ \PP \AA \QQ }_{\calS_p}^p +  \norm{ \Paren{ \II -\PP }\AA \Paren{\II - \QQ }  }_{\calS_p}^p.
\end{equation*}

\end{proof}

\begin{mdframed}
  \begin{algorithm}[Block Krylov Iteration, \cite{musco2015randomized}]
    \label{algo:simul_power_iter}\mbox{}
    \begin{description}
    \item[Input:] An $n \times d$ matrix $\AA$, target rank $k$,  iteration count $q$ and a block size parameter $s$ such that $k\leq s \leq d$.  
    
    %\item\mbox{}
    \begin{enumerate}
    \item Let $\UU$ be a $n \times s$ matrix such that each entry is drawn i.i.d. from $\mathcal{N}(0,1)$. Let $\KK = \left[ \AA^\top \UU ; (\AA^\top \AA) \AA^\top \UU; (\AA^\top \AA)^2 \AA^\top \UU; \ldots; (\AA^\top \AA)^q \AA^\top \UU\right]$ be the $d \times s(q+1)$  Krylov matrix obtained by concatenating the matrices $\AA^\top\UU, \ldots, \Paren{\AA^\top \AA}^q \AA^\top \UU$. 
    \item Compute an orthonomal basis $\QQ$ for the column span of $\KK$. Let $\MM = \QQ^\top \AA^\top \AA \QQ$. 
    \item Compute the top $k$ left singular vectors of $\MM$, and denote them by $\YY_k$.
    \end{enumerate}
    \item[Output:] $\ZZ = \QQ \YY_k$ 
    \end{description}
  \end{algorithm}
\end{mdframed}

Note, despite establishing Lemma \ref{lem:schatten_p_orthogonal_projections}, it is not immediately apparent how to lower bound $\norm{\AA \ZZ \ZZ^\top}_{\calS_p}^p$, where $\ZZ$ is a candidate solution. Next, we show how to translate a guarantee on the Euclidean norm of $\AA$ times a column of $\ZZ$ to a lower bound on $\norm{\AA \ZZ \ZZ^\top}_{\calS_p}^p$. 
% We show that if we carefully choose the projection $\PP$
% % \kcnote{The projection $\QQ$ here?} \abnote{both $\PP$ and $\QQ$ are projections, trying to say you need to pick $\PP$ and $\QQ$ carefully.} 
% such that $\QQ$ is correlated with the top-$k$ subspace of $\AA$, and $\PP$ is the projection on the span of $\AA \QQ$, then we can lower bound the Schatten-$p$ norm of $\PP \AA \QQ$ purely in terms of the singular values of $\AA$.  

\begin{lemma}[Per-Vector Guarantees to Schatten Norms]
\label{lem:correlated_projections}
Let $\AA$ be an $n\times d$ matrix with singular values denoted by $\{ \sigma_i\Paren{\AA } \}_{i \in [d]}$.
Let $\ZZ$ be a $d \times k$ matrix with orthonormal columns that is output by Algorithm  \ref{algo:simul_power_iter}, such that for all $i \in [k]$, with probability at least 99/100,  $ \norm{\AA\ZZ_{*,i} }_2^2 \geq \sigma_{i}^2\Paren{\AA} - \gamma_i  \sigma_{k+1}^2\Paren{\AA}$, for some $\gamma \in (0,1)$.
%\kcnote{I hope this sentence is a suitable replacement for the following commented out.}
% Let $\UU$ be a $d \times k$ matrix with orthonormal columns such that for all $i \in [k]$, $ \norm{\AA \UU_{*,i} }_2^2 \geq \sigma_{i}^2 - \gamma \sigma_{k+1}^2$, for some $\gamma \in (0,1)$.
% Let $\WW$ be a $n \times k$ matrix such that for all $i \in [k]$, $\WW_{i,*} = \AA \UU_{*,i}$ \kcnote{So $\WW=\AA\UU$ and $\norm{\WW_{*,i}}\ge...$?} \abnote{why do we need norms on rows of $\WW$? }
 Then, for any $p\geq1$, we have
\begin{equation*}
    \norm{ \AA \ZZ \ZZ^\top }_{\calS_p}^p \geq \norm{  \AA_k }_{\calS_p}^p - \sum_{i \in [k]}  \bigO{ \gamma_i p } \sigma_{k+1}^2\Paren{\AA} \sigma_{i}^{p-2}\Paren{\AA}.
\end{equation*}
\end{lemma}

% \kcnote{It makes me nervous to see $\bigO{}$ inside a sum. It's coming from the Bernoulli inequality, so it seems like $\bigO{ \gamma p }$ could be just $\gamma p$.} \abnote{yeah, you're right, but i like to introduce slack, i can get rid of the big-O but that makes me nervous haha  }
\begin{proof}
First, 
% since the Schatten-$p$ norm is unitarily invariant $\norm{\VV \VV^\top \AA \ZZ \ZZ^\top}_{\calS_p}^p= \norm{\VV \VV^\top \AA \ZZ }_{\calS_p}^p$.
% Next, 
% by definition, the column span of $\VV$ contains the column vectors of $\AA\ZZ$, and thus $\VV\VV^\top\AA\ZZ  = \AA \ZZ$.\kcnote{$\VV$ is mentioned nowhere else in lemma statement or proof.} Next,
we observe that it suffices to show that $\sigma_i(\AA \ZZ)^2 \geq \norm{\AA z_i}_2^2$, where $z_i $ is shorthand for $\ZZ_{*,i}$, the $i$-th columm of $\ZZ$.
% \kcnote{$u_i=\UU_{i,*}$?}\abnote{no columnn of $\UU$, defined now}
Assuming this inequality holds, we can complete the proof as follows:  we know that  for all $i \in [k]$,
\begin{equation}
\label{eqn:gap-free-guarantee-with-sing-val}
\begin{split}
    \sigma_i^2(\AA \ZZ)  \geq \norm{\AA z_i}_2^2 & \geq \sigma_i^2(\AA) - \gamma \sigma_{k+1}^2(\AA) \\
    & = \sigma_i^2(\AA)\Paren{ 1 - \gamma \frac{\sigma_{k+1}^2(\AA)}{\sigma_i^2(\AA)} }
\end{split}
\end{equation}
Then,
% \kcnote{I take it that all the $\sigma_i^2(\AA)$ could be $\sigma_i^2$, etc.} \abnote{that's right, but tedious, made everything the prior consistently.} 
taking $p/2$-th powers in \eqref{eqn:gap-free-guarantee-with-sing-val},
\begin{equation}
\label{eqn:intermediate-bound-on-singular-values}
\begin{split}
    \sigma_i^p(\AA \ZZ) & \geq  \sigma_i^p(\AA)\Paren{ 1 - \gamma \frac{\sigma_{k+1}^2(\AA)}{\sigma_i^2(\AA)} }^{p/2} \\
    & \geq \sigma_i^p(\AA)\Paren{ 1 - \bigO{  \frac{\gamma p \sigma_{k+1}^2(\AA)}{\sigma_i^2(\AA)} } } \\
    & = \sigma_i^p(\AA) - \bigO{ \gamma p } \sigma_{k+1}^2\Paren{\AA} \sigma_{i}^{p-2}\Paren{\AA}
\end{split}
\end{equation}
where the second inequality follows from the generalized Bernoulli inequality (see Fact \ref{fact:bernoulli}). Summing over all $i \in [k]$, we can conclude 
\begin{equation*}
      \norm{ \AA \ZZ }_{\calS_p}^p \geq \norm{\AA_k }_{\calS_p}^p - \sum_{i \in [k]} \bigO{ \gamma p } \sigma_{k+1}^2\Paren{\AA} \sigma_{i}^{p-2}\Paren{\AA}.
\end{equation*}
% \kcnote{Nothing in the lemma statement references Algorithm \ref{algo:simul_power_iter}. I take it that the lemma applies to $\ZZ$ output by the algorithm, except that $\ZZ$ is called here $\UU$, where $\UU$ is not the Gaussian matrix of the algorithm (or the left factor of the SVD...)? }
% \abnote{fixed. }
Therefore, it remains to show that $\sigma_i(\AA \ZZ)^2 \geq \norm{\AA z_i}_2^2$. First, we recall that Algorithm \ref{algo:simul_power_iter} outputs $\{ z_i \}_{i \in [k]}$ such that $z_i = \QQ \tilde{z}_i$, 
% \kcnote{Lemma statement refers to Alg \ref{algo:optimal_schatten_p_lra}.}
% \abnote{yeah, it can just consistently be the second one, we run different copies of Block Krylov in any case. }
% \kcnote{The $u_i$ are columns of $\YY_k$?}\abnote{yea, these are equivalent to the singular vecs of $\QQ^\top \AA\AA^\top \QQ$}
where $\QQ$ is an orthonormal basis for the Krylov space $\KK$ (an $d \times s(q+1)$ matrix) and $\tilde{z}_i$ is the $i$-th singular vector of $\QQ^\top \AA^\top \AA  \QQ$. Note that the $\tilde{z}_i$'s are $s(q+1)$-dimensional vectors. Let $\WW \mathbf{\Omega} \WW^\top$ be the SVD of $\QQ^\top \AA^\top \AA  \QQ$. Then, $\QQ \WW \mathbf{\Omega} \WW^\top\QQ^\top$ is the SVD of $\QQ \QQ^\top \AA^\top \AA \QQ \QQ^\top$. To see this, let the $i$-th column of $\QQ \WW$ be denoted by $\QQ \WW_{*,i}$. Then, 
% \kcnote{$(\QQ\WW)^\top \QQ\WW = \WW^\top\QQ^\top\QQ\WW = \WW^\top\WW = I$?} \abnote{umm $(\QQ\WW)^\top \QQ\WW = \WW^\top \QQ^\top \QQ \WW$ and $\QQ \QQ^\top$ is a $k \times k$ idenntity yea}\kcnote{So that's simpler than the following two equations..}
\begin{equation*}
    \left\langle \QQ \WW_{*,i} , \QQ \WW_{*,i} \right\rangle = \WW_{*,i}^\top \QQ^\top \QQ \WW_{*,i} = 1
\end{equation*}
and similarly for any $j\neq i$, 
\begin{equation*}
    \left\langle \QQ \WW_{*,i} , \QQ \WW_{*,j} \right\rangle = \WW_{*,i}^\top \QQ^\top \QQ \WW_{*,j} = 0
\end{equation*}
where we use that $\QQ^\top \QQ = \II$ and the columns of $\WW$ are orthonormal, which holds by definition. Therefore, $z_i = \QQ \tilde{z}_i$ is the $i$-th singular vector
% \kcnote{$u_i$ is a value? A real number?}\abnote{vector, typo thanks!}
of $\QQ\QQ^\top \AA^\top \AA \QQ \QQ^\top$. Let $\tilde{\ZZ}$ be the matrix obtained by stacking the vectors $\tilde{z}_i$ together. Then, we have 
\begin{equation}
    \begin{split}
        \sigma_i(\AA \ZZ)^2 = \sigma_i^2(\AA \QQ \tilde{\ZZ} ) & = \sigma_i^2(\AA \QQ ) \\
        & = \sigma_i^2(\AA \QQ \QQ^\top )\\
        &= z_i^T \QQ \QQ^\top \AA^\top \AA \QQ \QQ^\top z_i \\
        & = z_i^\top \AA^\top \AA z_i 
    \end{split}
\end{equation}
where the first  equality follows from the definition of $\tilde{\ZZ}$, the second follows from observing that $\tilde{\ZZ}$ are the singular vectors of $\AA\QQ$ as shown above, the third follows from $\QQ^\top$ having orthonormal rows, the fourth from $z_i$ being the $i$-th singular vector of $\AA \QQ \QQ^\top $ and the last from observing that $z_i$ is in the column span of $\QQ$ and thus $\QQ\QQ^\top z_i = z_i$. 
This concludes the proof. 
\end{proof}

Next, we show a lemma relating a high-accuracy per vector guarantee to cost on the residual subspace.

\begin{lemma}[High-Accuracy Per-Vector Guarantee to Residual Cost]
\label{lem:high-accuracy-to-residual}
Given a matrix $\AA \in \mathbb{R}^{n \times d}$, integer $p\geq1$, $k\in [d]$, $\ell\in[k]$ and orthonormal vectors $\{ w_i \}_{i \in [\ell]}$ such that $\norm{\AA^\top w_i }_2^2 \geq \sigma_{i}^2 - \poly\Paren{\epsilon/ d} \sigma_{k+1}^2$ and $(\sigma_{\ell}-\sigma_{\ell+1})/\sigma_{\ell} \geq \epsilon/d$. Let $\WW$ be the matrix formed by stacking together the $w_i$'s as columns. Then,
\begin{equation*}
    \norm{\AA_\ell^\top\Paren{ \II -\WW\WW^\top } }^2_F \leq \poly\Paren{\frac{\epsilon}{d}}\sigma_{k+1}^2 , 
\end{equation*}
where $\AA_{\ell}$ is the matrix obtained by truncating all but the top $\ell$ singular values of $\AA$. 
\end{lemma}
\begin{proof}
By Pythagorean Theorem,
\begin{equation}
\begin{split}
    \norm{\AA^\top \WW \WW^\top }_F^2 & = \norm{\AA_\ell^\top \WW \WW^\top }_F^2 + \norm{ \Paren{\AA - \AA_\ell}^\top \WW \WW^\top }_F^2 \\
    & \leq \norm{\AA_\ell^\top \WW \WW^\top }_F^2 + \sigma_{\ell +1}^2 \Paren{ \norm{\WW}^2_F - \norm{\UU_\ell^\top \WW }^2_F  },
\end{split}
\end{equation}
where $\AA = \UU \Sig\VV^\top$. Further,

\begin{equation}
\label{eqn:upper-bound-on-u-ell}
    \begin{split}
        \norm{\AA_\ell^\top \WW \WW^\top }_F^2 = \norm{ \Sig_\ell \UU_\ell^\top  \WW  }_F^2 \leq \sum_{i \in [\ell]} \sigma_i^2 - \sigma_\ell^2\Paren{\ell - \norm{\UU_\ell^\top \WW }_F^2},
    \end{split}
\end{equation}
where the last inequality is obtained by making the Euclidean norm of all of the $\UU_\ell\WW $'s in $[\ell-1]$to be $1$, and  the $\ell$-th row to be $\norm{(\UU_\ell^\top \WW)_{\ell,*} }^2_2=\norm{\UU_\ell^\top \WW }_F^2 - (\ell-1)$. Rearranging  Equation \eqref{eqn:upper-bound-on-u-ell}, we have
\begin{equation}
\label{eqn:bound-alignment-rearranged}
    \ell - \norm{\UU_\ell^\top \WW }_F^2 \leq  \frac{\norm{\AA_\ell^\top}_F^2 - \norm{\AA_\ell^\top\WW}_F^2}{\sigma_{\ell}^2}.
\end{equation}
Now, observe $\norm{\WW }_F^2 = \ell$, and substituting \eqref{eqn:bound-alignment-rearranged} back into \eqref{eqn:upper-bound-on-u-ell}, 

\begin{equation}
\label{eqn:upperbound-atw-f}
    \norm{\AA^\top \WW}_F^2 \leq \norm{\AA_\ell^\top \WW}_F^2 + \frac{\sigma_{\ell+1}^2 }{\sigma_{\ell}^2 } \Paren{ \norm{\AA_\ell^\top }_F^2  - \norm{\AA_\ell^\top \WW}_F^2 }.
\end{equation}
Next, we can use the guarantee's on the $w_i$ to obtain a lower bound on $\norm{\AA^\top \WW}_F^2$ as follows:
\begin{equation}
\label{eqn:lowerbound-atw-f}
    \norm{\AA^\top \WW }_F^2 = \sum_{i\in[\ell]}\norm{\AA^\top w_i}_2^2 \geq \norm{\AA_\ell^\top }_F^2 - \ell\cdot \poly\Paren{ \frac{\epsilon}{d} }\sigma_{k+1}^2,
\end{equation}
Combining equations \eqref{eqn:upperbound-atw-f} and \eqref{eqn:lowerbound-atw-f}, we have
\begin{equation}
\begin{split}
   \norm{\AA_\ell^\top  -\AA_\ell^\top \WW\WW^\top}_F^2 & = \Paren{\norm{\AA_\ell^\top }_F^2 - \norm{\AA_\ell^\top \WW}_F^2} \\
    & \leq\frac{\ell \poly(\epsilon/d) \sigma_{k+1}^2 }{ 1- (\sigma_{\ell+1}/\sigma_{\ell})^2 } \\
    &\leq \poly(\epsilon/d) \sigma_{k+1}^2,
\end{split}
\end{equation}
which concludes the proof.
\end{proof}

Next, we need a lemma relating the Schatten-$p$ norm of $\AA \ZZ$ to that of $\WW^\top \AA$, where $\ZZ$ is an arbitrary orthonormal basis and $\WW$ is an orthonormal basis for $\AA \ZZ$.

\begin{lemma}
\label{lem:relating-singular-values}
Given a full-rank $n \times d$ matrix $\AA$, let $\WW $ be a $n \times k$ matrix with orthonormal columns. Further, let $\ZZ$ be an $d \times k$ matrix with orthonormal columns such that $\ZZ$ is a basis for $\AA^\top \WW$. Then, for all $i \in [k]$,
\begin{equation*}
    \sigma_i\Paren{  \AA \ZZ }^p \geq \sigma_i\Paren{ \AA^\top \WW  }^p
\end{equation*}
\end{lemma}
\begin{proof}
We use the following fact that for two matrices $\AA$ and $\BB$, we have that for all
$i$, $\sigma_i(\AA \cdot \BB) \leq \sigma_i(\AA) \cdot \sigma_1(\BB)$; see, e.g., (2) in \cite{lcc15} and references [33-36] therein. 

Using this fact, we have
$$\sigma_i(\AA^\top \WW) = \sigma_i(\AA^\top \WW\WW^T) = \sigma_i(\ZZ \ZZ^T \AA^\top \WW \WW^T) \leq \sigma_i (\ZZ \ZZ^T \AA^\top) \cdot \sigma_1(\WW \WW^T) = \sigma_i (\ZZ \ZZ^T \AA^\top) = \sigma_i(\ZZ^T \AA),$$
where we have used that $\sigma_1(\WW \WW^T) = 1$ since $\WW \WW^T$ is a projection matrix, and the fact that $\ZZ \ZZ^T$ is a basis for the column span of $\AA^\top \WW$. Raising both sides to the $p$-th power establishes the lemma.

\end{proof}

Finally, we can combine the two aforementioned lemmas to obtain the following corollary:
\begin{corollary}[Changing Basis for high-accuracy vectors]
\label{cor:change-of-basis-high-acc}

Given a matrix $\AA \in \mathbb{R}^{n \times d}$, integer $p\geq1$, $k \in [d]$, $\ell \in [k]$ and orthonormal vectors $\{ w_i \}_{i \in [\ell]}$ such that $\norm{\AA^\top w_i }_2^2 \geq \sigma_{i}^2 - \poly\Paren{\epsilon/ d} \sigma_{k+1}^2$ and $(\sigma_{\ell}-\sigma_{\ell+1})/\sigma_{\ell} \geq \epsilon/d$. Let $\WW$ be the matrix formed by stacking together the $w_i$'s as columns and let $\ZZ$ be an orthonormal basis for $\AA^\top \WW\WW^\top$. Then,
\begin{equation*}
    \norm{\AA_\ell \Paren{ \II -\ZZ\ZZ^\top } }^2_F \leq \poly\Paren{\frac{\epsilon}{d}}\sigma_{k+1}^2 . 
\end{equation*}
\end{corollary}
\begin{proof}
We closely follow the proof in Lemma \ref{lem:high-accuracy-to-residual}.
By Pythagorean Theorem,
\begin{equation}
\begin{split}
    \norm{\AA  \ZZ \ZZ^\top }_F^2 & = \norm{\AA_\ell \ZZ \ZZ^\top }_F^2 + \norm{ \Paren{\AA - \AA_\ell} \ZZ \ZZ^\top }_F^2 \\
    & \leq \norm{\AA_\ell \ZZ \ZZ^\top }_F^2 + \sigma_{\ell +1}^2 \Paren{ \norm{\ZZ}^2_F - \norm{\VV_\ell \ZZ}^2_F  },
\end{split}
\end{equation}
where $\AA = \UU \Sig\VV^\top$. Further,

\begin{equation}
\label{eqn:upper-bound-on-u-ell-2}
    \begin{split}
        \norm{\AA_\ell \ZZ \ZZ^\top }_F^2 = \norm{ \Sig_\ell \VV_\ell^\top  \ZZ  }_F^2 \leq \sum_{i \in [\ell]} \sigma_i^2 - \sigma_\ell^2\Paren{\ell - \norm{\VV_\ell^\top \ZZ }_F^2},
    \end{split}
\end{equation}
where the last inequality is obtained by making the Euclidean norm of all of the $\VV_\ell^\top\ZZ $'s in $[\ell-1]$to be $1$, and  the $\ell$-th row to be $\norm{(\VV_\ell^\top \ZZ)_{\ell,*} }^2_2=\norm{\VV_\ell^\top \ZZ }_F^2 - (\ell-1)$. Rearranging  Equation \eqref{eqn:upper-bound-on-u-ell-2}, we have
\begin{equation}
\label{eqn:bound-alignment-rearranged-2}
    \ell - \norm{\VV_\ell^\top \ZZ }_F^2 \leq  \frac{\norm{\AA_\ell }_F^2 - \norm{\AA_\ell \ZZ}_F^2}{\sigma_{\ell}^2}.
\end{equation}
Now, observe $\norm{\ZZ}_F^2 = \ell$, and substituting \eqref{eqn:bound-alignment-rearranged-2} back into \eqref{eqn:upper-bound-on-u-ell-2}, 

\begin{equation}
\label{eqn:upperbound-az-f}
    \norm{\AA\ZZ}_F^2 \leq \norm{\AA_\ell \ZZ}_F^2 + \frac{\sigma_{\ell+1}^2 }{\sigma_{\ell}^2 } \Paren{ \norm{\AA_\ell }_F^2  - \norm{\AA_\ell \ZZ}_F^2 }.
\end{equation}
Next, we can use the guarantee's on the $w_i$ to obtain a lower bound on $\norm{\AA^\top \WW}_F^2$ as follows:
\begin{equation}
    \norm{\AA^\top \WW }_F^2 = \sum_{i\in[\ell]}\norm{\AA^\top w_i}_2^2 \geq \norm{\AA_\ell^\top }_F^2 - \ell\cdot \poly\Paren{ \frac{\epsilon}{d} }\sigma_{k+1}^2,
\end{equation}
Observe $\norm{\AA^\top \WW}_F^2 = \sum_{i \in [\ell]} \sigma_{i}^2(\AA^\top \WW)$. By Lemma \ref{lem:relating-singular-values}, we know that for all $i$, $\sigma_{i}^2(\AA^\top \WW)\leq \sigma_{i}^2(\AA \ZZ)$. Therefore, we can restate the above equation as follows:

\begin{equation}
\label{eqn:lowerbound-az-f}
    \norm{\AA\ZZ}_F^2 \geq \norm{\AA^\top \WW }_F^2\geq  \norm{\AA_\ell  }_F^2 - \ell\cdot \poly\Paren{ \frac{\epsilon}{d} }\sigma_{k+1}^2,
\end{equation}
Combining equations \eqref{eqn:upperbound-az-f} and \eqref{eqn:lowerbound-az-f}, we have
\begin{equation}
\begin{split}
   \norm{\AA_\ell   -\AA_\ell  \ZZ\ZZ^\top}_F^2 & = \Paren{\norm{\AA_\ell }_F^2 - \norm{\AA_\ell \ZZ }_F^2} \\
    & \leq\frac{\ell \poly(\epsilon/d) \sigma_{k+1}^2 }{ 1- (\sigma_{\ell+1}/\sigma_{\ell})^2 } \\
    &\leq \poly(\epsilon/d) \sigma_{k+1}^2,
\end{split}
\end{equation}
which concludes the proof.
\end{proof}

We now have all the ingredients we need to complete the proof of Theorem \ref{thm:optimal_schatten_p_lra}.

\begin{proof}[Proof of Theorem \ref{thm:optimal_schatten_p_lra}]
% We first analyze the case where $p \leq \log(1/\epsilon)/(6\log(c))$.
Observe, using Lemma~\ref{lem:gap_independent} with probability at least $97/100$,  Step 1 of Algorithm \ref{algo:optimal_schatten_p_lra} outputs $\hat{\sigma}_i$'s such that for all $i \in [k+1]$, $\hat{\sigma}_{i}^2 = \Paren{1\pm 0.1/p} \sigma_{i}^2$ and $\hat{\sigma}_s^2 = \Paren{1\pm 0.1/p} \sigma_{s}^2$, for $s = \bigO{kp^{-1/3}/\epsilon^{1/3}}$. Condition on this event. 

% Our proof proceeds first finding the top $\ell$ singular values, that have a constant gap with $\sigma_{k+1}$. We get highly accurate estimates of the singular values along these directions and we project away from this set. 
At a high level, we proceed via a case analysis: either the Schatten-$p$ norm of the tail is large compared to the $(k+1)$-st singular value, and we don't require a highly accurate solution, or the Schatten-$p$ norm of the tail is small, and increasing the block size induces a gap. We formalize this intuition into a proof. 
 
\paragraph{No large singular values.} Let us first consider the case where  $\sigma_1 < \Paren{1 + 1/p}\sigma_{k+1}$. 
% \kcnote{This is not the condition used or implied by the algorithm.}\abnote{we don't need this condition, it is just in the analysis} 
We yet again split into cases, and consider the case where the Schatten-$p$ norm of the tail is small, i.e. $\norm{\AA -\AA_k}^p_{\calS_p} \leq \frac{k}{p^{1/3} \epsilon^{1/3}} \cdot \sigma^p_{k+1}$.
Observe, for any $t \in [1, d-k-1]$,
\begin{equation}
\label{eqn:tradeoff-analysis}
    \frac{k}{p^{1/3} \epsilon^{1/3}} \cdot \sigma^p_{k+1}  \geq  \norm{\AA -\AA_k}^p_{\calS_p} \geq \sum_{i = k+1}^{k+1+t} \sigma_{i}^p \geq t \sigma_{k+1+t}^p.
\end{equation}
Then, setting $t = \frac{ \Paren{1 + 1/p}^p  k} {\epsilon^{1/3} p^{1/3}} = \Theta\Paren{  \frac{ k} {\epsilon^{1/3} p^{1/3}} } $,  we have  $\sigma_{k+1+t} \leq \sigma_{k+1}/\Paren{1 + 1/p}$. It suffices to show that we can detect this gap for some $s \geq k+1+t$.  Recall, we know that $\hat{\sigma}_{k+1} = (1\pm 0.1/p) \sigma_{k+1}$ and $\hat{\sigma}_s = (1\pm 0.1/p)\sigma_s$. Then, we have
\begin{equation}
\label{eqn:small-gap-case}
    \hat{\sigma}_s\leq \Paren{1+ \frac{0.1}{p}} \sigma_s \leq \Paren{1+ \frac{0.1}{p}} \sigma_{k+1 +t} \leq \Paren{1+ \frac{0.1}{p}} \cdot \Paren{\frac{1}{1+1/p}} \sigma_{k+1} \leq \frac{1}{\Paren{1 + \frac{0.5}{p}} } \hat{\sigma}_{k+1}.
\end{equation}
Therefore, Algorithm \ref{algo:optimal_schatten_p_lra} 
% \kcnote{Doesn't there need to be a claim that $\ZZ_2$ will be output in this case, that is, the right tests on the singular value estimates will be passed?}\abnote{i simplified the test, so it just outputs $\ZZ_2$ when the above equation holds. } 
outputs $\ZZ$, an orthonormal basis for $\AA^\top \WW_2$,  where $\WW_2$ is obtained by running  Algorithm \ref{algo:simul_power_iter} on $\AA^\top$,  initialized with a block size of $\Theta\Paren{  \frac{ k} {\epsilon^{1/3} p^{1/3}} }$ and run for $\bigO{\log(d/\epsilon)\sqrt{p} } $ iterations. Observe, since $\sigma_{k+1+t} \leq \sigma_{k+1}/\Paren{1 + 1/p}$,  this suffices to demonstrate a gap that depends on $p$ as follows: $\frac{\sigma_k}{\sigma_k - \sigma_{k+t+1}} \leq p$.
% \kcnote{I can't make out if this is stated as something proven, or something that it suffices to prove. I think it's using $\sigma_{k+1+t} \leq \sigma_{k+1}/\Paren{1 + 1/p}$, but a back pointer or reminder would be helpful.}\abnote{added.}
Recall, we account for this gap by running $\bigO{\log(d)\sqrt{p}}$ iterations. Using the gap dependent analysis (Lemma \ref{lem:gap_dependent}), we can conclude that with probability at least $99/100$, for all $i \in [k]$, 
\begin{equation}
\label{eqn:high-acc-no-gap}
    \norm{\AA^\top (\WW_2)_{*,i} }_2^2 \geq \sigma_i^2 - \poly\Paren{\frac{\epsilon}{d}} \sigma^2_{k+1}.
\end{equation}
Then, applying Lemma \ref{lem:correlated_projections} with $ \WW_2 \WW_2^\top$ satisfying the guarantee in \eqref{eqn:high-acc-no-gap}, we have 
\begin{equation}
\label{eqn:bound-azzt-p-p-high-acc}
\begin{split}
     \norm{ \AA^\top \WW_2 \WW_2^\top }_{\calS_p}^p & \geq \norm{  \AA_k }_{\calS_p}^p - \poly\Paren{\frac{\epsilon}{d}} \sum_{i \in [k]}  \sigma_{k+1}^2 \sigma_{i}^{p-2} \\
    & \geq \norm{  \AA_k }_{\calS_p}^p -  \poly\Paren{\frac{\epsilon}{d}}  \sigma_{k+1}^{p} .
\end{split}
\end{equation}
where the last inequality uses that $\sigma_1 < (1+1/p)\sigma_{k+1}$ and $(1+1/p)^{p-2} = \bigO{1}$. 
%\textcolor{red}{We invoke lemma 5.6 correctly here, this was already fixed in a prev version and let to change int he algorithm: running MM on $\AA^\top$ instead of $\AA$.}
Next, we use Lemma \ref{lem:schatten_p_orthogonal_projections} to relate $\norm{\AA^\top \WW_2 \WW_2^\top}_{\calS_p}^p$ to $\norm{\AA \Paren{\II - \ZZ \ZZ^\top}}_{\calS_p}^p$, where $\ZZ$ is an orthonormal basis for $\AA^\top \WW_2 \WW_2^\top$ as output by the algorithm. Setting $\QQ = \ZZ \ZZ^\top$ and $\PP = \WW_2 \WW_2^\top$, we observe that $\norm{\PP \AA \QQ}_{\calS_p}^p = \norm{ \AA^\top \WW_2 \WW_2^\top }_{\calS_p}^p =\norm{  \WW_2 \WW_2^\top\AA }_{\calS_p}^p  $ and $\norm{\Paren{\II - \PP } \AA \Paren{\II - \QQ} }_{\calS_p}^p = \norm{\AA \Paren{\II - \ZZ \ZZ^\top}}_{\calS_p}^p$.
Then, invoking Lemma  \ref{lem:schatten_p_orthogonal_projections}  and plugging in Equation \eqref{eqn:bound-azzt-p-p-high-acc}, we have
\begin{equation}
\begin{split}
    \norm{ \Paren{\II - \PP} \AA \Paren{\II - \QQ } }_{\calS_p}^p = \norm{  \AA \Paren{\II - \ZZ\ZZ^\top } }_{\calS_p}^p & \leq \norm{\AA}_{\calS_p}^p - \norm{\AA^\top \WW_2  \WW_2^\top}_{\calS_p}^p
    \\ 
    & \leq \norm{\AA}_{\calS_p}^p - \norm{  \AA_k }_{\calS_p}^p +  \poly\Paren{\frac{\epsilon}{d}}  \sigma_{k+1}^{p}  \\
    & \leq \Paren{1+\poly\Paren{\frac{\epsilon}{d}}}\norm{\AA - \AA_k }_{\calS_p}^p ,
\end{split}
\end{equation}
which concludes the analysis in this case.

As shown in Equation \ref{eqn:small-gap-case}, we can detect a gap between $\sigma_{k+1+t}$ and $\sigma_{k+1}$ by comparing $\hat{\sigma}_s$ and $\hat{\sigma}_{k+1}$. When  \ref{eqn:small-gap-case} does not hold,
we know that $\hat{\sigma}_s\geq \Paren{1+0.5/p}\hat{\sigma}_{k+1}$ and  Algorithm \ref{algo:optimal_schatten_p_lra} outputs $\ZZ$, an orthonormal basis for $\AA^\top \WW_1\WW_1^\top$. Since we have $(1\pm 0.1/p)$-approximate estimates to these quantities, we can conclude that $\sigma_s \geq \Paren{1+0.1/p}\sigma_{k+1}$. Then, we have
\begin{equation}
\label{eqn:large-tail-error-analysis}
    \norm{\AA - \AA_k }_{\calS_p}^p \geq s \cdot \sigma_s^p = \Omega\Paren{ \frac{k}{\epsilon^{1/3} p^{1/3}} } \sigma_{k+1}^p.
\end{equation}
It therefore remains to consider the case where $\norm{\AA -\AA_k}^p_{\calS_p}> \frac{ck}{p^{1/3} \epsilon^{1/3}} \cdot \sigma^p_{k+1}$, for a fixed universal constant $c$.  Here, we note that the tail is large enough that an additive error of  $\bigO{\epsilon^{2/3}p^{1/3}}\sigma_{k+1}^2$ on each of the top-$k$ singular values suffices. Formally, it follows from Lemma \ref{lem:gap_independent} (setting $\gamma = \epsilon^{2/3} p^{-1/3}$, and invoking it for $\AA^\top$) that initializing Algorithm \ref{algo:simul_power_iter} with block size $k$ and running for $\bigO{ \log(d/\epsilon) p^{1/6}/\epsilon^{1/3}}$ iterations suffices to output a $n \times k$ matrix $\WW_1$ such that with probability at least $99/100$, for all $i \in [k]$,
\begin{equation}
\label{eqn:bound-no-singular-value-gap}
    \norm{\AA^\top \Paren{\WW_1}_{*,i} }_2^2 \geq \sigma_i^2 - \epsilon^{2/3} p^{-1/3} \sigma_{k+1}^2.
\end{equation}
Then, invoking Lemma \ref{lem:correlated_projections} with $\AA^\top$ and $\WW_1$ as defined above, we have 
\begin{equation}
\label{eqn:bound-azzt-p-p}
\begin{split}
     \norm{ \AA^\top \WW_1 \WW_1^\top }_{\calS_p}^p & = \norm{\WW_1 \WW_1^\top \AA   }_{\calS_p}^p \\
     & \geq \norm{  \AA_k }_{\calS_p}^p - \sum_{i \in [k]} \bigO{ \epsilon^{2/3} p^{-1/3} p } \sigma_{k+1}^2 \sigma_{i}^{p-2} \\
    & \geq \norm{  \AA_k }_{\calS_p}^p -  \bigO{ k \epsilon^{2/3} p^{2/3} }  \sigma_{k+1}^{p} 
\end{split}
\end{equation}
where the last inequality uses that $\sigma_1 < (1+1/p)\sigma_{k+1}$ and $(1+1/p)^p = \bigO{1}$. Recall, in this case, Algorithm \ref{algo:optimal_schatten_p_lra} outputs $\ZZ \ZZ^\top$ where $\ZZ$ is an orthonormal basis for $ \AA^\top \WW_1 \WW^\top_1 $.  
Next, we invoke Lemma \ref{lem:schatten_p_orthogonal_projections} to relate $\norm{  \AA^\top \WW_1 \WW_1^\top  }_{\calS_p}^p$ to $\norm{\AA \Paren{\II - \ZZ \ZZ^\top}}_{\calS_p}^p$. Setting $\QQ = \ZZ \ZZ^\top$ and $\PP = \WW_1 \WW^\top_1$, we observe that $\norm{\PP \AA \QQ}_{\calS_p}^p = \norm{ \WW_1 \WW_1^\top   \AA }_{\calS_p}^p $ and $\norm{\Paren{\II - \PP } \AA \Paren{\II - \QQ} }_{\calS_p}^p = \norm{\AA \Paren{\II - \ZZ \ZZ^\top}}_{\calS_p}^p$. 
%\textcolor{red}{We invoke lemma 5.6 correctly here, this was already fixed in a prev version and lead to change in the algorithm: running MM on $\AA^\top$ instead of $\AA$.}
Then, invoking Lemma  \ref{lem:schatten_p_orthogonal_projections}  and plugging in Equation \eqref{eqn:bound-azzt-p-p}, we have
\begin{equation}
\begin{split}
    \norm{ \Paren{\II - \PP} \AA \Paren{\II - \QQ } }_{\calS_p}^p = \norm{  \AA \Paren{\II - \ZZ\ZZ^\top } }_{\calS_p}^p & \leq \norm{\AA}_{\calS_p}^p - \norm{\WW_1 \WW_1^\top \AA }_{\calS_p}^p
    \\ 
    & \leq \norm{\AA}_{\calS_p}^p - \norm{  \AA_k }_{\calS_p}^p +  \bigO{ k \epsilon^{2/3} p^{2/3} }  \sigma_{k+1}^{p}  \\
    & \leq \Paren{1+\bigO{p\epsilon}}\norm{\AA - \AA_k }_{\calS_p}^p ,
\end{split}
\end{equation}
where the last inequality follows from our assumption on the Schatten-$p$ norm of the tail, given the case we are in. Taking the $(1/p)$-th root, and recalling that $\epsilon< 1/2$, we obtain 
\begin{equation}
     \norm{  \AA \Paren{\II - \ZZ\ZZ^\top } }_{\calS_p} \leq \Paren{1+\bigO{\epsilon } } \norm{\AA - \AA_k}_p,
\end{equation}
which concludes the case where $\ell =0$.

\paragraph{Large Singular Values.} Next, we consider the case where $\sigma_1 > \Paren{1+1/p}\sigma_{k+1}$. Then, let $\ell \in[k]$ be the largest integer such that $\sigma_\ell > \Paren{1+0.5/p}\sigma_{k+1}$ and $\sigma_{\ell +1} < (1- \epsilon/d) \sigma_{\ell}$. Observe, such an $\ell$ is guaranteed to exist.  

We then note that in all settings Algorithm \ref{algo:optimal_schatten_p_lra} runs $\Omega(\log(d/\epsilon)\sqrt{p}$ iterations on $\AA^\top$ and 
since exists a gap of size $p$ between $\sigma_\ell$ and $\sigma_{k+1}$ it follows from Lemma \ref{lem:gap_dependent} that running Block Krylov Iteration
that Algorithm \ref{algo:optimal_schatten_p_lra} always outputs an orthonormal matrix $\WW$ s.t. for all $i \in [\ell]$,

\begin{equation}
\label{eqn:mm_guarantee_rank_k}
    \norm{ \AA^\top \WW_{*,i} }^2 \geq \sigma_{i}^2 - \poly\Paren{\frac{\epsilon}{d}}\sigma_{k+1}^2.
\end{equation}
Further, for all $ i \in [\ell + 1, k]$, we have 

\begin{equation}
\label{eqn:mm_guarantee_rank_k2}
    \norm{ \AA^\top \WW_{*,i} }^2 \geq \sigma_{i}^2 -  \gamma_i \sigma_{k+1}^2,
\end{equation}
where $\gamma_i$ is determined by whether $\WW = \WW_1$ or $\WW = \WW_2$, as we discuss later. 

% and compute a $d \times \ell$ matrix $\ZZ_1$ as stated in 
% \kcnote{This is not quite the condition used or implied by the algorithm.}\abnote{yea this just appears in the analysis, it does not matter what matrix the algorithm outputs in this case, both work.} 
% We further observe that the Algorithm \ref{algo:optimal_schatten_p_lra} runs at least $\Omega({\log(d/\epsilon)}\sqrt{p})$
% %\kcnote{Changed $\bigO{}$ to $\Omega$, this is needed elsewhere also.}
% iterations (since $p \leq \log(d)/\epsilon$) since exists a gap of size $p$ between $\sigma_\ell$ and $\sigma_{k+1}$ it follows from Lemma \ref{lem:gap_dependent} that running Block Krylov Iteration for $\bigO{\log(d/\epsilon)\sqrt{p}}$ iterations with block size $\geq k$ suffices to output a matrix $\ZZ_1$ such that with probability at least $99/100$, for all $i \in [\ell]$,  
% \begin{equation}
% \label{eqn:mm_guarantee_rank_k}
%     \norm{\AA (\ZZ_1)_{*,i} }_2^2 \geq \sigma_i^2(\AA) - \poly\Paren{\frac{\epsilon}{d}} \sigma^2_{k+1}(\AA),
% \end{equation}
% and for all $i \in [\ell+1, k]$, 
% \begin{equation}
% \label{eqn:mm_guarantee_rank_k2}
%     \norm{\AA (\ZZ_1)_{*,i} }_2^2 \geq \sigma_i^2(\AA) -  \epsilon^{2/3} \sigma^2_{k+1}(\AA),
% \end{equation}
We note that we cannot simply take $p/2$-th powers here (for large $p$) as this would introduce cross terms that scale proportional to $\sigma_i(\AA)$, which can be significantly larger than $\sigma_{k+1}(\AA)$. Instead, we require a finer analysis by splitting $\AA$ into a head and tail term. Further, we let $\ZZ$ be an orthonormal basis for $\AA^\top \WW \WW^\top$.

We are now ready to bound $\norm{\AA \Paren{\II - \ZZ\ZZ^\top} }_{\calS_p}$. By the triangle inequality,
\begin{equation}
\label{eqn:tri_ineq_schatten_p-up}
\begin{split}
    \norm{\AA \Paren{\II - \ZZ \ZZ^\top}}_{\calS_p}  & \leq \norm{\AA_{\ell}\Paren{\II -  \ZZ \ZZ^\top} }_{\calS_p} + \norm{\Paren{\AA - \AA_{\ell}} \Paren{\II - \ZZ \ZZ^\top}}_{\calS_p}  \\
\end{split}
\end{equation}
By Corollary \ref{cor:change-of-basis-high-acc}, we know that $\norm{\AA_{\ell}\Paren{\II -\ZZ\ZZ^\top}}_F^2 \leq \poly\Paren{\frac{\epsilon}{d}}\sigma_{k+1}^2$, and since all Schatten norms are within a $\sqrt{d}$ factor of the Frobenius norm, we have

\begin{equation*}
    \norm{\AA_{\ell}\Paren{\II -\ZZ\ZZ^\top}}_{\calS_p} \leq \poly \Paren{\frac{\epsilon}{d}} \sigma_{k+1}.
\end{equation*}
Substituting this back into Equation \eqref{eqn:tri_ineq_schatten_p-up}, we have
\begin{equation}
\label{eqn:bounding-all-but-one-cost}
    \norm{\AA \Paren{\II - \ZZ \ZZ^\top}}_{\calS_p} \leq \bigO{\frac{\epsilon}{d}} \norm{\AA - \AA_k}_{\calS_p} +  \underbrace{  \norm{\Paren{ \AA - \AA_{\ell }} \Paren{\II - \ZZ\ZZ^\top} }_{\calS_p}}_{\ref{eqn:bounding-all-but-one-cost}.1} .
\end{equation}
It remains to bound term \ref{eqn:bounding-all-but-one-cost}.1 above.
% \textcolor{red}{We invoke 5.6 correctly and get a $W$ instead of $Z$ on the rhs below}
By triangle inequality, we have

\begin{equation}
\label{eqn:triangle-w-iminusw-split}
\begin{split}
    &\norm{  \Paren{ \AA -\AA_\ell} (\II - \ZZ \ZZ^\top) }_{\calS_p}^p \\
    & \leq \Paren{ \norm{(\II -\WW\WW^\top) \Paren{ \AA -\AA_\ell} (\II - \ZZ \ZZ^\top) }_{\calS_p} + \norm{\WW\WW^\top \Paren{ \AA -\AA_\ell} (\II - \ZZ \ZZ^\top) }_{\calS_p}}^p 
\end{split}
\end{equation}

We bound the two terms on the RHS independently. 
To upper bound $\norm{\WW\WW^\top \Paren{ \AA -\AA_\ell} (\II - \ZZ \ZZ^\top) }_{\calS_p}^p$ we use the relation between Frobenius and Schatten norms, and recall that by definition, $\WW\WW^\top\AA\Paren{\II - \ZZ\ZZ^\top} = \WW\WW^\top\AA - \WW\WW^\top \AA \ZZ\ZZ^\top =0$, and thus

\begin{equation}
\begin{split}
    \norm{\WW\WW^\top \Paren{ \AA -\AA_\ell} (\II - \ZZ \ZZ^\top) }_{\calS_p} & \leq \sqrt{k} \norm{\WW \WW^\top \Paren{\AA - \AA_\ell} (\II - \ZZ \ZZ^\top)  }_F \\
    & = \sqrt{k} \norm{\WW \WW^\top  \AA_\ell (\II - \ZZ \ZZ^\top)  }_F  \\
    & \leq \sqrt{k} \norm{  \AA_\ell (\II - \ZZ \ZZ^\top)  }_F \\
    & \leq  \poly\Paren{\frac{\epsilon}{d}} \sigma_{k+1} ,
\end{split}
\end{equation}
where the last inequality  follows from Corollary \ref{cor:change-of-basis-high-acc}. Therefore, combining the above, we have 

\begin{equation}
\label{eqn:bound-w-a-minus-al}
\begin{split}
    \norm{\WW\WW^\top \Paren{ \AA -\AA_\ell} (\II - \ZZ \ZZ^\top) }_{\calS_p} 
     \leq \poly\Paren{\frac{\epsilon}{d}} \sigma_{k+1} ,
\end{split}
\end{equation}

It remains to upper bound $\norm{(\II -\WW\WW^\top) \Paren{ \AA -\AA_\ell} (\II - \ZZ \ZZ^\top) }_{\calS_p}^p$. 
Recall, $\WW$ is the orthonormal basis output by Block Krylov run on $\AA^\top$, and in Algorithm \ref{algo:optimal_schatten_p_lra} $\WW$ is either $\WW_1$ or $\WW_2$. Let $\ZZ$ is a basis for $\AA^\top \WW \WW^\top$. Then, applying Lemma \ref{lem:schatten_p_orthogonal_projections} with $\QQ = \ZZ\ZZ^\top$ and $\PP = \WW\WW^\top$, we have

\begin{equation}
   \begin{split}
     \norm{\Paren{\II - \WW\WW^\top}  \Paren{\AA - \AA_{\ell }} \Paren{\II - \ZZ\ZZ^\top} }^p_{\calS_p} & \leq  \norm{\Paren{ \AA - \AA_{\ell }}  }^p_{\calS_p} -  \norm{\WW \WW^\top\Paren{  \AA - \AA_{\ell } \ZZ\ZZ^\top } }^p_{\calS_p}\\
     & = \sum_{j \in [\ell+1, d]  }\sigma_j^p - \sum_{j \in [k ] } \sigma_j^p\Paren{ \WW^\top \Paren{\AA -\AA_{\ell}} \ZZ  }
\end{split}
\end{equation}

% \textcolor{red}{A: fix here. the above should be $\norm{\Paren{ \AA - \AA_{j^*}} \Paren{\II - \ZZ\ZZ^\top} }^p_{\calS_p}  \leq  \norm{\Paren{ \AA - \AA_{j^*}}  }^p_{\calS_p} -  \norm{\WW \WW^\top \Paren{ \AA - \AA_{j^*}}  }^p_{\calS_p}$ }
% \textcolor{red}{We  propgate this $W (A-A_{j^*})$ instead of $(A- A_j*)Z$ every where below. }

Next, we show that for all $j \in [k]$, $\sigma_j\Paren{\WW^\top  \Paren{\AA -\AA_{\ell}  }\ZZ  } \geq \sigma_{j + \ell}\Paren{\WW^\top \AA} $. Here, we invoke Fact  \ref{fact:cauchy} for $\XX = \WW^\top \Paren{ \AA - \AA_{\ell}} \ZZ$ and $\YY = \WW^\top \AA_{\ell}\ZZ$, with $i = j$ and $j = \ell$. Note, the precondition on the indices $i,j$ in Fact \ref{fact:cauchy} is satisfied since $\XX, \YY$ are $n \times k$ matrices, and $j \in[k]$ and $\ell < k$. Then, we have 

\begin{equation*}
\begin{split}
    \sigma_{j + \ell}\Paren{ \WW^\top \AA \ZZ  } & =  \sigma_{j+ \ell}\Paren{ \WW^\top \Paren{ \AA - \AA_{\ell}\ZZ } +  \WW^\top \AA_{\ell} \ZZ  }\\
    & \leq \sigma_{j}\Paren{ \WW^\top \Paren{ \AA - \AA_{\ell}} \ZZ } + \sigma_{\ell+1}\Paren{ \WW^\top \AA_{\ell} \ZZ },
\end{split}
\end{equation*}
but $\WW^\top \AA_{\ell}\ZZ$ is a rank $\leq \ell$ matrix, and thus $\sigma_{\ell+1}\Paren{\WW^\top \AA_{\ell} \ZZ}=0$. Therefore, we can conclude, 

% Observe, since $\ZZ$ has orthonormal columns, the singular values of  $\Paren{\AA -\AA_{j^*}} \ZZ$ are the same as the singular values of  $ \Paren{\AA -\AA_{j^*}} \ZZ \ZZ^\top$. Further, the singular values of $\AA- \AA_{j^*}$ are $\{ \sigma_{j^* +1}, \ldots , \sigma_{d}\}$. 
% % \kcnote{Is that $\UU_{j^*}\UU^\top_{j^*}$ and not $\VV_{j^*}\VV^\top_{j^*}$?}\abnote{great catch!}
% Then, 
% by Cauchy's Interlacing theorem (see Fact \ref{fact:cauchy}), $\sigma_j\Paren{\Paren{\AA -\AA_{j^*}} \ZZ} \geq \sigma_{j + j^* }\Paren{\AA \ZZ } $ \textcolor{red}{we need to prove this inequality in some other way.} 
% \begin{equation*}
%     \sigma_j\Paren{ \Paren{\AA - \AA_{j*}} \ZZ \ZZ^\top } = \lambda_j^{1/2}\Paren{ \ZZ \ZZ^\top \VV \Paren{\Sig - \Sig_{j*} } \VV^\top \ZZ \ZZ^\top } = \lambda_j^{1/2}\Paren{\ZZ \ZZ^\top \Paren{ \VV - \VV_{j^*} }  \Sig  \VV^\top \ZZ  } = \lambda_j^{1/2}\Paren{\ZZ \ZZ^\top \Paren{ \VV - \VV_{j^*} }  \Sig  \VV^\top \ZZ  }
% \end{equation*}

\begin{equation}
\label{eqn:tail-bound-projection-sp}
    \norm{\Paren{\II - \WW \WW^\top } \Paren{ \AA - \AA_{\ell}} \Paren{\II - \ZZ\ZZ^\top} }_{\calS_p}^p \leq  \sum_{j \in [\ell+1, d]  }\sigma^p_j - \sum_{j \in [\ell+1, k +\ell] } \sigma_{j }^p\Paren{\WW^\top \AA  \ZZ}. 
\end{equation} 
However, now we observe that $\WW \WW^\top \AA \ZZ\ZZ^\top = \WW \WW^\top \AA$, and thus $\sigma_j\Paren{\WW^\top \AA\ZZ} = \sigma_j \Paren{\WW \AA}$. 
Recall, for all $j \in [k]$, it follows from Equation \eqref{eqn:intermediate-bound-on-singular-values} in the proof of Lemma \ref{lem:correlated_projections} that $\sigma_j^p(\WW^\top \AA ) =\sigma_j^p(  \AA^\top \WW )  \geq \sigma_j^p(\AA) - \bigO{\gamma_j p } \sigma_{k+1}^2 \sigma_j^{p-2}$. Further, by definition, for $j \in[\ell +1, k]$, $\sigma_{j} \leq \Paren{1+1/p}\sigma_{k+1}$ and thus, for all $j \in[\ell+1, k]$, 
\begin{equation}
\begin{split}
    \sigma_j^p(\WW^\top \AA ) & \geq \sigma_j^p - \bigO{\gamma_j p \Paren{1 + 1/p}^{p-2} } \sigma_{k+1}^p\\
    & \geq \sigma_j^p  - \bigO{\gamma_j p }\sigma_{k+1}^p.
\end{split}
\end{equation}

% \textcolor{red}{We use this new lemma to get a bound on $\sigma_j^p(\WW^\top \AA)$, we already had a bound on $\sigma_j^p(\AA\ZZ)$. } 

Recall, we can correctly determine whether $\norm{\AA - \AA_k }_{\calS_p}^p \leq \frac{k}{p^{1/3}\epsilon^{1/3} } \sigma_{k+1}^p$ or not, up to error in estimating the singular values as shown earlier, in the analysis for the case where there are no large singular values. In particular, let us first consider the case where this is true. Repeating the argument in equations \eqref{eqn:tradeoff-analysis}, \eqref{eqn:small-gap-case}, we can conclude that Algorithm \ref{algo:optimal_schatten_p_lra} outputs $\WW = \WW_2$ and thus for all $j \in [\ell+1, k]$, $\gamma_j = \poly(\epsilon/d)$.  Therefore, substituting this back into Equation \eqref{eqn:tail-bound-projection-sp}, we have
\begin{equation}
\begin{split}
    \norm{\Paren{\II- \WW\WW^\top }\Paren{ \AA - \AA_{\ell}} \Paren{\II - \ZZ\ZZ^\top} }_{\calS_p}^p & \leq \sum_{j \in [\ell+1, d]  }\sigma^p_j - \Paren{ \sum_{j \in [\ell +1, k] }  \sigma^p_j - \bigO{\gamma_j  p} \sigma_{k+1}^p }\\
    & \leq  \norm{\AA - \AA_k}_{\calS_p}^p + \sum_{j \in [\ell+1, k]} \bigO{\gamma_j  p}  \sigma_{k+1}^p \\
    & \leq \Paren{1 + \poly\Paren{\frac{\epsilon}{d}} } \norm{\AA - \AA_k}_{\calS_p}^p ,
\end{split}
\end{equation}
concluding the analysis in this case.

% $ k \sigma_{k+1} \Paren{ \epsilon^{2/3} p^{-1/3} }^{1/2} \leq \epsilon  (k /( p^{1/3} \epsilon^{1/3})) \sigma_{k+1} $

Next, consider the case where $\norm{\AA - \AA_k }_{\calS_p}^p > \frac{k}{p^{1/3}\epsilon^{1/3} } \sigma_{k+1}^p$. Then, repeating the analysis in Equation \eqref{eqn:large-tail-error-analysis}, we know that Algorithm \ref{algo:optimal_schatten_p_lra} outputs $\WW = \WW_2$, and thus for all $j \in [\ell+1, k]$, $\gamma_j = \epsilon^{2/3} p^{-1/3}$ as shown in Equation \eqref{eqn:bound-no-singular-value-gap}. Again, substituting this back into Equation \eqref{eqn:tail-bound-projection-sp}, we have
\begin{equation}
\label{eqn:bound-i-minus-w-a-minus-al}
\begin{split}
    \norm{\Paren{\II- \WW\WW^\top } \Paren{\AA - \AA_{\ell}} \Paren{\II - \ZZ\ZZ^\top} }_{\calS_p}^p  & \leq  \norm{\AA - \AA_k}_{\calS_p}^p + \sum_{j \in [\ell+1, k]} \bigO{\gamma_j  p}  \sigma_{k+1}^p \\
    & \leq \norm{\AA - \AA_k}_{\calS_p}^p +   \bigO{ kp  \cdot \epsilon^{2/3} p^{-1/3}  }  \sigma_{k+1}^p \\
    & \leq \Paren{1+\bigO{\epsilon p} } \norm{\AA - \AA_k}_{\calS_p}^p,
\end{split}
\end{equation}
where the last inequality follows from our assumption. Taking the $(1/p)$-th root and substituting equations \eqref{eqn:bound-w-a-minus-al} and \eqref{eqn:bound-i-minus-w-a-minus-al} back into \eqref{eqn:triangle-w-iminusw-split}, we can conclude

\begin{equation*}
   \norm{\Paren{\AA - \AA_\ell}\Paren{\II - \ZZ\ZZ^\top } }_{\calS_p}^p \leq \Paren{1 +\bigO{\epsilon}}\norm{\AA-\AA_k}_{\calS_p} + \poly\Paren{\frac{\epsilon}{d}} \sigma_{k+1} \leq \Paren{1 +\bigO{\epsilon}}\norm{\AA-\AA_k}_{\calS_p},
\end{equation*}
which when substituted into Equation \eqref{eqn:tri_ineq_schatten_p-up} concludes the analysis.

Next, we analyze the running time and matrix-vector products. Running Algorithm \ref{algo:simul_power_iter}  with block size $k$ for $q = \bigO{\log(d)p^{1/6}/\epsilon^{1/3}}$ iterations requires $\bigO{\frac{\nnz(\AA)k p^{1/6} \log(d)}{\epsilon^{1/3}} }$ time and $\bigO{ \frac{k p^{1/6} \log(d)}{\epsilon^{1/3}} }$ matrix-vector products.
Similarly, running with block size $\bigO{k/\Paren{\epsilon p}^{1/3}}$ for $q = \bigO{\log(d/\epsilon)\sqrt{p}}$ iterations  requires $\bigO{\frac{\nnz(\AA)k p^{1/6} \log(d/\epsilon)} {\epsilon^{1/3}} }$ time and $\bigO{ \frac{k p^{1/6} \log(d)}{\epsilon^{1/3}} }$ matrix-vector products. Finally, we observe that to obtain a $\Paren{1+1/p}$-approximation to $\sigma_{1}$ and $\sigma_{k+1}$, we need $\bigO{\log(d)\sqrt{p}}$ iterations with blocksize $k+1$ and this requires $\bigO{\log(d) \sqrt{p} k} $ matrix-vector products.
Note, our setting of the exponent of $p$ and $\epsilon$ was chosen to balance the two cases, and this concludes the proof.

% \abnote{Here we can just repeat the above analysis for large p, but instead running $\epsilon^{1/2}$ suffices to get a better bound in 4.17.}
% Finally, we consider the case where $p > \log(1/\epsilon)/(6\log(c))$, and here we simply run Algorithm \ref{algo:simul_power_iter} with block size $k$ for $q = \bigO{\log(d)/\epsilon^{1/2}}$ and it follows from Lemma \ref{lem:gap_independent} that with probability at least $99/100$, the resulting matrix $\ZZ$ satisfies that for all $i\in[k]$, $\norm{\AA \ZZ_{*,i} }_2^2 \geq \sigma_i^2 - \epsilon \sigma_{k+1}^2$. When $\sigma_1 < 2\sigma_{k+1}$, applying Lemma \ref{lem:correlated_projections} with $\QQ = \ZZ \ZZ^\top$ and $\PP$ being the projection on the column span of $\AA \ZZ\ZZ^\top$, we again have 
% \begin{equation}
% \begin{split}
%     \norm{\PP  \AA \QQ }_{\calS_p}^p = \norm{ \AA \ZZ \ZZ^\top }_{\calS_p}^p & \geq \norm{  \AA_k }_{\calS_p}^p - \sum_{i \in [k]} \bigO{ \epsilon p } \sigma_{k+1}^2 \sigma_{i}^{p-2} \\
%     & \geq \norm{  \AA_k }_{\calS_p}^p -  \bigO{ k \epsilon p }  \sigma_{k+1}^{p} 
% \end{split}
% \end{equation}

% Repeating the argument above, i.e. splitting into the case where $\sigma_1 > 2\sigma_{k+1}$ and  

% completing the proof. 

\end{proof}

\section{Query Lower Bounds}
\label{sec:lower_bound}
Next, we show that the $\epsilon$-dependence obtained by our algorithms for Schatten-$p$ low-rank approximation is optimal in the restricted computation model of matrix-vector products. The matrix-vector product model is defined as follows: given a matrix $\AA$, our algorithm is allowed to make adaptive matrix-vector queries to $\AA$, where one matrix-vector query is of the form $\AA v$, for any $v \in \mathbb{R}^d$.
Our lower bounds are information-theoretic and rely on the hardness of estimating the smallest eigenvalue of a Wishart ensemble, as established in recent work of Braverman,  Hazan, Simchowitz and Woodworth~\cite{braverman2020gradient}. 

We split the lower bounds into the case of $p\in[1,2]$ and $p>2$. For $p\in[1,2]$, we have a simple argument based on the Araki-Lieb-Thirring inequality (Fact \ref{fact:alt_ineq}), whereas for $p>2$, our lower bounds require an involved argument using a norm compression inequality for partitioned operators (Fact \ref{fact:aligned-norm-compression}).

\subsection{Lower Bounds for $p \in [1,2]$}

The main lower bound we prove in this sub-section is as follows: 

% \begin{fact}[Mahler's Orthogonal Operator Inequality, Theorem 1.7 in \cite{maher1990some}]
% \label{fact:mahler_ortho_ineq}
% Given $p \geq 2$, and matrices $\PP$ and $\QQ$ such that $\langle \PP, \QQ\rangle_H =0$, the following inequality holds: 
% \[ \norm{\PP }^p_{\calS_p}+\norm{\QQ }^p_{\calS_p} \leq \norm{\PP +\QQ}^p_{\calS_p}. \]
% \end{fact}

\begin{theorem}[Query Lower Bound for $p \in \textrm{[} 1,2 \textrm{]}$]
\label{thm:mv_lowerboundn_schatten_12}
Given $\eps >0$, and $p \in [1,2]$, there exists a distribution $\mathcal{D}$ over $n \times n$ matrices such that for $\AA \sim \mathcal{D}$, any randomized algorithm that with probability at least $9/10$ outputs a  rank-$1$ matrix $\BB$ such that $\norm{ \AA - \BB}^p_{\calS_p} \leq (1+\eps)\norm{\AA - \AA_1}^p_{\calS_p}$ must make $\Omega(1/\eps^{1/3})$ matrix-vector queries to $\AA$.  
\end{theorem}

%We prove the aforementioned theorem by splitting into $p\leq 2$

We require the following theorem on the hardness of computing the minimum eigenvalue of a Wishart Matrix, introduced recently by  Braverman, Hazan, Simchowitz and Woodworth~\cite{braverman2020gradient}:

\begin{theorem}[Computing Min Eigenvalue of Wishart, Theorem 3.1 \cite{braverman2020gradient}]
\label{thm:min_eigenvalue}
Given $\epsilon \in (0,1)$, there exists a function $\mathbf{d}: (0,1) \to \mathbb{N}$ such that for all $d \geq \mathbf{d}(\epsilon) $, the following holds. Let $\WW \sim \textrm{Wishart}(d)$ be a Wishart matrix and $\{\lambda_i \}_{i \in [d]}$ be the eigenvalues of $\WW$, in descending order. Then, there exists a universal constant $c^*$  such that:
\begin{enumerate}
    \item \label{wish:cond1} Let $\zeta_1$ be the event that $\lambda_{d}(\WW) \leq c_1/d^2$, $\zeta_2$ be the event that $\lambda_{d-1}(\WW) - \lambda_{d}(\WW) \geq c_2/d^2$ and $\zeta_3$ be the event that $\norm{\WW}_{op} \leq 5$, where $c_1$ and $c_2$ are constants that depend only on $\epsilon$. Then, $\Pr_{\WW}\left[ \zeta_1 \cap \zeta_2 \cap \zeta_3 \right]\geq 1- \frac{c^*\sqrt{\epsilon} }{2}$. 
    \item\label{wish:cond2} Any randomized algorithm that makes at most $(1-\epsilon)d$ adaptive matrix-vector queries and outputs an estimate $\hat{\lambda}_{d}$ must satisfy 
    \[
    \Pr_{\WW}\left[\abs{\hat{\lambda}_d - \lambda_d}  \geq \frac{1}{4d^2} \right] \geq c^* \sqrt{\epsilon}.
    \]
\end{enumerate}
\end{theorem}

We also use the following lemma  from \cite{braverman2020gradient} bounding the minimum eigenvalue of a Wishart ensemble:

\begin{lemma}[Non-Asymptotic Spectra of Wishart Ensembles, Corollary 3.3 \cite{braverman2020gradient}]
Let $\WW\sim \textsf{Wishart}(n)$ be such that $n = \Omega(1/\eps^3)$.
% \kcnote{$\textsf{Wishart}(d)$ before, now $\textsf{Wishart}(n)$.} \abnote{it was Wishart d in just the theorem statement, invoking it with specific $n$.}
Then, there exists a universal constant $c_2 > 0$ such that
\begin{equation*}
    \Pr\left[\lambda_n\Paren{\WW} \geq \frac{1}{n^2} \right] \geq c_2, \hspace{0.2in} \textrm{and} \hspace{0.2in} \Pr\left[\lambda_n\Paren{\WW} < \frac{1}{2 n^2} \right] \geq \frac{c_2}{2}.
\end{equation*}
\end{lemma}

% \abnote{the non-asymptotic bound holds for some large enough $n$, the exact function does not follow from Braverman et al. }

We are now ready to prove Theorem \ref{thm:mv_lowerboundn_schatten_12}. Our high level approach is to show that we can take any solution that is a $(1+\eps)$-relative-error Schatten-$p$ low-rank approximation to the hard instance $\II -\frac{1}{5} \WW$, where $\WW$ is a Wishart ensemble, and extract from it an accurate estimate of the minimum eigenvalue of $\WW$, thus appealing to the hardness stated in (2) of Theorem \ref{thm:min_eigenvalue} above.

\begin{proof}[Proof of Theorem \ref{thm:mv_lowerboundn_schatten_12}]
Let $n = \Theta\Paren{1/\epsilon^{1/3}}$ and let $\AA = \II - \frac15\WW$ be an $n \times n$ instance where $\WW\sim \textsf{Wishart}(n)$. Let $\zeta_1$ be the event that $\norm{\WW}_{\textrm{op}}\leq 5$. It follows from Fact \ref{fact:norms_of_wishart} that $\zeta_1$ holds with probability at least $99/100$, and we condition on this event. 
% \kcnote{I take it that $\WW\sim \textsf{Wishart}(n)$ such that $n = \Omega(1/\eps^3)$?,
% and that these $\zeta$s have nothing to do with theorem \ref{thm:min_eigenvalue}? 
% Is there a $\zeta_1$?}
Let $\zeta_2$ be the event that $\lambda_n\Paren{\WW} \geq \frac{1}{n^2} = \frac{\eps^{2/3}}{c^*}$ and $\zeta_3$ be the event that $\lambda_n\Paren{\WW} < \frac{1}{2 n^2} = \frac{ \eps^{2/3}}{2c^* }$.
%\abnote{there is a probability analysis missing here. we need a non-asymptotic bound for this events being true. } 

Then, conditioning on $\zeta_2$, we have that 
\begin{equation}
    1- \frac{1}{5} \lambda_n(\WW) \leq 1- \frac{\eps^{2/3}}{5c^*} .
\end{equation}
Similarly, conditioning on $\zeta_3$, we have that
\begin{equation}
    1- \frac{1}{5} \lambda_n(\WW) \geq 1- \frac{\eps^{2/3}}{10c^*} .
\end{equation}

% Following the proof of Theorem \ref{thm:query_lower_bound_rank_1} we can establish that given $\WW \sim \textsf{Wishart}(n)$, where $n = \sqrt{c^*}/\epsilon^{1/3}$, for a universal constant $c^*$, we have that  
% \begin{equation}
% 1- \frac{\eps^{2/3}}{10c^*}   \leq 1- \frac{1}{5} \lambda_n(\WW) \leq 1- \frac{\eps^{2/3}}{5c^*} 
% \end{equation}
We observe that for $p\in [1,2]$,
% \kcnote{does the expansion below make sense for $p\in [1,2]$, for example $p=1.5$?}
using Bernoulli's inequality (Fact \ref{fact:bernoulli}) we have 
\begin{equation*}
    \Paren{ 1 - \frac{1}{5}\lambda_n(\WW)}^p \geq 1 -  \frac{p}{5}\lambda_n(\WW)
\end{equation*}
and since $(1-x)^p \leq (1-x)$ for any $x \in (0,1)$, we also have that,
\begin{equation*}
    \Paren{ 1 - \frac{1}{5}\lambda_n(\WW)}^p \leq 1 -  \frac{1}{5}\lambda_n(\WW)
\end{equation*}
Therefore, we can conclude, $\Paren{ 1 - \frac{1}{5}\lambda_n(\WW)}^p =1- \Theta\Paren{  \lambda_n(\WW) }$.
% \begin{equation}
% \begin{split}
%  \Paren{ 1 - \frac{1}{5}\lambda_n(\WW)}^p &= 1- \frac{p}{5}\lambda_n(\WW) + \sum^{p}_{i =2} {p \choose i} \Paren{-\frac{1}{5}\lambda_n(\WW) }^{i}  \\
%  & = 1- \Theta\Paren{  \lambda_n(\WW) }.
% \end{split}
% \end{equation}
Further, it follows from part (1) of Fact \ref{fact:norms_of_wishart} that $0 \preceq \II-\frac{1}{5}\WW  \preceq \II$, and thus\kcnote{If you don't mind, I'll define $\AA= \II-\frac{1}{5}\WW$, as is done in $p\ge 2$ case, and use it in the below, making it shorter and a little easier to read.}
\begin{equation}
\label{eq:bound-schatten-p}
    \norm{\AA }^p_{\calS_p} = \sum_{i \in [n]} \lambda_i^p\Paren{\II-\frac{1}{5}\WW} \leq \sum_{i \in [n]} \lambda_i\Paren{\II-\frac{1}{5}\WW} \leq \bigO{\frac{1}{\epsilon^{1/3}}}
\end{equation}
where the last inequality follows from the fact that $n = \sqrt{c^*}/\epsilon^{1/3}$. Let $\AA_1$ denote the best rank-$1$ approximation to $\AA$.
Then, it follows from Equation \eqref{eq:bound-schatten-p} that
%\kcnote{What does $\Paren{\II-\frac{1}{5}\WW}_1$ mean?}
%\abnote{defined to be best rank 1 approx now. }
\begin{equation}
    \epsilon \norm{\AA - \AA_1 }^p_{\calS_p} \leq \epsilon \norm{\AA  }^p_{\calS_p} \leq \bigO{\epsilon^{2/3}}
\end{equation}
Observe, any $(1+\epsilon)$-approximate relative-error Schatten-$p$ low-rank approximation algorithm for $k =1$ outputs a matrix $v v^\top$ such that 
% \kcnote{"low-rank" should be "rank-one", probably, and $vv^\top$, as below, not $uv^\top$.}
\begin{equation}
\label{eqn:upper_bound_sp}
\begin{split}
    \norm{ \AA \Paren{ \II - vv^\top}}^p_{\calS_p} & \leq (1+\epsilon) \norm{ \AA -\AA_1 }^p_{\calS_p} \\
    & \leq \norm{ \AA }^p_{\calS_p} - \norm{\AA }^p_{\textrm{op}} + \Theta(\epsilon^{2/3})
\end{split}
\end{equation}
 By definition of the Schatten-$p$ norm we have:
\begin{equation}
\label{eqn:lower_bound_sp}
    \begin{split}
        \norm{ \AA\Paren{ \II - vv^\top}}^p_{\calS_p} & = \trace{\Paren{\Paren{ \II - vv^\top}^2 \AA^2\Paren{ \II - vv^\top}^2}^{p/2}} \\
        & \geq \trace{\Paren{ \II - vv^\top}^p \AA^p\Paren{ \II - vv^\top}^p}\\
        & = \trace{ \AA^p - \AA^p vv^\top } \\
        & = \norm{  \AA }^p_{\calS_p} -  \trace{ \Paren{ vv^\top}^{p/2} \Paren{\AA^2}^{p/2}  \Paren{ vv^\top}^{p/2} } \\
        & \geq \norm{  \AA }^p_{\calS_p} -  \trace{ \Paren{  vv^\top \AA^2   vv^\top}^{p/2} } \\
        & = \norm{  \AA }^p_{\calS_p} -  \norm{\AA vv^\top }^p_{\calS_p} \\
        & = \norm{  \AA }^p_{\calS_p} -  \norm{\AA v}^p_{2}
    \end{split}
\end{equation}
where the first and last inequality follows from the reverse Araki-Lieb-Thirring inequality (Fact \ref{fact:alt_ineq}). Combining equations \eqref{eqn:upper_bound_sp} and \eqref{eqn:lower_bound_sp}, we have that 
\begin{equation}
\label{eqn:combined_bound_on_pth_power}
    \norm{ \AA  }^p_{\textrm{op}} \geq \norm{\AA v}^p_{2}\geq \norm{\AA }^p_{\textrm{op}} - \Theta(\epsilon^{2/3})
\end{equation}

Next, we observe that $\AA v = \Paren{\II - 1/5 \WW}v$ can be computed with one additional matrix-vector product and 
\begin{equation}
\label{eqn:min_eigenvalue}
    \norm{\AA  }_{\textrm{op}}^p = \Paren{1- \frac{1}{5}\lambda_n(\WW) }^p = 1 -  \frac{p}{5}\lambda_n(\WW) + \bigO{ \lambda^2_n(\WW)  }
\end{equation}
Consider the estimator $\hat{\lambda}(\WW) = \frac{5}{p}\Paren{ 1-\norm{ \Paren{ \II-\frac{1}{5}\WW }v }_{2}^p}$. Combining equations   \eqref{eqn:combined_bound_on_pth_power} and \eqref{eqn:min_eigenvalue}, we can conclude
\begin{equation*}
    \hat{\lambda}(\WW) = \lambda_{\min}(\WW) \pm \Theta(\epsilon^{2/3}).
\end{equation*}
obtaining an additive error estimate to the minimum eigenvalue of $\WW$ by computing an additional matrix-vector product. 
It follows that we satisfy conditions \eqref{wish:cond1} and \eqref{wish:cond2} in Theorem \ref{thm:min_eigenvalue} and thus any algorithm for computing a rank-$1$ approximation to the matrix $\AA = \II - \frac{1}{5}\WW$ in Schatten $p$ norm must make at least $\frac{1}{\epsilon^{1/3}}$ queries to the aforementioned matrix, completing the proof.
The claim follows from Theorem \ref{thm:min_eigenvalue}.
\end{proof}

\subsection{ Lower Bound for $p>2$}
\label{sec:lower_bound_p_large}

We now consider the case when $p>2$. We note that the previous approach no longer works since we cannot lower bound the cost of $\|\Paren{\II - \WW/5}\Paren{\II - vv^\top}\|_{\calS_p}$, as the Araki-Lieb-Thirring inequality reverses (see application in Equation \ref{eqn:lower_bound_sp}). Therefore, we require a new approach, and appeal to a special case of Conjecture \ref{conj:norm_compression} that is known to be true, i.e. the Aligned Norm Compression inequality (see Fact \ref{fact:aligned-norm-compression}).  The main theorem we prove in this sub-section is as follows: 

\begin{theorem}[Query Lower Bound for $p>2$]
\label{thm:query_lower_bound_p>2}
Given $\eps >0$, and $p \geq 2$ such that $p = \bigO{1}$, there exists a distribution $\mathcal{D}$ over $n \times n$ matrices such that for $\AA \sim \mathcal{D}$, any randomized algorithm that with probability at least $99/100$ outputs a unit vector $u$  such that $\norm{ \AA - \AA uu^\top}^p_{\calS_p} \leq (1+\eps)\norm{\AA - \AA_1}^p_{\calS_p}$ must make $\Omega\Paren{1/\eps^{1/3}}$ matrix-vector queries to $\AA$.
\end{theorem}

We first introduce a sequence of key lemmas required for our proof.

\begin{corollary}[Special Case of Lemma \ref{lem:gap_independent}]
\label{lem:bound_convergence_in_gamma_iter}
 Given $\gamma \in [0,1]$, a vector $v \in \mathbb{R}^d$ and an $n \times d$ matrix $\AA$, let $t = \log(n/\gamma)/(c \sqrt{ \gamma} )$, for a fixed universal constant $c$. Then, there exists an algorithm that computes $t$ matrix-vector products with $\AA$ and outputs a unit vector $u$ such that with probability at least $99/100$,  
\begin{equation*}
    \norm{\AA}_{\textrm{op}}^2 - \norm{\AA u}_2^2 \leq O\Paren{ \gamma \sigma_2^2 }.
\end{equation*}
where $\sigma_2$ is the second largest singular value of $\AA$. 
\end{corollary}

Next, we prove a key lemma relating the norm of a matrix to norms of orthogonal projections applied to the matrix. We note that this lemma is straight forward and holds for arbitrary vectors unit $u, v$ if Conjecture \ref{conj:norm_compression} holds. However, we show that we can transform our matrix to have structure such that we can apply Fact \ref{fact:aligned-norm-compression} instead.

\begin{lemma}[Orthogonal Projectors to Block Matrices ]
\label{lem:orthogonal_proj_block_matrices}
Given an $n\times d$ matrix $\AA$, $p > 2$ and unit vectors $u \in \mathbb{R}^d, v  \in \mathbb{R}^n$, such that $\Paren{\II- vv^\top }\AA u u^\top =0$. Then,   we have 
\begin{equation*}
    \norm{\AA }_{\calS_p} \leq   \norm{ \begin{pmatrix}
        \norm{ vv^\top \AA uu^\top}_{\calS_p}  & \norm{vv^\top  \AA \Paren{ \II - uu^\top}  }_{\calS_p} \\\
        0 & \norm{\Paren{ \II - vv^\top} \AA \Paren{ \II - uu^\top}}_{\calS_p} 
    \end{pmatrix} }_{\calS_p} .
\end{equation*}
\end{lemma}
\begin{proof}
% This statement appears in the proof of Lemma \ref{lem:schatten_p_orthogonal_projections} for arbitrary projection matrices applied to $\AA$.

Let $\II - vv^\top = \YY \YY^\top$, where $\YY$ has $n -1$ orthonormal columns. Further, since $v$ and $\YY$ span disjoint subspaces, and the union of their span is $\mathbb{R}^n$, the matrix $\Paren{  v \mid \YY }$, obtained by concatenating their columns is unitary. Then, let $\RR = \Paren{  v \mid \YY }^\top$ and observe, $\RR$ has orthonormal rows and columns (since $\RR$ is unitary).  
Next, let  $\II - uu^\top = \ZZ \ZZ^\top$, where $\ZZ$ is $d \times (d-1)$ and has orthonormal columns. Let $\SS = \Paren{ u  \mid \ZZ }^\top$,  and observe $\SS$ has orthonormal rows and columns. 

Let $\hat{\AA} = \RR \AA \SS^\top$, which admits the following block-matrix form:

\begin{equation*}
\begin{split}
    \hat{\AA} = \begin{pmatrix}v^\top \\
    \YY^\top \end{pmatrix} \cdot \AA \cdot  \Paren{ u \mid \ZZ  } =\begin{pmatrix}v^\top \\
    \YY^\top \end{pmatrix}  \Paren{ \AA u \mid \AA \ZZ  } = \begin{pmatrix} v^\top \AA u & v^\top \AA \ZZ \\
    \YY^\top \AA u & \YY^\top \AA \ZZ   \end{pmatrix}
\end{split}
\end{equation*}
Since $\RR$ and $\SS$ are unitary, it follows from unitary invariance of the Schatten-$p$ norm that
\begin{equation}
    \norm{\AA}_{\calS_p} = \norm{\hat{\AA}}_{\calS_p} = \norm{ \begin{pmatrix} v^\top \AA u & v^\top \AA \ZZ \\
    \YY^\top \AA u & \YY^\top \AA \ZZ   \end{pmatrix}  }_{\calS_p} = \norm{ \begin{pmatrix} v^\top \AA u & v^\top \AA \ZZ \\
    0 & \YY^\top \AA \ZZ   \end{pmatrix}  }_{\calS_p} ,
\end{equation}
where the last equality follows from observing that $\norm{ \YY^\top \AA u }_F = \norm{\YY \YY^\top \AA u u^\top }_F = \norm{\Paren{\II - vv^\top} \AA u u^\top }_F=0$ and therefore $\YY^\top \AA u$ is a matrix of all $0$s. Next, we append a set of $d -2$ columns of $0$'s to make the top left and top right block the same size. Since this does not change the singular values, we have 
\begin{equation}
    \norm{\AA}_{\calS_p}= \norm{ \begin{pmatrix} v^\top \AA u & 0 & v^\top \AA \ZZ \\
    0 & 0 &  \YY^\top \AA \ZZ   \end{pmatrix}  }_{\calS_p}
\end{equation}
Next, we construct a rotation matrix $\RR$ such that on right multiplying a row vector by $\RR$, the first $d-1$ coordinates remain the same and on the remaining coordinates, the vector $v^\top \AA \ZZ $ gets mapped to $c e_1^\top$ for some scalar $c$. Let $\SS$ be the $d-1 \times d -1$ rotation matrix such that $v^\top \AA \ZZ  \SS = c e_1^\top$. Then, $\RR = \begin{pmatrix} \II & 0 \\
0 & \SS \end{pmatrix}$ and it is easy to verify that $\RR$ is unitary. Therefore,
\begin{equation*}
    \begin{pmatrix} v^\top \AA u & 0 & v^\top \AA \ZZ \\
    0 & 0 &  \YY^\top \AA \ZZ   \end{pmatrix} \cdot \RR = \begin{pmatrix} v^\top \AA u & 0 & c e_1^\top \\
    0 & 0 &  \YY^\top \AA \ZZ \SS   \end{pmatrix}   
\end{equation*}
Now, we observe the final matrix above has a block matrix form we can apply the Aligned Norm Compression inequality from Fact \ref{fact:aligned-norm-compression}, with $\alpha_1 = v^\top\AA u$, $\alpha_2 = c$, $\beta_1  = 0$ and $\beta_2 = 0$, and therefore 

\begin{equation}
\begin{split}
    \norm{\AA}_{\calS_p}= \norm{\begin{pmatrix} v^\top \AA u & 0 & c e_1^\top \\
    0 & 0 &  \YY^\top \AA \ZZ \SS\end{pmatrix}    }_{\calS_p} & \leq  \norm{\begin{pmatrix} \norm{ v^\top \AA u}_{\calS_p} & 0 & \norm{c e_1^\top}_{\calS_p} \\
    0 & 0 &  \norm{ \YY^\top \AA \ZZ \SS}_{\calS_p} \end{pmatrix}    }_{\calS_p} \\
    & = \norm{ \begin{pmatrix} \norm{ v v^\top \AA u u^\top }_{\calS_p} & \norm{v v^\top \AA \ZZ \ZZ^\top}_{\calS_p} \\
    0& \norm{\YY \YY^\top \AA \ZZ\ZZ^\top }_{\calS_p}   \end{pmatrix}  }_{\calS_p}
\end{split}
\end{equation}
where the last equality follows from unitary invariance and substituting the definition of $\YY\YY^\top$ and $\ZZ \ZZ^\top$ completes the proof. 
% Applying Conjecture \ref{conj:norm_compression} with $\MM_1= v^\top \AA u $, $\MM_2 = v^\top \AA \ZZ$, $\MM_3 = \YY^\top \AA u$ and $\MM_4 = \YY^\top \AA \ZZ $, we have 
% \begin{equation}
% \begin{split}
%     \norm{ \begin{pmatrix} v^\top \AA u & v^\top \AA \ZZ \\
%     \YY^\top \AA u & \YY^\top \AA \ZZ   \end{pmatrix}  }_{\calS_p} & \leq \norm{ \begin{pmatrix} \norm{ v^\top \AA u}_{\calS_p} & \norm{v^\top \AA \ZZ}_{\calS_p} \\
%     \norm{\YY^\top \AA u}_{\calS_p} & \norm{\YY^\top \AA \ZZ}_{\calS_p}   \end{pmatrix}  }_{\calS_p}\\
%     & = \norm{ \begin{pmatrix} \norm{ v v^\top \AA u u^\top }_{\calS_p} & \norm{v v^\top \AA \ZZ \ZZ^\top}_{\calS_p} \\
%     \norm{\YY \YY^\top \AA u u^\top}_{\calS_p} & \norm{\YY \YY^\top \AA \ZZ\ZZ^\top }_{\calS_p}   \end{pmatrix}  }_{\calS_p} 
% \end{split}
% \end{equation}

% %%%% Constraints for rotation $\RR$

% % \kcnote{This paragraph and the next are garbled; seems like you're already done.}?\abnote{we are trying to get rid of the dependence on this conjecture, and we have an outline which is this paragraph here.}
% \paragraph{\color{red} sketch of conjecture removal.}
% Left mult LHS above by $\RR$ such that $R (v^\top \AA u \mid 0 \mid Y^\top \AA u )^\top  \to   \Paren{z\mid z}^\top$ for some vector $z$ and $R(v^\top \AA \ZZ \mid 0 \mid Y^\top \AA \ZZ )  \to \Paren{0\mid \BB}$ for some matrix $\BB$. and the claim is $R = \Paren{ \II_d , 0 ;0 , \SS}$ for some rotation $\SS$ such that $S\YY^\top \AA u = c e_1$.
% %%the above is potientially a new proof.
% %%%%%%%%%%%%%

\end{proof}

\begin{fact}[SVD of a $2\times 2$ Matrix]
\label{fact:singular_vals_2x2}
Given a $2 \times 2$ matrix $\MM =\begin{pmatrix}
        a& b \\
        c & d \end{pmatrix} $ let $\UU \Sig \VV^\top$ be the SVD of $\MM$. Then,
        % \kcnote{I don't see why it's not $(a^2 + b^2 + c^2 + d^2)^2$ inside the inner sqrt, from the quadratic formula.}
\begin{equation*}
    \Sig_{1,1} = \sqrt{ \frac{ a^2 + b^2 + c^2 + d^2 + \sqrt{ \Paren{a^2 + b^2 - c^2 - d^2 }^2 + 4 \Paren{ac +bd }^2 }  }{2} } ,
\end{equation*}
and
\begin{equation*}
    \Sig_{2,2} =\sqrt{ \frac{ a^2 + b^2 + c^2 + d^2 - \sqrt{ \Paren{a^2 + b^2 - c^2 - d^2 }^2 + 4 \Paren{ac +bd }^2 }  }{2} } .
\end{equation*}
\end{fact}

Now, we are ready to prove Theorem \ref{thm:query_lower_bound_p>2}. 

\begin{proof}[Proof of Theorem \ref{thm:query_lower_bound_p>2}]
Let $\AA = \II - \frac{1}{5}\WW$ where $\WW$ is an $n\times n$ Wishart matrix as in the proof
of Theorem~\ref{thm:mv_lowerboundn_schatten_12}
%\kcnote{$\GG$ hasn't been used up to now in the lower bounds section}
and we have by hypothesis
% \kcnote{by hypothesis, yes?} \abnote{yes}\kcnote{re-worded it slightly,then}
that there is an algorithm that with probability at least $99/100$, outputs a unit vector $u$ such that $\norm{\AA \Paren{\II - uu^\top } }_{\calS_p}^p \leq (1+\eps) \norm{\AA - \AA_1 }_{\calS_p}^p $. Let $v = \AA u/\norm{\AA u}_2$ and  observe, $\Paren{ \II - vv^\top} \AA uu^\top =0$. Further, by the unitary invariance of the Schatten-$p$ norm, 
\begin{equation}
\label{eqn:term_a}
    \norm{vv^\top \AA uu^\top}_{\calS_p} =  \abs{v^\top \AA  u } = \frac{ \abs{u^\top\AA^\top \AA  u }} {\norm{\AA u}_2 } = \norm{\AA u}_2.
\end{equation}
Similarly, 
\begin{equation}
\label{eqn:term_c}
\begin{split}
    \norm{ vv^\top \AA \Paren{\II -  uu^\top } }_{\calS_p} & =  \sqrt{ \norm{ v^\top \AA \Paren{\II -  uu^\top } }_{2}^2  } = \sqrt{ \norm{v^\top \AA}_2^2 - \norm{ v^\top \AA uu^\top}_2^2  }  \\
    & = \sqrt{ \frac{ \norm{u^\top \AA^\top \AA}_2^2}{\norm{\AA u}_2^2 } - \norm{\AA u}_2^2  } \\
    & \leq \sqrt{ \frac{ \norm{u^\top \AA^\top }_2^2\cdot \norm{\AA}_{\textrm{op}}^2 }{\norm{\AA u}_2^2 } - \norm{\AA u}_2^2  } \\
    & \leq \epsilon^{1/3}\sigma_2, 
\end{split}
\end{equation}
where we use sub-multiplicativity of the $\ell_2$ norm and Corollary \ref{lem:bound_convergence_in_gamma_iter} with $\gamma = \epsilon^{2/3}$. Note that we can assume w.l.o.g. that Corollary \ref{lem:bound_convergence_in_gamma_iter} holds since we can just iterate Block Krylov $q = (1/c\epsilon^{1/3})$ times, for a sufficiently large constant $c$, starting the iterations with the vector~$u$ output by the algorithm hypothesized for the theorem, and pay only $(1/c\epsilon^{1/3})$ extra matrix-vector products.
% \kcnote{Sorry, what? We take $u$ produced by some hypothesized algorithm, and then assume that it satisfies a condition for the output of the possibly different algorithm of the corollary. Iterate what for $(1/c\eps^{1/3})$ ? What starting vector?}
% \abnote{right, so we know the vector that is output by a candidate algorithm satisfies the LRA guarantee. But we also need it to satisfy corollary 6.5. so we just run Block Krylov again with the starting vector being the one output by the candidate algorithm. this is now being done as a part of the reduction. }
% \kcnote{So the wording might be: ``...we can just iterate Block Krylov $q = (1/c\epsilon^{1/3})$ times, for a sufficiently large constant $c$, with the starting vector $u$ output by the algorithm hypothesized for the theorem, and pay only $(1/c\epsilon^{1/3})$ extra matrix-vector products.''}
% \abnote{that's good, feel free to edit it. }
Since $vv^\top \AA + \AA uu^\top  - vv^\top \AA u u^\top $ has rank at most $3$, 
\begin{equation}
\label{eqn:term_d}
\begin{split}
\norm{\Paren{ \II - vv^\top} \AA \Paren{ \II - uu^\top} }_{\calS_p}^p &  =  \norm{ \AA - vv^\top \AA - \AA uu^\top +  vv^\top \AA u u^\top  }^p_{\calS_p} \\
& \geq \norm{\AA - \AA_3 }^p_{\calS_p}\\
& = \Omega\Paren{ \frac{1}{\epsilon^{1/3}} } ,
\end{split}
\end{equation}
where the last inequality follows from  Fact \ref{fact:norms_of_wishart}.
\kcnote{Not sure where $\norm{\AA }_{\calS_p}^p = \Omega_p(1/\epsilon^{1/3})$ comes from}\abnote{fixed.} \kcnote{At last, use of Fact \ref{fact:norms_of_wishart}.2! Although it was an upper bound before, I think, and this is a lower bound.}
\abnote{modified the fact to be $\Theta$. }

%  $\norm{\Paren{ \II - vv^\top} \AA \Paren{ \II - uu^\top} }_{\calS_p} = \Omega(1/\epsilon^{1/3})$ (this uses interlacing).  

Let 
% \kcnote{typo? The off-diagonal entries of $\MM$ are equal, and neither is the term that is zero.} 
$\MM=\begin{pmatrix}
        a& b \\
        c & d \end{pmatrix}  = \begin{pmatrix}
        \norm{vv^\top \AA uu^\top}_{\calS_p}  & \norm{vv^\top  \AA  \Paren{\II - uu^\top} }_{\calS_p}  \\\
        \norm{ \Paren{\II - vv^\top} \AA uu^\top }_{\calS_p} & \norm{\Paren{ \II - vv^\top} \AA \Paren{ \II - uu^\top} }_{\calS_p}
    \end{pmatrix}^\top$. Then, it follows from Fact \ref{fact:singular_vals_2x2} that 

\begin{equation}
\label{eqn:first_singular_value}
\begin{split}
    \Sig_{1,1}\Paren{\MM} & = \frac{1 }{\sqrt{ 2}  } \cdot  \sqrt{a^2  + c^2 + d^2 + \sqrt{ \Paren{a^2  - c^2 - d^2 }^2 + 4 \Paren{ac }^2 }   } \\
    & = \frac{1 }{\sqrt{ 2}  } \cdot  \sqrt{ a^2  + c^2 + d^2 +  \Paren{c^2 + d^2 - a^2 }  + \Theta\Paren{ \frac{4 a^2 c^2 }{c^2 + d^2 - a^2  } }    }  \\
    & = \sqrt{ c^2 + d^2    + \Theta\Paren{\frac{  a^2 c^2  }{c^2 + d^2 - a^2 }} } ,
\end{split}
\end{equation}
where
% \kcnote{Here and below, where did the $\frac{1 }{\sqrt{ 2}  }$ go?}
we use that $b=0$,  $c, a \leq 1$ and $1 \ll d$ and the Taylor expansion of $\sqrt{x + y}$ for $x,y\geq0$. Similarly, 

\begin{equation}
\label{eqn:second_singular_value}
    \Sig_{2,2}\Paren{\MM}   = \sqrt{ a^2    - \Theta\Paren{\frac{  a^2 c^2  }{c^2 + d^2 - a^2 }} }.
\end{equation}
Then, using equations \eqref{eqn:first_singular_value} and  \eqref{eqn:second_singular_value} we can bound the Schatten-$p$ norm of $\MM$ as follows:

\begin{equation}
\label{eqn:schatten_p_of_a}
    \norm{\MM }_{\calS_p}^p \leq   \underbrace{  \Paren{ c^2 + d^2    + \Theta\Paren{\frac{  a^2 c^2  }{c^2 + d^2 - a^2 }}  }^{p/2} }_{\ref{eqn:schatten_p_of_a}.1} + \underbrace{ \Paren{a^2    - \Theta\Paren{\frac{  a^2 c^2  }{c^2 + d^2 - a^2 }} }^{p/2}}_{\ref{eqn:schatten_p_of_a}.2}  .
\end{equation}
We now bound each of the terms above. Consider the first term: 
\begin{equation}
\begin{split}
    \Paren{ c^2 + d^2    + \Theta\Paren{\frac{  a^2 c^2  }{c^2 + d^2 - a^2 }}  }^{p/2} & =  \Bigg( \norm{vv^\top \AA \Paren{\II - uu^\top } }_{\calS_p}^2  \\
    & \hspace{0.2in} + \norm{\Paren{ \II - vv^\top} \AA \Paren{ \II - uu^\top} }_{\calS_p}^2    
    + \Theta\Paren{\eps^{2/3} \norm{\AA u}_2^2 }  \Bigg)^{p/2}\\
    & \leq  \Paren{ \Theta\Paren{\eps^{2/3}}  + \norm{ \AA \Paren{ \II - uu^\top} }_{\calS_p}^2 }^{p/2} \\
    & \leq \Paren{ 1+\bigO{ \epsilon^{2p/3}}  } \norm{ \AA - \AA_1 }_{\calS_p}^p,
\end{split}
\end{equation}
where we use equation \eqref{eqn:term_a}, \eqref{eqn:term_c}, and \eqref{eqn:term_d}, and $\norm{ \AA \Paren{ \II - uu^\top} }_{\calS_p}^2 \leq (1+\epsilon)^{2/p} \norm{\AA - \AA_1}_{\calS_p}^2 $. The last inequality follows from observing that
\begin{equation*}
     \eps^{2/3} \leq \bigO{  \eps^{4/3} \cdot \frac{1}{\epsilon^{2/3p }}} \leq  \bigO{ \eps^{4/3} \cdot \norm{\AA - \AA_1}_{\calS_p}^2}.
\end{equation*}
We can now bound the second term in Equation \ref{eqn:schatten_p_of_a} as follows: 

\begin{equation}
    \Paren{a^2    - \Theta\Paren{\frac{  a^2 c^2  }{c^2 + d^2 - a^2 }} }^{p/2} = \Paren{ \norm{\AA u}_2^2 -  \Theta\Paren{ \epsilon^{ 2/3 }\norm{\AA u}_2^2 }  }^{p/2} \leq  \norm{\AA u }_2^p  .
\end{equation}
Then, we have

\begin{equation*}
     \norm{\MM }_{\calS_p}^p \leq \Paren{ 1+ \bigO{ \epsilon^{2p/3}}  } \norm{ \AA - \AA_1 }_{\calS_p}^p + \norm{\AA u }_2^p . 
\end{equation*}
It follows from Lemma  \ref{lem:orthogonal_proj_block_matrices}, that  $\norm{\MM }^p_{\calS_p} \geq \norm{\AA }_{\calS_p}^p$ and thus
\begin{equation}
\begin{split}
    \norm{\AA u }_2^p  & \geq \norm{\AA}_{\calS_p}^p - \Paren{ 1+ \bigO{ \epsilon^{2p/3}}  } \norm{ \AA - \AA_1 }_{\calS_p}^p  \\
    & =  \norm{\AA}_{\textrm{op}}^p - \bigO{ \epsilon^{2p/3}}   \norm{ \AA - \AA_1 }_{\calS_p}^p  \\
    & \geq \norm{\AA}_{\textrm{op}}^p -\bigO{ \epsilon    \norm{ \AA - \AA_1 }_{\calS_p}^p }\\
    & \geq \norm{\AA}_{\textrm{op}}^p  - \bigO{ \eps^{2/3} }
\end{split}
\end{equation}
where the second to last inequality follows from recalling $p\geq2$. The remainder of the proof
is as in that following \eqref{eqn:combined_bound_on_pth_power} in the proof of Theorem \ref{thm:mv_lowerboundn_schatten_12}.
\end{proof}

\paragraph{Acknowledgments:}  A. Bakshi and D. Woodruff would like to thank the National Science Foundation under Grant No. CCF-1815840, Office of Naval Research (ONR) grant N00014-18-1-2562, and a Simons Investigator Award. Part of this work was done while A. Bakshi was an intern at IBM Almaden. The authors thank Praneeth Kacham for pointing out an error in a previous version and for suggesting simplified proofs of lemma 5.8 and 5.9. The authors also thank anonymous reviewers for their careful reading of our manuscript and for several insightful suggestions.

\bibliographystyle{alpha}
\bibliography{references}

\newcommand{\etalchar}[1]{$^{#1}$}
\begin{thebibliography}{MMMW21}

\bibitem[ACW17]{avron2017sharper}
Haim Avron, Kenneth~L. Clarkson, and David~P. Woodruff.
\newblock Sharper bounds for regularized data fitting, 2017.

\bibitem[AK12]{audenaert2012problems}
Koenraad~MR Audenaert and Fuad Kittaneh.
\newblock Problems and conjectures in matrix and operator inequalities.
\newblock {\em arXiv preprint arXiv:1201.5232}, 2012.

\bibitem[AN13]{andoni2013eigenvalues}
Alexandr Andoni and Huy~L. Nguyen.
\newblock Eigenvalues of a matrix in the streaming model.
\newblock In {\em Proceedings of the twenty-fourth annual ACM-SIAM symposium on
  Discrete algorithms}, pages 1729--1737. Society for Industrial and Applied
  Mathematics, 2013.

\bibitem[Ara90]{araki1990inequality}
Huzihiro Araki.
\newblock On an inequality of {L}ieb and {T}hirring.
\newblock {\em LMaPh}, 19(2):167--170, 1990.

\bibitem[Aud08]{audenaert2008norm}
Koenraad~MR Audenaert.
\newblock On a norm compression inequality for 2$\times$ {N} partitioned block
  matrices.
\newblock {\em Linear algebra and its applications}, 428(4):781--795, 2008.

\bibitem[Avr10]{avron2010counting}
Haim Avron.
\newblock Counting triangles in large graphs using randomized matrix trace
  estimation.
\newblock In {\em Workshop on Large-scale Data Mining: Theory and
  Applications}, volume~10, pages 10--9, 2010.

\bibitem[BBB{\etalchar{+}}19]{ban2019ptas}
Frank Ban, Vijay Bhattiprolu, Karl Bringmann, Pavel Kolev, Euiwoong Lee, and
  David~P Woodruff.
\newblock A {PTAS} for lp-low rank approximation.
\newblock In {\em Proceedings of the Thirtieth Annual ACM-SIAM Symposium on
  Discrete Algorithms}, pages 747--766. SIAM, 2019.

\bibitem[BBK{\etalchar{+}}21]{bakshi2021learning}
Ainesh Bakshi, Chiranjib Bhattacharyya, Ravi Kannan, David~P Woodruff, and
  Samson Zhou.
\newblock Learning a latent simplex in input-sparsity time.
\newblock {\em arXiv preprint arXiv:2105.08005}, 2021.

\bibitem[BCW20]{bakshi2020robust}
Ainesh Bakshi, Nadiia Chepurko, and David~P Woodruff.
\newblock Robust and sample optimal algorithms for {PSD} low rank
  approximation.
\newblock In {\em 2020 IEEE 61st Annual Symposium on Foundations of Computer
  Science (FOCS)}, pages 506--516. IEEE, 2020.

\bibitem[BDN15]{BDN15}
Jean Bourgain, Sjoerd Dirksen, and Jelani Nelson.
\newblock Toward a unified theory of sparse dimensionality reduction in
  euclidean space.
\newblock In {\em Proceedings of the Forty-Seventh Annual {ACM} on Symposium on
  Theory of Computing, {STOC} 2015, Portland, OR, USA, June 14-17, 2015}, pages
  499--508, 2015.

\bibitem[BFG96]{bai1996some}
Zhaojun Bai, Gark Fahey, and Gene Golub.
\newblock Some large-scale matrix computation problems.
\newblock {\em Journal of Computational and Applied Mathematics},
  74(1-2):71--89, 1996.

\bibitem[Bha13]{bhatia2013matrix}
Rajendra Bhatia.
\newblock {\em Matrix analysis}, volume 169.
\newblock Springer Science \& Business Media, 2013.

\bibitem[BHSW20]{braverman2020gradient}
Mark Braverman, Elad Hazan, Max Simchowitz, and Blake Woodworth.
\newblock The gradient complexity of linear regression.
\newblock In {\em Conference on Learning Theory}, pages 627--647. PMLR, 2020.

\bibitem[BKKS19]{braverman2019schatten}
Vladimir Braverman, Robert Krauthgamer, Aditya Krishnan, and Roi Sinoff.
\newblock Schatten norms in matrix streams: Hello sparsity, goodbye dimension.
\newblock {\em arXiv preprint arXiv:1907.05457}, 2019.

\bibitem[BKL02]{bhatia2002pinchings}
Rajendra Bhatia, William Kahan, and Ren-Cang Li.
\newblock Pinchings and norms of scaled triangular matrices.
\newblock {\em Linear and Multilinear Algebra}, 50(1):15--21, 2002.

\bibitem[BW18]{bakshi2018sublinear}
Ainesh Bakshi and David Woodruff.
\newblock Sublinear time low-rank approximation of distance matrices.
\newblock In {\em Advances in Neural Information Processing Systems}, pages
  3782--3792, 2018.

\bibitem[BWZ16]{boutsidis2016optimal}
Christos Boutsidis, David~P Woodruff, and Peilin Zhong.
\newblock Optimal principal component analysis in distributed and streaming
  models.
\newblock In {\em Proceedings of the forty-eighth annual ACM symposium on
  Theory of Computing}, pages 236--249, 2016.

\bibitem[BWZ19]{ban2019regularized}
Frank Ban, David Woodruff, and Qiuyi Zhang.
\newblock Regularized weighted low rank approximation.
\newblock {\em arXiv preprint arXiv:1911.06958}, 2019.

\bibitem[CCHW20]{chepurko2020quantum}
Nadiia Chepurko, Kenneth~L Clarkson, Lior Horesh, and David~P Woodruff.
\newblock Quantum-inspired algorithms from randomized numerical linear algebra.
\newblock {\em arXiv preprint arXiv:2011.04125}, 2020.

\bibitem[CLMW11]{candes2011robust}
Emmanuel~J Cand{\`e}s, Xiaodong Li, Yi~Ma, and John Wright.
\newblock Robust principal component analysis?
\newblock {\em Journal of the ACM (JACM)}, 58(3):1--37, 2011.

\bibitem[CLW18]{chia2018quantum}
Nai-Hui Chia, Han-Hsuan Lin, and Chunhao Wang.
\newblock Quantum-inspired sublinear classical algorithms for solving low-rank
  linear systems.
\newblock {\em arXiv preprint arXiv:1811.04852}, 2018.

\bibitem[Coh16]{cohen2016nearly}
Michael~B Cohen.
\newblock Nearly tight oblivious subspace embeddings by trace inequalities.
\newblock In {\em Proceedings of the twenty-seventh annual ACM-SIAM symposium
  on Discrete algorithms}, pages 278--287. SIAM, 2016.

\bibitem[CP10]{candes2010matrix}
Emmanuel~J Candes and Yaniv Plan.
\newblock Matrix completion with noise.
\newblock {\em Proceedings of the IEEE}, 98(6):925--936, 2010.

\bibitem[CR09]{candes2009exact}
Emmanuel~J Cand{\`e}s and Benjamin Recht.
\newblock Exact matrix completion via convex optimization.
\newblock {\em Foundations of Computational mathematics}, 9(6):717--772, 2009.

\bibitem[CW09]{clarkson2009numerical}
Kenneth~L Clarkson and David~P Woodruff.
\newblock Numerical linear algebra in the streaming model.
\newblock In {\em Proceedings of the forty-first annual ACM symposium on Theory
  of computing}, pages 205--214. ACM, 2009.

\bibitem[CW13]{clarkson2013low}
Kenneth~L Clarkson and David~P Woodruff.
\newblock Low rank approximation and regression in input sparsity time.
\newblock In {\em Proceedings of the forty-fifth annual ACM symposium on Theory
  of computing}, pages 81--90. ACM, 2013.

\bibitem[GKX19]{gkx2019investigation}
Behrooz Ghorbani, Shankar Krishnan, and Ying Xiao.
\newblock An investigation into neural net optimization via {H}essian
  eigenvalue density.
\newblock In Kamalika Chaudhuri and Ruslan Salakhutdinov, editors, {\em
  Proceedings of the 36th International Conference on Machine Learning},
  volume~97 of {\em Proceedings of Machine Learning Research}, pages
  2232--2241. PMLR, 09--15 Jun 2019.

\bibitem[GLT18]{gilyen2018quantum}
Andr{\'a}s Gily{\'e}n, Seth Lloyd, and Ewin Tang.
\newblock Quantum-inspired low-rank stochastic regression with logarithmic
  dependence on the dimension.
\newblock {\em arXiv preprint arXiv:1811.04909}, 2018.

\bibitem[GSLW19]{gilyen2019quantum}
Andr{\'a}s Gily{\'e}n, Yuan Su, Guang~Hao Low, and Nathan Wiebe.
\newblock Quantum singular value transformation and beyond: exponential
  improvements for quantum matrix arithmetics.
\newblock In {\em Proceedings of the 51st Annual ACM SIGACT Symposium on Theory
  of Computing}, pages 193--204. ACM, 2019.

\bibitem[GXM{\etalchar{+}}17]{gu2017weighted}
Shuhang Gu, Qi~Xie, Deyu Meng, Wangmeng Zuo, Xiangchu Feng, and Lei Zhang.
\newblock Weighted nuclear norm minimization and its applications to low level
  vision.
\newblock {\em International journal of computer vision}, 121(2):183--208,
  2017.

\bibitem[GZZF14]{gu2014weighted}
Shuhang Gu, Lei Zhang, Wangmeng Zuo, and Xiangchu Feng.
\newblock Weighted nuclear norm minimization with application to image
  denoising.
\newblock In {\em Proceedings of the IEEE conference on computer vision and
  pattern recognition}, pages 2862--2869, 2014.

\bibitem[HZRS16]{he2016deep}
Kaiming He, Xiangyu Zhang, Shaoqing Ren, and Jian Sun.
\newblock Deep residual learning for image recognition.
\newblock In {\em Proceedings of the IEEE conference on computer vision and
  pattern recognition}, pages 770--778, 2016.

\bibitem[IVWW19]{indyk2019sample}
Piotr Indyk, Ali Vakilian, Tal Wagner, and David Woodruff.
\newblock Sample-optimal low-rank approximation of distance matrices.
\newblock {\em arXiv preprint arXiv:1906.00339}, 2019.

\bibitem[KH{\etalchar{+}}09]{krizhevsky2009learning}
Alex Krizhevsky, Geoffrey Hinton, et~al.
\newblock Learning multiple layers of features from tiny images.
\newblock 2009.

\bibitem[KP16]{kerenidis2016quantum}
Iordanis Kerenidis and Anupam Prakash.
\newblock Quantum recommendation systems.
\newblock {\em arXiv preprint arXiv:1603.08675}, 2016.

\bibitem[KV09]{KV09}
Ravi Kannan and Santosh~S. Vempala.
\newblock Spectral algorithms.
\newblock {\em Found. Trends Theor. Comput. Sci.}, 4(3-4):157--288, 2009.

\bibitem[LC15]{lcc15}
Sergey Loyka and Charalambos~D. Charalambous.
\newblock Novel matrix singular value inequalities and their applications to
  uncertain {MIMO} channels.
\newblock {\em {IEEE} Trans. Inf. Theory}, 61(12):6623--6634, 2015.

\bibitem[LNW14a]{li2014sketching}
Yi~Li, Huy~L Nguyen, and David~P Woodruff.
\newblock On sketching matrix norms and the top singular vector.
\newblock In {\em Proceedings of the twenty-fifth annual ACM-SIAM symposium on
  Discrete algorithms}, pages 1562--1581. SIAM, 2014.

\bibitem[LNW14b]{LNW14}
Yi~Li, Huy~L. Nguyen, and David~P. Woodruff.
\newblock Turnstile streaming algorithms might as well be linear sketches.
\newblock In {\em Symposium on Theory of Computing, {STOC} 2014, New York, NY,
  USA, May 31 - June 03, 2014}, pages 174--183, 2014.

\bibitem[LS13]{liesen2013krylov}
J{\"o}rg Liesen and Zdenek Strakos.
\newblock {\em Krylov subspace methods: principles and analysis}.
\newblock Oxford University Press, 2013.

\bibitem[LW16a]{li2016approximating}
Yi~Li and David~P Woodruff.
\newblock On approximating functions of the singular values in a stream.
\newblock In {\em Proceedings of the forty-eighth annual ACM symposium on
  Theory of Computing}, pages 726--739, 2016.

\bibitem[LW16b]{li2016tight}
Yi~Li and David~P Woodruff.
\newblock Tight bounds for sketching the operator norm, {S}chatten norms, and
  subspace embeddings.
\newblock In {\em Approximation, Randomization, and Combinatorial Optimization.
  Algorithms and Techniques (APPROX/RANDOM 2016)}. Schloss
  Dagstuhl-Leibniz-Zentrum fuer Informatik, 2016.

\bibitem[LW17]{li2017embeddings}
Yi~Li and David~P Woodruff.
\newblock Embeddings of {S}chatten norms with applications to data streams.
\newblock In {\em 44th International Colloquium on Automata, Languages, and
  Programming (ICALP 2017)}. Schloss Dagstuhl-Leibniz-Zentrum fuer Informatik,
  2017.

\bibitem[LW20]{lw20}
Yi~Li and David~P. Woodruff.
\newblock Input-sparsity low rank approximation in {S}chatten norm.
\newblock {\em CoRR}, abs/2004.12646, 2020.

\bibitem[Mah90]{maher1990some}
Philip~J Maher.
\newblock Some operator inequalities concerning generalized inverses.
\newblock {\em Illinois Journal of Mathematics}, 34(3):503--514, 1990.

\bibitem[Mah11]{M11}
Michael~W. Mahoney.
\newblock Randomized algorithms for matrices and data.
\newblock {\em Found. Trends Mach. Learn.}, 3(2):123--224, 2011.

\bibitem[MH02]{mason2002chebyshev}
John~C Mason and David~C Handscomb.
\newblock {\em Chebyshev polynomials}.
\newblock CRC press, 2002.

\bibitem[MM13a]{meng2013low}
Xiangrui Meng and Michael~W Mahoney.
\newblock Low-distortion subspace embeddings in input-sparsity time and
  applications to robust linear regression.
\newblock In {\em Proceedings of the forty-fifth annual ACM symposium on Theory
  of computing}, pages 91--100. ACM, 2013.

\bibitem[MM13b]{mm13}
Xiangrui Meng and Michael~W. Mahoney.
\newblock Low-distortion subspace embeddings in input-sparsity time and
  applications to robust linear regression.
\newblock In {\em Symposium on Theory of Computing Conference, STOC'13, Palo
  Alto, CA, USA, June 1-4, 2013}, pages 91--100, 2013.

\bibitem[MM15]{musco2015randomized}
Cameron Musco and Christopher Musco.
\newblock Randomized block {K}rylov methods for stronger and faster approximate
  singular value decomposition.
\newblock In {\em Advances in Neural Information Processing Systems}, pages
  1396--1404, 2015.

\bibitem[MMMW21]{MMMW21}
Raphael~A. Meyer, Cameron Musco, Christopher Musco, and David~P. Woodruff.
\newblock Hutch++: Optimal stochastic trace estimation.
\newblock In {\em 4th Symposium on Simplicity in Algorithms, {SOSA} 2021,
  Virtual Conference, January 11-12, 2021}, pages 142--155, 2021.

\bibitem[Mut05]{M05}
S.~Muthukrishnan.
\newblock Data streams: Algorithms and applications.
\newblock {\em Found. Trends Theor. Comput. Sci.}, 1(2), 2005.

\bibitem[MW17a]{musco2017input}
Cameron Musco and David Woodruff.
\newblock Is input sparsity time possible for kernel low-rank approximation?
\newblock {\em Advances in Neural Information Processing Systems},
  30:4435--4445, 2017.

\bibitem[MW17b]{mw17}
Cameron Musco and David~P. Woodruff.
\newblock Sublinear time low-rank approximation of positive semidefinite
  matrices.
\newblock In {\em 58th {IEEE} Annual Symposium on Foundations of Computer
  Science, {FOCS} 2017, Berkeley, CA, USA, October 15-17, 2017}, pages
  672--683, 2017.

\bibitem[MW21]{mahankali2021optimal}
Arvind~V Mahankali and David~P Woodruff.
\newblock Optimal {L}1 column subset selection and a fast {PTAS} for low rank
  approximation.
\newblock In {\em Proceedings of the 2021 ACM-SIAM Symposium on Discrete
  Algorithms (SODA)}, pages 560--578. SIAM, 2021.

\bibitem[Nel11]{nelson2011sketching}
Jelani Jelani~Osei Nelson.
\newblock {\em Sketching and streaming high-dimensional vectors}.
\newblock PhD thesis, Massachusetts Institute of Technology, 2011.

\bibitem[NN13]{NN13}
Jelani Nelson and Huy~L. Nguyen.
\newblock {OSNAP:} faster numerical linear algebra algorithms via sparser
  subspace embeddings.
\newblock In {\em 54th Annual {IEEE} Symposium on Foundations of Computer
  Science, {FOCS} 2013, 26-29 October, 2013, Berkeley, CA, {USA}}, pages
  117--126, 2013.

\bibitem[Pea94]{pearlmutter1994hv_trick}
Barak~A. Pearlmutter.
\newblock Fast exact multiplication by the {H}essian.
\newblock {\em Neural Computation}, 6:147--160, 1994.

\bibitem[PPZ{\etalchar{+}}20]{peebles2020hessian}
William Peebles, John Peebles, Jun-Yan Zhu, Alexei~A. Efros, and Antonio
  Torralba.
\newblock The {H}essian penalty: A weak prior for unsupervised disentanglement.
\newblock In {\em Proceedings of European Conference on Computer Vision
  (ECCV)}, 2020.

\bibitem[RFP10]{recht2010guaranteed}
Benjamin Recht, Maryam Fazel, and Pablo~A Parrilo.
\newblock Guaranteed minimum-rank solutions of linear matrix equations via
  nuclear norm minimization.
\newblock {\em SIAM review}, 52(3):471--501, 2010.

\bibitem[Riv20]{rivlin2020chebyshev}
Theodore~J Rivlin.
\newblock {\em Chebyshev polynomials}.
\newblock Courier Dover Publications, 2020.

\bibitem[RS00]{roweis2000nonlinear}
Sam~T Roweis and Lawrence~K Saul.
\newblock Nonlinear dimensionality reduction by locally linear embedding.
\newblock {\em science}, 290(5500):2323--2326, 2000.

\bibitem[RSML18]{rebentrost2018quantum}
Patrick Rebentrost, Adrian Steffens, Iman Marvian, and Seth Lloyd.
\newblock Quantum singular-value decomposition of nonsparse low-rank matrices.
\newblock {\em Physical review A}, 97(1):012327, 2018.

\bibitem[RSW16]{razenshteyn2016weighted}
Ilya Razenshteyn, Zhao Song, and David~P Woodruff.
\newblock Weighted low rank approximations with provable guarantees.
\newblock In {\em Proceedings of the forty-eighth annual ACM symposium on
  Theory of Computing}, pages 250--263, 2016.

\bibitem[RWYZ21]{woodruff21}
Cyrus Rashtchian, David~P. Woodruff, Peng Ye, and Hanlin Zhu.
\newblock Average-case communication complexity of statistical problems, 2021.

\bibitem[RWZ20]{RWZ20}
Cyrus Rashtchian, David~P. Woodruff, and Hanlin Zhu.
\newblock Vector-matrix-vector queries for solving linear algebra, statistics,
  and graph problems.
\newblock In {\em Approximation, Randomization, and Combinatorial Optimization.
  Algorithms and Techniques, {APPROX/RANDOM} 2020, August 17-19, 2020, Virtual
  Conference}, pages 26:1--26:20, 2020.

\bibitem[Saa81]{saad1981krylov}
Yousef Saad.
\newblock Krylov subspace methods for solving large unsymmetric linear systems.
\newblock {\em Mathematics of computation}, 37(155):105--126, 1981.

\bibitem[SAR18]{SAR18}
Max Simchowitz, Ahmed~El Alaoui, and Benjamin Recht.
\newblock Tight query complexity lower bounds for {PCA} via finite sample
  deformed wigner law.
\newblock In {\em Proceedings of the 50th Annual {ACM} {SIGACT} Symposium on
  Theory of Computing, {STOC} 2018, Los Angeles, CA, USA, June 25-29, 2018},
  pages 1249--1259, 2018.

\bibitem[Sch60]{schatten1960norm}
Robert Schatten.
\newblock Norm ideals of completely continuous operators.
\newblock 1960.

\bibitem[SJ03]{srebro2003weighted}
Nathan Srebro and Tommi Jaakkola.
\newblock Weighted low-rank approximations.
\newblock In {\em Proceedings of the 20th International Conference on Machine
  Learning (ICML-03)}, pages 720--727, 2003.

\bibitem[SW19]{sw19}
Xiaofei Shi and David~P. Woodruff.
\newblock Sublinear time numerical linear algebra for structured matrices.
\newblock In {\em The Thirty-Third {AAAI} Conference on Artificial
  Intelligence, {AAAI} 2019, The Thirty-First Innovative Applications of
  Artificial Intelligence Conference, {IAAI} 2019, The Ninth {AAAI} Symposium
  on Educational Advances in Artificial Intelligence, {EAAI} 2019, Honolulu,
  Hawaii, USA, January 27 - February 1, 2019.}, pages 4918--4925, 2019.

\bibitem[SWYZ19]{SunWYZ19}
Xiaoming Sun, David~P. Woodruff, Guang Yang, and Jialin Zhang.
\newblock Querying a matrix through matrix-vector products.
\newblock In {\em 46th International Colloquium on Automata, Languages, and
  Programming, {ICALP} 2019, July 9-12, 2019, Patras, Greece}, pages
  94:1--94:16, 2019.

\bibitem[SWZ17]{song2017low}
Zhao Song, David~P Woodruff, and Peilin Zhong.
\newblock Low rank approximation with entrywise l1-norm error.
\newblock In {\em Proceedings of the 49th Annual ACM SIGACT Symposium on Theory
  of Computing}, pages 688--701, 2017.

\bibitem[SWZ20]{song2020average}
Zhao Song, David~P Woodruff, and Peilin Zhong.
\newblock Average case column subset selection for entrywise l1-norm loss.
\newblock {\em arXiv preprint arXiv:2004.07986}, 2020.

\bibitem[Tan19]{tang2019quantum}
Ewin Tang.
\newblock A quantum-inspired classical algorithm for recommendation systems.
\newblock In {\em Proceedings of the 51st Annual ACM SIGACT Symposium on Theory
  of Computing}, pages 217--228. ACM, 2019.

\bibitem[Tao12]{tao2012topics}
Terence Tao.
\newblock {\em Topics in random matrix theory}, volume 132.
\newblock American Mathematical Soc., 2012.

\bibitem[Tao20]{tao2020notes}
Terence Tao.
\newblock Notes 3a: Eigenvalues and sums of hermitian matrices, 2020.

\bibitem[TDSL00]{tenenbaum2000global}
Joshua~B Tenenbaum, Vin De~Silva, and John~C Langford.
\newblock A global geometric framework for nonlinear dimensionality reduction.
\newblock {\em science}, 290(5500):2319--2323, 2000.

\bibitem[Tso08]{tsourakakis2008fast}
Charalampos~E Tsourakakis.
\newblock Fast counting of triangles in large real networks without counting:
  Algorithms and laws.
\newblock In {\em 2008 Eighth IEEE International Conference on Data Mining},
  pages 608--617. IEEE, 2008.

\bibitem[VN37]{von1937some}
John Von~Neumann.
\newblock {\em Some matrix-inequalities and metrization of matric space}.
\newblock 1937.

\bibitem[Woo14]{woodruff2014sketching}
David~P. Woodruff.
\newblock Sketching as a tool for numerical linear algebra.
\newblock {\em Foundations and Trends{\textregistered} in Theoretical Computer
  Science}, 10(1--2):1--157, 2014.

\bibitem[WWZ14]{WWZ14}
Karl Wimmer, Yi~Wu, and Peng Zhang.
\newblock Optimal query complexity for estimating the trace of a matrix.
\newblock In {\em Automata, Languages, and Programming - 41st International
  Colloquium, {ICALP} 2014, Copenhagen, Denmark, July 8-11, 2014, Proceedings,
  Part {I}}, pages 1051--1062, 2014.

\bibitem[XCS10]{xu2010robust}
Huan Xu, Constantine Caramanis, and Sujay Sanghavi.
\newblock Robust {PCA} via outlier pursuit.
\newblock {\em arXiv preprint arXiv:1010.4237}, 2010.

\bibitem[YGKM20]{ygkm2020pyhessian}
Zhewei Yao, Amir Gholami, Kurt Keutzer, and Michael Mahoney.
\newblock {PyHessian}: Neural networks through the lens of the {H}essian, 2020.

\bibitem[YPCC16]{yi2016fast}
Xinyang Yi, Dohyung Park, Yudong Chen, and Constantine Caramanis.
\newblock Fast algorithms for robust {PCA} via gradient descent.
\newblock In {\em Advances in neural information processing systems}, pages
  4152--4160, 2016.

\bibitem[YZ16]{yuan2016tensor}
Ming Yuan and Cun-Hui Zhang.
\newblock On tensor completion via nuclear norm minimization.
\newblock {\em Foundations of Computational Mathematics}, 16(4):1031--1068,
  2016.

\end{thebibliography}

\appendix

\section{Extending Prior Work on Lower Bounds}
\label{sec:appendix-sar}

In this section, we briefly discuss prior work on estimating top singular/eigenvalues in the matrix-vector product model and why existing approaches do not immediately imply a lower bound for low-rank approximation, under any unitarily invariant norm, including Frobenius and spectral norm. 

In a sequence of works,
Braverman, Hazan, Simchowitz and Woodworth~\cite{braverman2020gradient} and Simchowitz, Alaoui and Recht~\cite{SAR18} establish eigenvalue estimation lower bounds in the matrix-vector query model. We draw on their techniques and use the hard instance at the heart of their lower bound, but require additional techniques to obtain a lower bound for low-rank approximation.

The main theorem  (Theorem 2.2 of \cite{SAR18}), for k =1,  states that any randomized algorithm which outputs a vector $v$ such that with constant probability 
\begin{equation*}
    v^\top \abs{\AA} v >= (1- \bigO{ \textrm{gap} } ) \norm{\AA}_{\textrm{op}},
\end{equation*}
requires $\Omega\Paren{1/\sqrt{\textrm{gap}}}$ matrix-vector products, where $\abs{\AA} = (\AA^2)^{1/2}$ has the same singular values as $\AA$ and $\textrm{gap} \in (0,1)$. However, this guarantee is too weak to imply a lower bound for spectral low-rank approximation. 

Indeed, for this theorem to be meaningful in our setting, we require setting $\textrm{gap} = \Theta(\epsilon)$. However, there exist input matrices $\AA$, e.g., $\AA = \textrm{diag}\Paren{1+\epsilon, 1, \ldots, 1, 0}$, and vector $v = \Theta\Paren{\sqrt{\epsilon}} e_1 + \Paren{(1-\Theta(\epsilon)} e_n$ such that
\begin{equation*}
    \norm{\AA(\II-v v^\top)}_{\textrm{op}} \leq \Paren{1+\epsilon} \sigma_2(\AA),
\end{equation*}
i.e. $v$ yields a valid low-rank approximation but $v^\top \AA v$ is only $\Theta(\epsilon)$. Note, here the gap is $\Theta(1)$ instead of the required $1-\epsilon$ and thus we obtain no lower bound for spectral low-rank approximation.

Moreover, it can be shown that when $\AA$ is the hard instance considered in \cite{SAR18}, i.e. $\AA= \GG + \lambda u u^T$, where $\GG$ is a Gaussian Orthogonal Ensemble (GOE) and $u$ is a random unit vector on the sphere, there exists a vector $v$ that does not satisfy the guarantee of Theorem 2.2, yet yields a spectral low-rank approximation. In particular, consider $v = \Theta(\sqrt{\epsilon}) r_1 + \Paren{1-\Theta(\epsilon)} r_d$ where $r_1$ is the largest singular vector of $\abs{\AA}$ and $r_d$ is the smallest singular vector. Since the smallest $O(1)$ singular values of a $d \times d$ GOE can be shown to be $O(1/d)$, and $\AA$ is a rank-$1$ perturbation of a GOE, similar to the diagonal case above, we can show
\begin{equation*}
    \norm{ \AA \Paren{\II-vv^T} }_{\textrm{op}} \leq \Paren{1+ \epsilon} \sigma_2(\AA),
\end{equation*}
yet $v^\top \abs{\AA} v$ is only $\Theta(\epsilon)$. Therefore, it is not possible to obtain a lower bound for low-rank approximation from Theorem 2.2 in a black-box manner.

\section{Low Rank Approximation of Matrix Polynomials} \label{sec:polynomial}

We note that polynomials of matrices are implicitly defined, even in the RAM model, and computing them explicitly would be prohibitively expensive and may destroy any sparsity structure. The proof just follows from running our algorithm on $\MM =\Paren{\AA^\top\AA}^\ell$. It is straightforward to simulate a matrix-vector product of the form $\MM v$ using access to matrix-vector products for $\AA$ and $\AA^\top$ with an $\bigO{\ell}$ overhead.

\begin{theorem}[Low Rank Approximation of Matrix Polynomials]
Given an $n \times d$ matrix $\AA$, $\ell \in \mathbb{N} $, target rank $k$ and an accuracy parameter $\eps>0$, let $\MM = \Paren{\AA^\top \AA}^{\ell}$ or $\MM = \AA \Paren{\AA^\top \AA}^{\ell}$. Then, for any $p\geq 1$, there exists an algorithm that uses at most  $\bigO{k\ell\log(nk)p^{1/6}/\eps^{1/3}}$ matrix-vector products and with probability at least $9/10$ outputs a matrix $\ZZ \in \mathbb{R}^{d\times k}$ with orthonormal columns such that,
\begin{equation*}
    \norm{ \MM  \Paren{\II - \ZZ\ZZ^\top}  }_{\calS_p} \leq (1+\eps)\min_{\UU : \hspace{0.05in} \UU^\top \UU = \II_k} \norm{ \MM  \Paren{\II - \UU \UU^\top}}_{\calS_p}.
\end{equation*}
\end{theorem}
The only prior work we are aware of is the algorithm of \cite{musco2015randomized}, which would achieve a worse $\bigO{k\ell \log(nk)/\eps^{1/2}}$ number of matrix-vector products for the Frobenius norm and match our guarantee for the spectral norm. 

\section{Improved Streaming Bounds}\label{sec:stream}

In the streaming model, the input matrix is initialized to all zeros, and at each time step, the $(i,j)$-th entry is updated. The updates can be positive or negative, and the goal is to output a low-rank approximation, without storing the whole matrix.  
The number of passes required by our algorithm is proportional to the number of \textit{adaptive} matrix-vector queries we require. As an immediate corollary of this observation, we obtain the following formal guarantee:

\begin{corollary}[Schatten LRA in a Stream]
\label{cor:schatten-lra-stream}
Given a matrix $\AA \in \mathbb{R}^{n \times d}$, a target rank $k \in [d]$, an accuracy parameter $\epsilon \in (0,1)$ and any $p\geq 1$, there exists a streaming algorithm that makes $\bigO{\log(d/\epsilon)p^{1/6}/\epsilon^{1/3}}$ passes over the input, requires $\bigO{n k/\epsilon^{1/3}}$ space, and outputs a $d\times k$ matrix $\ZZ$ with orthonormal columns such that with probability at least $9/10$,
\begin{equation*}
    \norm{ \AA \Paren{\II - \ZZ \ZZ^\top}  }^p_{\calS_p} \leq (1+\epsilon) \min_{\UU : \hspace{0.05in} \UU^\top \UU = \II_k} \norm{ \AA \Paren{\II - \UU \UU^\top} \ }^p_{\calS_p}.
\end{equation*}
\end{corollary}

The only prior work on low-rank approximation in a stream is by Boutsidis, Woodruff and Zhong, who consider the special case of $p=2$~\cite{boutsidis2016optimal}. They obtain a single pass algorithm that requires $\bigO{nk/\epsilon +\textrm{poly}(k/\epsilon) }$ space and a two pass algorithm  that requires $\bigO{nk +\textrm{poly}(k/\epsilon) }$ space. For general $p$, we note that recent work by Li and Woodruff~\cite{lw20} can be used to derive a streaming algorithm that obtains a worse space dependence but only requires a single pass: for $1\leq p<2$, the space required is $\tilde{\mathcal{O}}\Paren{n \Paren{ \frac{k+ k^{2/p} }{\epsilon^2} +  \frac{ k^{2/p} }{\epsilon^{1+2/p} }  }  }$ and for $p > 2$, the space required is $\tilde{\mathcal{O}}\Paren{ n\Paren{\frac{k n^{1-2/p} }{\epsilon^2} + \frac{k^{2/p} + n^{1-2/p}}{\epsilon^{2+2/p}  } } }$.

We note that for $p<2$, we obtain a polynomially better dependence on $\epsilon$ and for $p>2$, the space complexity of our algorithm is linear in $n$, as compared to $n^{2-2/p}$ above. 
The optimal space complexity of Schatten-$p$ low-rank approximation (for $p \neq 2$) in a single pass remains open.

\end{document}